\documentclass[11pt,a4paper]{article}
\usepackage[margin=1.206in]{geometry}
\usepackage[latin1]{inputenc}
\usepackage[boxed]{algorithm2e}
\usepackage{amsmath}
\usepackage{amsthm}
\usepackage{float}
\usepackage{bbm}
\usepackage{mathtools}
\usepackage{amssymb,url}
\usepackage{amsfonts}

\usepackage[scaled=.98,sups,osf]{XCharter}% osf in text, lining figures in math
\usepackage[scaled=1.04,varqu,varl]{inconsolata}% inconsolata typewriter
\usepackage[type1]{cabin}% sans serif
\usepackage[charter,vvarbb,scaled=1.05]{newtxmath}

\usepackage{marginnote}
\usepackage{amscd}%
\usepackage{graphicx}
\usepackage{dsfont}
\usepackage{multirow}
\usepackage{enumerate}
\usepackage{scrlayer-scrpage}
\usepackage{bm}
\usepackage{caption}
\usepackage{subcaption}
\usepackage[dvipsnames]{xcolor} %[dvipsnames]
\usepackage{url}
\usepackage[citecolor=NavyBlue,colorlinks=true,linkcolor=NavyBlue,urlcolor=NavyBlue,breaklinks]{hyperref}

\usepackage[compact,raggedright,small]{titlesec}
\usepackage{enumitem}
\usepackage{cleveref}

\usepackage{tikz} 
\usepackage{pgfplots}
\pgfplotsset{compat=newest}
\usetikzlibrary{plotmarks}
\usetikzlibrary{arrows.meta}
\usepgfplotslibrary{patchplots}

\pgfplotsset{plot coordinates/math parser=false}
\newlength\figureheight
\newlength\figurewidth

\usepackage{lineno}
\modulolinenumbers[5]
\DeclareMathOperator{\spn}{span}

\DeclareMathOperator*{\Ran}{Ran}

\DeclareMathOperator*{\argmin}{arg\,min}

\DeclareMathOperator{\rank}{rank}

\DeclareMathOperator{\R}{\mathbb{R}}

\DeclareMathOperator{\E}{\mathbb{E}}

\DeclareMathOperator{\tr}{\operatorname{tr}}

\DeclareMathOperator{\RE}{\operatorname{Re}}

\newcommand{\norm}[1]{\left\lVert #1\right\rVert}
\newcommand\floor[1]{\lfloor#1\rfloor}
\newcommand\ceil[1]{\lceil#1\rceil}
\newcommand\bb[1]{\mathbb{#1}} % mathbb
\newcommand\q[1]{\mathcal{#1}} % math cursive
\newcommand\Ex{\mathbb{E}}

\newtheorem{theorem}{Theorem}
\newtheorem{lemma}[theorem]{Lemma}
\newtheorem{proposition}[theorem]{Proposition}
\newtheorem{corollary}[theorem]{Corollary}

\newtheorem{definition}[theorem]{Definition}
\crefformat{equation}{\textup{#2(#1)#3}}
\crefrangeformat{equation}{\textup{#3(#1)#4--#5(#2)#6}}
\crefmultiformat{equation}{\textup{#2(#1)#3}}{ and \textup{#2(#1)#3}}
{, \textup{#2(#1)#3}}{, and \textup{#2(#1)#3}}
\crefrangemultiformat{equation}{\textup{#3(#1)#4--#5(#2)#6}}%
{ and \textup{#3(#1)#4--#5(#2)#6}}{, \textup{#3(#1)#4--#5(#2)#6}}{, and \textup{#3(#1)#4--#5(#2)#6}}

\Crefformat{equation}{#2Equation~\textup{(#1)}#3}
\Crefrangeformat{equation}{Equations~\textup{#3(#1)#4--#5(#2)#6}}
\Crefmultiformat{equation}{Equations~\textup{#2(#1)#3}}{ and \textup{#2(#1)#3}}
{, \textup{#2(#1)#3}}{, and \textup{#2(#1)#3}}
\Crefrangemultiformat{equation}{Equations~\textup{#3(#1)#4--#5(#2)#6}}%
{ and \textup{#3(#1)#4--#5(#2)#6}}{, \textup{#3(#1)#4--#5(#2)#6}}{, and \textup{#3(#1)#4--#5(#2)#6}}

\makeatletter
\newcommand{\footremember}[2]{%
\footnote{#2}
\newcounter{#1}
\setcounter{#1}{\value{footnote}}%
}
\newcommand{\footrecall}[1]{%
\footnotemark[\value{#1}]%
}
\makeatother
\title{\large\bfseries On the Robustness of Noise-Blind Low-Rank Recovery from %with 
Rank-One Measurements}
\date{\today}
\author{Felix Krahmer\footremember{TUM}{Department of Mathematics, Technical University of Munich, 85748 Garching bei M\"unchen, Germany (\href{mailto:felix.krahmer@tum.de}{felix.krahmer@tum.de})} \and Christian K\"ummerle\footremember{JHU}{Department of Applied Mathematics \& Statistics, Johns Hopkins University, Baltimore, MD 21218, USA (\href{mailto:c.kuemmerle@tum.de}{kuemmerle@jhu.edu})}  \and Oleh Melnyk\footrecall{TUM} \footremember{HMGU}{Mathematical Imaging and Data Analysis, ICT, Helmholtz Center Munich, 85764 Neuherberg, Germany (\href{mailto:oleh.melnyk@tum.de}{oleh.melnyk@tum.de})}
}

\begin{document}

\maketitle
\vspace{-3mm}

\begin{abstract}
We prove new results about the robustness of well-known convex noise-blind optimization formulations for the reconstruction of low-rank matrices from underdetermined linear measurements. Our results are applicable for symmetric \emph{rank-one} measurements as used in a formulation of the \emph{phase retrieval} problem.

We obtain these results by establishing that with high probability rank-one measurement operators defined by i.i.d. Gaussian vectors exhibit the so-called Schatten-$1$ \emph{quotient property}, which corresponds to a lower bound for the inradius of their image of the nuclear norm (Schatten-$1$) unit ball.

We complement our analysis by numerical experiments comparing the solutions of noise-blind and noise-aware formulations. These experiments confirm that noise-blind optimization methods
exhibit comparable robustness to noise-aware formulations.


\end{abstract}
\vspace{4mm}

%Highlights:
%\begin{enumerate}
%\item With high probability the Gaussian rank-one measurement operator exhibits the nuclear norm quotient property.  
%\item Low-rank matrices can be robustly recovered from Gaussian rank-one measurements via equality constrained nuclear norm minimization even if they are not positive semi-definite.
%\item Numerical simulations confirm that low-rank matrix recovery from Gaussian rank-one measurements via equality-constrained nuclear norm minimization and its inequality-constrained variant are comparably robust.
%\end{enumerate}
%\tableofcontents

\section{Introduction}\label{sec: introduction}

\subsection{Motivation and Literature Overview} \label{sec:intro:mot}
Initiated by the seminal works introducing the idea of compressive sensing \cite{CRT06a,CRT06b,Donoho:2006}, the problem of recovering structured data from an underdetermined system random linear measurements has been subject of intensive study in mathematics, signal processing, and computer science in recent years. 

Two classes of structural models that have proven particularly useful are variants of sparsity, where the signal is assumed to have only few significant coefficients in an appropriate basis, and low-rank models, where a matrix-valued signal is assumed to be well-approximated by a matrix of rank much less than the dimension. In both cases, the measurements are of the form
\begin{equation}\label{eq: gen measurements}
\q A X + w= b ,
\end{equation}
where $\q A$ is a (random) linear operator, which maps the unknown structured object $X$ to the observed measurements $b$ and $w$ is a noise. One typically assumes that the number of measurement is much smaller than the signal dimension, which is why such measurement scenarios (especially in combination with a sparsity assumption) are commonly referred to as {\em compressive sensing}. 

%In case of the sparse recovery, $X$ is a so-called sparse vector with only few non-zero coefficients in a given basis, and for low-rank matrix recovery, $X$ is a matrix of rank considerably smaller than the dimension. 
It is not difficult to show that for both the sparsity and the low-rank model, the object that is extremal with regards to the underlying structure, i.e., has fewest non-zero coefficients or smallest rank, respectively, is indeed the correct solution in many cases if there is no additive noise.

That is, in case of sparse recovery with $w=0$ one can often recover $X$ by solving
\[
\min_Z \norm{Z}_{\ell_0}  \text{ such that } \q{A} Z =  b
\]
with $\norm{Z}_{\ell_0}$ denoting the number of non-zero entries of $Z$, and for noiseless low-rank matrix recovery $X$ can typically be recovered via the problem
\[
\min_Z \operatorname{rank}Z  \text{ such that } \q{A} Z =  b.
\]

%,  \eqref{eq: gen measurements} is a linear system of equations such that  and $\q A$ a linear map. Then, $X$ can be obtained by solving
%\[
%\argmin_Z \norm{Z}_{\ell_0}  \text{ such that } \q{A} Z =  b
%\]
%with $\norm{Z}_{\ell_0}$ denoting the number of non-zero entries of $Z$, which is 
Both these problems, however, are non-convex optimization problems and, in general, NP-hard \cite{Natarajan95, Fazel.March2002}.

However, tractable reconstruction of $X$ via appropriate convex relaxations is provably possible \cite{Foucart.2013} if $\q A$ possesses the \emph{null space property} \cite{Cohen09} or the \emph{restricted isometry property} \cite{CT06}. If $X$ is sparse, one can reconstruct it in the noiseless case ($w = 0$) as a minimizer of  
\begin{equation}\label{eq: sparse rec}
\min_Z \norm{Z}_{\ell_1}  \text{ such that } \q{A} Z =  b,
\end{equation}
which can be efficiently solved numerically \cite{DE03}. If $X$ is of low rank, one can find it via the second order cone problem
\begin{equation}\label{eq: lowrank rec}
\min_Z \norm{Z}_{S_1}  \text{ such that } \q{A} Z =  b,
\end{equation}
where the nuclear norm $\|\cdot\|_{S_1}$ of a matrix is defined to be the sum of its singular values.

Subsequent works in the field gave an answer to the questions that followed:
\begin{enumerate}
\item[Q1] Will the reconstruction work if the assumption of  sparsity or low rank is only approximately satisfied \cite{CRT06b, Candes.2011}?  \label{en: 1} 
\item[Q2] Can $X$ be reconstructed robustly in the presence of noise with known strength as constraints in \eqref{eq: sparse rec} or \eqref{eq: lowrank rec} are no longer valid \cite{CRT06b, Candes.2011}? \label{en: 2} 
\item[Q3] What if the strength is unknown or underestimated \cite{Wojtaszczyk10,BA18}? \label{en: 3} 
\item[Q4] Can one blindly apply \eqref{eq: sparse rec} or \eqref{eq: lowrank rec} in the presence of noise and still obtain reasonable reconstruction \cite{Wojtaszczyk10,DeVoreWojtaszczyk09,GKKMR19}? \label{en: 4} 
\end{enumerate} 

The answer the to the first two question is ``yes'', if $\q A$ possesses the \emph{robust null space property} \cite{Foucart.2013}, which is known to be fulfilled for many random matrices with high probability, mostly matrices with independent rows.

The last two questions, however, cannot be answered affirmatively in general if $\mathcal{A}$ fulfills only a suitable \emph{restricted isometry property} \cite{Can08,Candes.2011} or \emph{robust null space property} \cite{Cohen09,Kabanava.2016} with respect to the underlying structure. On the other hand, if in addition $\mathcal{A}$ fulfills the so-called \emph{quotient property} \cite{Wojtaszczyk10,Foucart.2013} (see \Cref{sec:previousresults} below), the last two questions can be answered with ``yes'' in the case of the sparse recovery problem, i.e., if $X$ is a sparse vector. For low-rank matrix recovery answers to question \hyperref[en: 4]{Q3} are only available in limited classes of scenarios and \hyperref[en: 4]{Q4} has hardly been studied in the literature. %\textcolor{cyan}{ADD SOMETHING ABOUT QP FOR NUCLEAR NORM} %\CK{The first property is a modified version of the null space property, which promotes sparsity during the reconstruction process in the presence of noise and commonly is satisfied when the number of measurements is sufficiently high.} 
For sparse recovery, in contrast, these questions are much better understood, in parts also because the quotient property is closely related to well-studied questions in the field of high-dimensional geometry \cite{Gluskin89}.
	%and ensures the existence of a feasible point for noisy measurements in the optimization problem \eqref{eq: sparse rec} within distance proportional to the noise strength. %As the number of measurements increases, the number of constraints in \eqref{eq: sparse rec} increases what naturally leads to the upper bound on the number of measurements when performing a  reconstruction with $\Delta$. \OM{number of measurements talk seems odd in this paragraph}
   
As the null space property, the quotient property is typically established via random constructions, and the %key factors impacting the number of measurements are the structural complexity of $X$
 choice of the random distribution of $\q A$ is critical. In contrast to the null space property, however, these constructions often rely on the independence of the columns rather than the rows. The quotient property adapted for \cref{eq: sparse rec} has been studied in a number of works over the recent years,  
 %complexity of the proofs to show any of the mentioned properties is depends heavily on the distribution, which ranges from
 extending the seminal results for \mbox{Gaussian} matrices \cite{Gluskin89} to matrices with i.i.d. subgaussian entries such as Bernoulli \cite{LPRT05}, Weibull matrices \cite{Fou14}, log-concave distributions \cite{DGT09}, and, more recently, even more heavy-tailed distributions including Cauchy random vectors \cite{GKKMR19}.

When the random matrix has independent entries (and, consequently, both independent rows to yield the quotient property and independent rows to yield the null space property), these results can be combined with results about the null space property to yield noise-blind recovery guarantees, e.g., for Gaussian \cite{Wojtaszczyk10} and subgaussian random matrices \cite{DeVoreWojtaszczyk09}. As such bounds require a logarithmic number of finite moments \cite{ML17}, the resulting guarantees for noise-blind sparse recovery require slightly stronger conditions on the tail decay of the random matrix entries than the quotient property by itself. Similar results about null space and quotient properties hold for random measurement operators based on independent matrices with i.i.d. subgaussian entries and can be used to characterize the performance of \cref{eq: lowrank rec} for low-rank matrix recovery.

For most application scenarios 
of % compressed sensing \cite{Donoho:2006}, which is another well-known name for 
sparse recovery, 
the assumption of measurement operators with independent entries is too restrictive. In more realistic measurement scenarios, one encounters structural constraints imposed by the application. Examples for sparse recovery include Fourier structure \cite{RV08},  combinations of Fourier and wavelet structure \cite{BreakingCoherence,KW14} as motivated, e.g., by magnetic resonance imaging \cite{LDSP08}
 and convolutional structure as motivated, e.g., by channel identification \cite{Hauptetal, RRT12, KMR14}. Such structured scenarios are also preferred for reasons of reduced computational complexity \cite{Rauhut10}, but their theoretical analysis pose considerable challenges as compared to unstructured random ones \cite{Candes.2010b,Recht.2010}.
 
Structural constraints are even more significant in the low-rank matrix recovery as the matrices that constitute the linear operator $\q A$ of \cref{eq: gen measurements} are in many applications not modeled very accurately by random matrices with i.i.d. entries. For example, in recommender systems, it is common to model a user-product matrix as approximately low-rank \cite{Koren.2009}, which is known only through a subset of its entries, often assumed to be revealed at random locations \cite{Candes.2010, Candes.2010b}. This corresponds to a low-rank matrix recovery problem with an operator $\q A$ whose constituent matrices \emph{do not} have i.i.d. entries, but are outer products of random standard basis vectors. Other relevant structured measurement scenarios for low-rank matrix recovery include symmetric rank-one measurements (see also \Cref{sec:intro:rankone} below), as related to phase retrieval \cite{CandesEldarStrohmerVoro.2013}, often in combination with additional structure such as coded diffraction patterns \cite{CLS15} or non-symmetric rank-one and low-rank measurements as encountered in blind deconvolution \cite{ARR13} and blind demixing \cite{JKS18}, and various structural models related to quantum state tomography \cite{Gross.2010,Flammia.2012,Liu.2011}. 

Despite this large variety of relevant structured measurement scenarios, the study of the quotient property with structure is only in its beginnings. To our knowledge, there are two main contributions to mention here: \cite{BA18} has provided such an estimate for sparse recovery with Fourier structure, and \cite{Liu.2011} observed for low-rank matrix recovery from Pauli measurements the quotient property directly follows from orthogonality.

One should mention though that solving \eqref{eq: sparse rec} is not the only way to attempt noise-blind recovery. For measurement systems satisfying the restricted isometry property, recovery guarantees have been established for various greedy and greedy-type methods such as orthogonal matching pursuit \cite{Tropp07} and compressive sampling matching pursuit \cite{CoSaMP}, which do not require a prior knowledge of the noise level, but instead require an upper bound on the sparsity level, which may also not be available in all cases. Furthermore, the square root lasso, a variant of \eqref{eq: sparse rec} given by the minimization problem
\begin{equation}\label{eq: sqrtlasso}
\min_Z \norm{Z}_{\ell_1} + \lambda \|b - \q{A} Z\|_2
\end{equation}
has been shown to yield recovery with a parameter $\lambda$ independent of the noise level under restricted eigenvalue conditions \cite{BCW11}, and also under the robust null space property \cite{PetersenJung20}.

For the low-rank matrix recovery problem, the restricted isometry property has mainly been established for unstructured measurement systems (with the notable exception of Pauli measurements as relevant in quantum information theory \cite{Liu.2011}), so the guarantees for greedy and greedy-type methods have limited applicability. While an approach analogous to the square root lasso also allows for noise-blind recovery when the noise is random \cite{GaiffasKlopp17} -- for example for matrix completion -- no analysis based on the robust null space property analogous to \cite{PetersenJung20} is available yet. 
For the phase retrieval problem, where the solution is known to be positive semidefinite (see below for more details), it has been shown that noise-blind recovery at near-optimal rate can be achieved by completely ignoring the nuclear norm objective of \cref{eq: lowrank rec} and just optimizing the data fidelity over the positive semidefinite cone \cite{DemanetHand,Candes.2014,SlawskiLiHein:2015}. 

%ADD DETAILS... (Add here that greedy algorithms work, but require some estimate of the sparsity and the square root lasso will work, also the consistent reconstruction approach for phase retrieval does not require any estimate)

  Despite the availability of alternative approaches for a number of scenarios, however, we feel that understanding the quotient property of structured random measurement systems is of interest both from the mathematical point of view -- given the fundamental role of algorithms \eqref{eq: sparse rec} and \eqref{eq: lowrank rec} as a benchmark -- and from the viewpoint of applications, given that these approaches form the basis of competitive solution approaches such as Iteratively Reweighted Least Squares \cite{Daubechies10,Fornasier11,Lai13,KS18}, and hence their understanding sheds light on these methods as well.

  %Commonly there is an observed trade-off between the complexity of the measurement operator $\q A$ and the number of measurements required for successful reconstruction. 
 %The interest in more structured random models is mainly driven by the practical experimental setups and computational complexity needs \cite{Rauhut10}. 

%In this paper, we study a low-rank matrix recovery, the ideological counterpart of the sparse recovery. Problems of this type naturally arise in many applications: In recommender systems, it is of interest to find a low-rank matrix compatible with partial information \cite{Koren.2009,Candes.2009,Candes.2010,Candes.2010b}. In quantum state tomography, a positive semidefinite low rank matrix is to be identified from linear measurements \cite{Gross.2010,Flammia.2012}.

\subsection{Rank-One Measurements of Low-Rank Matrices} \label{sec:intro:rankone}

In this article, we study the recovery of low-rank matrices from measurements $\q A$ described by random, symmetric rank one matrices, which constitutes a structured random measurement system that arises in many applications.
%
%ollow the program of studying structured random measurement systems by focusing on the recovery of low-rank matrices from measurements $\q A$ by described by random rank-one matrices. % the case of similar measurement matrices $A_j$ that are of \emph{rank-one}. 

Such measurements arise, for instance, in covariance estimation \cite{Leus.2011,Ariananda.2012,Chen.2015}, where the goal is to recover the covariance matrix $X$ of a random distribution from a quadratic sketches, and in the recovery of images from phaseless measurements \cite{Candes.2011,Candes.2013,Candes.2014,Kueng.2017,Kabanava.2016}, which is relevant for example in Fourier ptychographic microscopy \cite{TianLiRamWal.2014}: If an image is represented by a vector $x \in \bb C^{n}$, phaseless, noisy measurements reflecting only the squared magnitude $|a_{j}^* x|^2$ of the scalar products $a_{j}^* x$  are measured such that
\begin{equation} \label{eq:lowrankrecovery:prob}
b_j = \q{A} (x x^*)_j  + w_j = \langle A_j, x x^* \rangle_F + w_j
\end{equation}
for all $j= 1, \ldots, m$ where the \emph{measurement matrices} $A_j := a_j a_j^*$ are \emph{Hermitian rank-one} since for each $j$,
\begin{equation} \label{eq:phaseretrieval:measurements}
b_j - w_j = |a_j^*x|^2  = a_j^* (x x^*) a_j = \tr(a_j a_j^* x x^*) = \langle a_j a_j^*, x x^* \rangle_F = \langle A_j, x x^* \rangle_F = \langle A_j, X \rangle_F
\end{equation}
with $X = x x^*$ being a Hermitian rank-one matrix \cite{CandesEldarStrohmerVoro.2013}. Beyond that, we also consider the recovery of Hermitian matrices $X$ of small rank $r>1$.% that is small, but higher than one, for which the different coordinates of \cref{eq: gen measurements} are
%\begin{equation} \label{eq:measop:rankone} 
%b_j = \q{A} (X)_j  + w_j = \langle A_j, X \rangle_F + w_j,   \text{ for all } j= 1, \ldots, m
%\end{equation}
%with $A_j = a_j a_j^*$ for all $j = 1,\ldots, m$.

The randomness of such measurements is described by the random distribution of the vectors $a_j$. In this paper, we study the case of independent standard complex Gaussian random vectors $\{a_j: j=1,\ldots,m\}$.%for $j = 1, \ldots, m$.

For this measurement model, we briefly review prior results addressing the questions \hyperref[en: 1]{Q1}-\hyperref[en: 4]{Q4} with respect to the tractable program based on nuclear norm minimization \cref{eq: lowrank rec}: Let $\q H_n$ denote set of Hermitian $(n \times n)$ matrices.
%With $\q H_n$ denoting set of Hermitian $(n \times n)$ matrices, the operator $\q A$ is a map from $\q H_n$ to $\bb R^{m}$ and with this notation, one may recover the matrix of interest in the noiseless case by solving the \emph{rank minimization} problem 
%\begin{equation}\label{eq: rec unrelaxed}
%\min_{Z \in \q H_n} 
%\left\{\rank(Z) \text{ such that } \q{A} (Z) =  b  \right\}. 
%%\quad \text { or } \quad  \min_{Z \succeq 0} \left\{\rank(Z) \text{ such that } \q{A} (Z) =  b  \right\}.
%\end{equation}
%However, these rank optimization problems are known to be NP-hard in general \cite{Fazel.March2002}. On the other hand, a tractable convex optimization formulation that solves the original problem under certain conditions can be achieved by minimizing the \emph{nuclear norm} \cite{Candes.2010b,Recht.2010} such that
%\begin{equation}\label{eq: rec noiseless}
%\Delta(b) := \argmin \left\{ \norm{Z}_{S_1},\ Z \in \q H_n \text{ such that } \q{A} Z =  b  \right\},
%\end{equation}
%where the non-convex rank objective is replaced by the convex function $\norm{Z}_{S_1} = \sum_{i} \sigma_i(Z)$, which sums the singular values of the matrix $Z$. The formulation \cref{eq: rec noiseless} has the advantage that it is equivalent to a semidefinite program and can be solved rather efficiently by a variety of available solvers \cite{MazumderHastieTibshirani.2010,MichaelGrantandStephenBoyd.2020}. 
%At this point one can repeat the questions \hyperref[en: 1]{Q1}-\hyperref[en: 4]{Q4} for low-rank matrix recovery. 
It has been shown that for this model, the solution of \cref{eq: lowrank rec}, i.e., the nuclear norm minimizer 
\begin{equation}\label{eq: rec noiseless}
\Delta(b) := \argmin \left\{ \norm{Z}_{S_1},\ Z \in \q H_n \text{ such that } \q{A} Z =  b  \right\},
\end{equation}
provides good reconstruction even if $X$ is only approximately low-rank \cite{Kueng.2017,Kabanava.2016}, resulting in a positive answer to \hyperref[en: 1]{Q1}.  Similar results \cite{Candes.2013,Chen.2015} have been obtained for the optimization problem $\Delta^{\succeq 0}$ , a variant of $\Delta$ with an additional positive semidefinite constraint on $Z$. The positive semidefinite constraint significantly reduces the search space. However, one
%. As a drawback, we 
no longer %considers
minimizes over a subspace of matrices, but a convex cone.
 
In \cite{Kueng.2017} authors introduce the recovery method
\begin{equation}\label{eq: rec noisy}
\Delta_{q,\eta}(b) := \argmin \left\{ \norm{Z}_{S_1}, Z \in \q H_n \text{ such that } \norm{\q{A} Z - b }_{q} \le \eta \right\},
\end{equation}
a variant of $\Delta$ where its equality constraint is replaced by an inequality constraint on the residual error in terms of $\ell_q$-norm. If $\eta$ is chosen compatibly with the noise level $\|w\|_{\ell_q}$, i.e., if we encounter noise of \emph{known level}, \cite{Kueng.2017} shows that $\Delta_{q,\eta}$ robustly  reconstructs $X$, answering \hyperref[en: 2]{Q2} affirmatively. However, an analysis of the performance of $\Delta$ in the presence of noise, i.e., an answer to \hyperref[en: 4]{Q4} has not been achieved.

% another variant of $\Delta$, see \cref{eq: rec noisy} for more details, which robustly reconstructs $X$ in the presence of noise with known strength and, thus, answers \hyperref[en: 2]{Q2} affirmatively.

%\CK{(INCLUDE THIS IN A SHORT WAY HERE:) As discussed in the introduction, it is quite well-understood when the nuclear norm minimizer $\Delta$ of \cref{eq: rec noiseless} can successfully recover $X$ from noiseless measurements. However, an analysis of the performance of $\Delta$ in the presence of noise, i.e., an answer to \hyperref[en: 4]{Q4} has not been achieved using the same methods as for the noiseless case. As a remedy, an alternation of $\Delta$ was proposed in \cite{Kueng.2017} to accommodate the noisy case. In this method, the equality constraint is replaced by an inequality constraint on the residual error in terms of $\ell_q$-norm, that is   
%
%where the parameter $\eta$ should be chosen as an upper bound on the noise strength $\norm{w}_{q}$. {\color{cyan} In this sense, $\Delta_{q,\eta}$ is designed as a {\em noise aware} reconstruction method.}}

 Analogous results for rank-one and general semidefinite matrices $X$ have been established in \cite{Candes.2013} and \cite{Chen.2015}, respectively, and have been extended to simultaneous stability in terms of both measurement noise and structure violation later developed in \cite{Kabanava.2016}. We review these results in detail in \Cref{sec:previousresults}.  

An alternative approach for general Hermitian low-rank matrices is to minimize an of $S_1$-norm Lasso-style \cite{Tib96} objective 
\begin{equation}\label{eq: rec Lasso}
\Delta^{\text{Lasso}}_{q,\mu}(b) := \argmin_{Z \in \q H_n} \left\{ \norm{\q{A}Z - b }_{\ell_q}^q + \mu \norm{Z}_{S_1} \right\},
\end{equation}
as considered, e.g., in \cite{Candes.2011,Liu.2011}. It is not hard to show that for an appropriate parameter $\mu$, the minimizer will exhibit similar behavior of \cref{eq: rec noisy}.

As observed in \cite{DemanetHand,Candes.2014}, the nuclear norm objective can be omitted provided that positive semidefiniteness is enforced: In absence of any measurement noise, there is typically only one positive semidefinite solution; likewise in noisy scenarios, an accurate reconstruction can be obtained by choosing the positive semidefinite matrix that best fits the data in $\ell_q$-norm such that
\begin{equation} \label{eq:least:ellq}
\Delta_{q}^{\succeq 0}(b): = \argmin_{Z \succeq 0} \norm{\mathcal{A}Z - b}_{\ell_q}.
\end{equation}
If the matrix of interest $X$ is positive semidefinite, arguably a situation of particular interest in many applications, this provides an affirmative answer to \hyperref[en: 3]{Q3}, as confirmed by the recovery guarantees of \cite{Candes.2014} (for the case $q=1$) and \cite{Kabanava.2016}. However, the positive semidefiniteness of $X$ is crucial for $\Delta_{q}^{\succeq 0}$ and its analysis in \cite{Candes.2014, Kabanava.2016} does extend directly beyond that case.
%, it is unclear how this can be extended to the recovery of low-rank matrices in general.

 %develop recovery guarantees for the reconstruction of rank-one matrices via $\Delta_{1}^{\succeq 0}$ and  quantify the specifically the Frobenius ($p=2$) error of the reconstruction. These results have been extended and generalized in \cite{Kabanava.2016} to the recovery of positive semi-definite matrices of low rank via $\Delta_{q}^{\succeq 0}$ 
% However, as the answer does not proceed via \hyperref[en: 4]{Q4}, this approach is not expected to be generalized to the recovery of matrices that are \emph{not} positive (or negative) semidefinite.

%%It turns out that its hyperparameter $\mu \ge 0$ can be chosen \emph{independently} of the noise level; as we detail in \Cref{app:proof:noiseless:lasso}, in this sense, $\Delta^{\text{Lasso}}_{q,\mu}$ provides an affirmative answer to \hyperref[en: 3]{Q3} when $\mu$ is chosen sufficiently small. On the other hand, the noise-dependent error scales with $\q O( \mu^{-1})$ and thus it will, if $\mu$ is too small, blow up, restricting the noise-blind application of the Lasso reconstruction. This leaves potential space for improvement.  
In fact, we are not aware of any approaches to \hyperref[en: 4]{Q4} nor any alternative approaches to \hyperref[en: 3]{Q3} that would cover general Hermitian matrices.
%any answers in the literature addressing \hyperref[en: 3]{Q3} for general Hermitian matrices. %answers to \hyperref[en: 3]{Q3} for general Hermitian matrices are rare in the literature and not complete. 
%However, we did not find any results for the recovery of Hermitian low-rank matrices in the literature. 
While arguably this scenario has not found many applications yet, we still believe that it is an important case for a thorough understanding of low-rank matrix recovery problems. We expect that its analysis can help to pave the way for an understanding of other structured measurement scenarios without positive semidefiniteness such as randomized blind deconvolution \cite{ARR13}.
% \hyperref[en: 4]{Q4} remained unanswered until now as well. 

\subsection{Our Contribution and Outlook}
In this article, we provide answers for questions \hyperref[en: 3]{Q3} and \hyperref[en: 4]{Q4} for random rank-one measurements and general Hermitian matrices $X$, thus including the case of Hermitian matrices that are not positive semidefinite.   

In \Cref{thm: QP}, we show that with high probability the rank-one measurement operator $\mathcal{A}: \q H_n \to \R^m$ defined by independent Gaussian vectors $a_j$ fulfills, under appropriate conditions, the \emph{nuclear norm quotient property} \cite{Candes.2011,Liu.2011}. Our proof technique for \Cref{thm: QP} is entirely different from techniques used to establish the various results for the quotient property in compressed sensing and low-rank matrix recovery mentioned in \Cref{sec:intro:mot}, and might be of independent interest. As our techniques are tailored to a low-rank recovery problem from structured measurements, we expect that they will prove useful to study other structured measurement scenarios such as randomized blind deconvolution.

In analogy to sparse recovery, our findings for rank-one measurements (see \Cref{col: final}) entail recovery guarantees for the noise-blind nuclear norm minimizer $\Delta$ and its inequality-constrained variant $\Delta_{q,\eta}$ with incompatible choice of parameter $\eta$.  These results answer \hyperref[en: 3]{Q3} and \hyperref[en: 4]{Q4} affirmatively.

%{\color{cyan} In analogy to sparse recovery, our findings entail noise-blind recovery guarantees for $\Delta$ and other common choices of the reconstruction methods for rank-one measurements, answering \hyperref[en: 3]{Q3} and \hyperref[en: 4]{Q4} affirmatively, see \Cref{col: final}.   }

%As a %tool to obtain an understanding
%{\color{cyan} key ingredient} of noise-blind recovery and to answer \hyperref[en: 3]{Q3} and \hyperref[en: 4]{Q4} affirmatively, we show 
%\Cref{thm: quotient+NSP to optimality}, which implies that the nuclear norm quotient property can be combined with well-established \emph{rank null space properties} \cite{RechtXuHassibi:2011,Kabanava.2016} to obtain guarantees about the robustness of reconstruction maps which already perform well in the noiseless case or in the case of controlled noise.
%
%By combining \Cref{thm: QP} and \Cref{thm: quotient+NSP to optimality}, our analysis culminates in \Cref{col: final}, which establishes recovery guarantees for $\Delta$ and other common choices of the reconstruction methods for rank-one measurements. 

The rest of the paper is structured as following. The next section establishes basic notation and definitions used throughout the paper. In  \Cref{sec:previousresults}, we discuss other possible choices of reconstruction maps beyond $\Delta$ and results regarding their performance provided in \cite{Kabanava.2016}. 
After presenting our results in \Cref{sec: results}, 
%with \Cref{sec:noiseblind:guarantees} containing general case and \Cref{sec:rankone:results} dedicated to the rank-one random Gaussian measurements. 
we perform numerical experiments in \Cref{sec: numerics} to shed light on the actual difference in the performance of various reconstruction methods for noisy measurements.
%To the best of our knowledge, a numerical comparison of the noise robustness of different reconstruction methods based on convex optimization has not been performed in the literature.
Finally, we detail the proofs of our results in \Cref{sec: proofs}. %the reader can find the proofs of our results.

\subsection{Preliminaries and Notation}\label{sec: preliminaries}
%\CK{Maybe separate the preliminaries and previous results? I have the impression that some preliminaries are quite basic and can maybe introduced, if necessary, whenever needed, e.g. before our theorem or in the proofs.}
%
%\textcolor{red}{
%\begin{itemize}
%\item Discussion about distribution of $a_j$ (maybe in conclusion/\OM{outlook}): How about sub-Gaussian? Asymmetric measurements? \\
%\item issues with Danzig selector $\Delta_{q,\eta}(b)$ and alternative Lasso + result by Kabanava\\
%	\item Pauli measurements and use rank quotient property for Dantzig/Matrix Lasso: \cite{Liu.2011}
%	\item \cite{CaiZhang:2015} (studies asymmetric rank-one measurement and a nuclear norm estimator with two convex constraints. Q: relevant for us? \OM{outlook. Did not found it in the paper. Incorect reference?})
%	\item Psd-constrained least-squares estimator \cite{SlawskiLiHein:2015,Kabanava.2016} or psd-constrained $\ell_1$-estimator \cite{Candes.2014}. \OM{You touch them above and we don't really need any results concerning them.}
%\end{itemize}}

%\subsubsection{Basics of matrix analysis}
We now set up some notation that will be used throughout the paper. We denote the Frobenius (Hilbert-Schmidt) product $\langle \cdot, \cdot \rangle_F$ of two matrices $B, D \in \bb{C}^{n_1 \times n_2}$ by $\langle B, D \rangle_F = \tr( B D^* )$.
%\textcolor{red}{Should we immediately concentrate on the Hermitian matrix space?}
For any matrix  $X \in \bb{C}^{n \times n}$, we use the following notation for its \emph{singular value decomposition (SVD)} 
\[
X = \sum_{j=1}^{ n } \sigma_{j}(X) u_j v_j^* = U \Sigma V^*,
\]
where $U$ and $V$ are $n \times n$ unitary matrices and $\Sigma \in \bb{C}^{n_1 \times n_2}$ is diagonal with entries non-negative entries $\sigma_{j}(X)$, the singular values of $X$ ordered in decreasing order. The rank of matrix $X$ is given by the number of positive  singular values. 

For $1 \le p \le \infty$, the \emph{Schatten-$p$ norm} of the matrix $X \in \bb{C}^{n \times n}$ is the $\ell_p$-norm of its singular values, that is
\[
\norm{X}_{S_p} := \norm{ (\sigma_1(X), \ldots, \sigma_n (X) )}_{\ell_p}.
\]  
%We note that the nuclear norm $\norm{\cdot}_*$ is the Schatten-1 norm $\norm{\cdot}_{S_1}$.
We will always write $\norm{\cdot}_{p} = \norm{\cdot}_{S_p}$ for matrices and $\norm{\cdot}_{p} = \norm{\cdot}_{\ell_p}$ for vectors. 

The \emph{dual norm} $\norm{\cdot}_{*}$ associated with $\norm{\cdot}$ 
%in normed space $(\q X, \norm{\cdot})$ 
is given by $\norm{v}_{*} = \sup_{ \norm{z} = 1} z^* v.$  
For $\ell_p$- and $S_p$-norms, the associated dual norms are $\ell_{p^*}$ and $S_{p^*}$ respectively, where $p^{*}$ is a H{\"o}lder dual of $p$, that is $1/p + 1/p^* = 1$.  We recall that H{\"o}lder's inequality holds for Schatten-$p$ norms \cite[p. 92, 95]{Bhatia.1997} such that
\[
|\langle X, Z \rangle_F| \le \norm{X}_{p} \norm{Z}_{p^*},
\quad \text{ for all } Z,X \in \bb{C}^{n \times n} \text{ and } 1 \le p \le \infty.
\] 
The \emph{best rank-$r$ approximation} $X_r$ of $X$ and its \emph{complement} $X_r^c$ are defined as
\[
X_{r}:= \sum_{j=1}^{r} \sigma_{j}(X) u_j v_j^*,
\ \text{ and }\ 
X_{r}^{c} := X - X_{r} = \sum_{j=r+1}^{n} \sigma_{j}(X) u_j v_j^*,
\] 
and $X_{r}$ is a matrix which minimizes the projection error on the manifold of rank $r$ matrices, that is
\begin{equation}\label{eq: rank proj}
\norm{X_r^c}_1 =  \norm{X - X_r}_1 
=\min_{\substack{ Z \in \bb{C}^{n \times n}, \\ \rank(Z) = r}}  \norm{X - Z}_1.
\end{equation}

The measurement operator $\q A$ is a linear operator mapping Hermitian matrices $\q H_n$ to real vectors $\bb{R}^m$. The \emph{kernel (null space)} and the \emph{range} of $\q A$ are denoted by
\begin{align*}
\ker \q A & := \{ Z \in \q H_n : \q A(Z) = 0\}, \\
\Ran \q A & := \{ v \in \bb{R}^m : \text{ there exists } Z \in \q H_n \text{ such that } \q A(Z) = v  \},
\end{align*}
respectively.

%In the following statements we will denote by $\q A$ the measurement operator with rank-1 Gaussian measurement matrices and by $\q B$ an abstract measurement operator
%Finally, we define a scaled version of the measurement operator $\q A_s := \frac{1}{\sqrt m} A$

%The adjoint of the measurement operator $\q A$ is defined as
%\[
%\left \langle \q{A} Z, v \right \rangle = \left \langle Z, \q{A}^* v \right \rangle_F
%\]
%for all $Z \in \q H_n$ and for all $v \in \bb{R}^{m}$. 

\subsection{Technical foundations} \label{sec:previousresults}
In this section, we review previous results about recovery guarantees for convex recovery methods for low-rank recovery, with a particular focus on those suitable for rank-one measurement operators as defined in \cref{eq:lowrankrecovery:prob}. 
 
An important tool in the performance analysis of recovery methods $\Delta$ and $\Delta_{q,\eta}$ as introduced in \cref{eq: rec noiseless} and \cref{eq: rec noisy} in \Cref{sec:intro:rankone}  is the \emph{robust rank null space property} \cite{Kabanava.2016}, which is similar in its core to the robust null space property that is used to characterize bounds for $\ell_1$-type optimization methods for the sparse vector recovery problem \cite{Foucart.2013}. 
\begin{definition}[\cite{Kabanava.2016}, Definition 3.1] \label{def:rNSP}
For $p \ge 1$, the measurement operator $\q{A}: \q{H}_{n} \to \bb{R}^m$ satisfies the \emph{$S_p$-robust rank null space property with respect to norm $\norm{\cdot}$ on $\bb{R}^m$ of order $r$} with constants $0 < \rho <1$ and $\tau >0$ if for all $X \in \q{H}_{n}$, the inequality
\[
\norm{X_{r}}_{p} \le \frac{\rho} { r^{1 - \frac{1}{p}} } \norm{ X_{r}^{c}}_1 + \tau \norm{ \q{A} X }
\]
holds.
\end{definition}

%\textcolor{green}{Some other references for NSP for matrices?}
%Similarly to the null space property in the sparse recovery, $S_p$-robust rank null space property promotes low-rank solution of the optimization problems $\Delta$ and $\Delta_{q,\eta}$. 
The following theorem shows that the $S_p$-robust rank null space property is sufficient to establish the recovery guarantees for nuclear norm minimization problems such as $\Delta$ from \cref{eq: rec noiseless} and $\Delta_{q,\eta}$ from \cref{eq: rec noisy}, both in the noiseless case and in the presence of noise.

\begin{proposition}[{\cite[version of Theorem 3.1]{Kabanava.2016}}]\label{thm: Kabanava3.1}
For $q \ge 1$ and $1 \le p \le 2$, let the measurement operator $\q{A}: \q{H}_{n} \to \bb{R}^{m}$ satisfy the $S_2$-robust rank null space property with respect to norm $\norm{\cdot}_{q}$ of order $r$ with constants $0 < \rho < 1$ and $\tau > 0$. Then, for any $X \in \q{H}_{n}$ and any noise $w \in \bb{R}^{m}, \norm{w}_{q} \le \eta$, the inequalities
\[
\norm{X - \Delta \left( \q{A} X \right) }_p \le \frac{ 2 (1 + \rho)^2 } { (1 - \rho) r^{1 - 1/p} } \norm{X_{r}^{c} }_1
\]
and
\begin{equation} \label{eq: rec guarantees eta}
\norm{X - \Delta_{q,\eta} \left( \q{A} X + w \right) }_p \le \frac{ 2 (1 + \rho)^2 } { (1 - \rho) r^{1 - 1/p} } \norm{X_{r}^{c} }_1 + \frac{2 \tau (3 + \rho)}{ 1 - \rho } r^{1/p - 1/2} \eta
\end{equation}
hold. 
%The same bounds holds for the case of Hermitian matrix measurement and reconstruction.
\end{proposition}

Similar results involving weaker versions of the null space property were established for the recovery of the real symmetric positive semidefinite matrices \cite{Chen.2015}, recovery of general \cite{Liu.2011} and Hermitian matrices \cite{Flammia.2012} with Pauli measurements  and phase retrieval \cite{Candes.2013}. 

% \OM{COMMENTED OUT}
%{\color{cyan} I DON'T UNDERSTAND - THE HYPERPARAMETER DOES DEPEND ON THE NOISE, DOESN'T IT?}
%The requirement of the estimate $\eta$ is a crucial point for optimization problem $\Delta_{q,\eta}$ and when chosen too small, the results of  \Cref{thm: Kabanava3.1} no longer apply. On the contrary, when $\eta$ is too large, the set of feasible points explodes and solutions of $\Delta_{q,\eta}$ are no longer close to $X$ (see Section \ref{sec: numerics} for details).
%As a way to avoid noise estimation, another optimization problem is considered in \cite{Candes.2011,Liu.2011}, which instead minimizes the residual error and incorporates {\color{cyan} the} $S_1$-norm as a Lasso-style \cite{Tib96} penalty. Namely, it is given by
%\begin{equation}\label{eq: rec Lasso}
%\Delta^{\text{Lasso}}_{q,\mu}(b) := \argmin \left\{ \norm{\q{A} Z - b }_{\ell_q} + \mu \norm{Z}_{S_1}, Z \in \q H_n \right\},
%\end{equation}
%where $\mu \ge 0$ is a hyperparameter to be chosen. As we detail in \Cref{app:proof:noiseless:lasso}, it is possible to derive recovery guarantees for $\Delta^{\text{Lasso}}_{q,\mu}$ that are similar to ones of \Cref{thm: Kabanava3.1}.

Similarly to \Cref{thm: Kabanava3.1},  %are %\emph{not} formulated to hold for specific measurement operators $\mathcal{A}$, but rather hold for those fulfilling the 
 our noise-blind recovery guarantees also rely on the $S_2$-robust rank null space property. 
 %This shows that the $S_2$-robust rank null space property is a crucial property of the measurement operator to quantify the error of recovery methods for low-rank matrices. In view of this, we note that 
 The next theorem analyzes this property for rank-one Gaussian measurements and will hence be a crucial ingredient of our proof.
 %to expect if we want to robustly recover low-rank matrices via any method. 
 %For many other measurements scenarios, similar results can be found in \cite{Kabanava.2016}.

\begin{theorem}[{\cite[Section 6]{Kabanava.2016}}]\label{thm: NSP}
Let $\q A: \q H_n \to \bb R^{m}$ be a measurement operator defined by $m$ rank-one measurements of the form \eqref{eq:lowrankrecovery:prob} such that the entries of the vectors $a_j$ defining the measurement matrices $A_j = a_j a_j^*$ are i.i.d. standard complex Gaussian random variables. If
\[
m \ge c_1 \rho^{-2} r n,
\]
then the scaled measurement operator $\frac{1}{\sqrt m}\q A$ possesses the $S_2$-robust rank null space property with respect to the norm $\norm{\cdot}_{2}$ on $\bb{R}^m$ of order $r$ with constants $0 < \rho < 1$ and $\tau>0$ with probability at least $1 - e^{- \gamma_1 m}$, where $c_1,\gamma_1>0$ are absolute constants. 
\end{theorem}

Another important component of the noise-blind recovery will be the following property.
\begin{definition}[$S_q$-quotient property \cite{Candes.2011,Liu.2011}] \label{def: qp}
Given $q \ge 1$, a measurement operator $\q{A}: \q{H}_{n} \to \bb{R}^{m}$ is said to possesses the \emph{$S_q$-quotient property} with constant $d$ and rank $r_*$ relative to norm $\norm{\cdot}$ on $\bb{R}^{m}$ if for all $w \in \bb{R}^{m}$ there exists $U \in \q{H}_{n}$ such that
\[
\q{A} U = w \text{ and } \norm{U}_{q} \le d r_{*}^{ \frac{1}{q} - \frac{1}{2}} \norm{w}.
\]
\end{definition}
Beyond an algebraic relationship on the elements of the quotient space $S_q / \ker \mathcal{A}$, the $S_q$-quotient property has also a geometric interpretation: It is easy to check that \Cref{def: qp} is equivalent to the measurement operator $\q{A}$ fulfilling
\begin{equation} \label{eq: QP: geometric}
\q{A}(B_{S_q}^{n \times n}) \supset \frac{1}{d} B_{2}^{m}, 
\end{equation}
where $B_{S_q}^{n \times n} = \{X \in \q{H}_n: \| X\|_{S_q} \leq 1\}$ and $B_{2}^{m} = \{y \in \R^m: \| y\|_{2} \leq 1\}$ are the nuclear norm unit ball and $\ell_2$-unit ball, respectively. In this sense, the quotient property means that the an $\ell_2$-ball of radius $1/d$ is contained in the image of the $S_q$-unit ball with respect to the measurement operator. In other words, the quotient property provides a lower bound on the inradius of the image of the $S_q$-unit ball.

The $S_q$-quotient property has been used in the literature on low-rank matrix recovery before: The $S_1$-quotient property coincides with the \emph{Nuclear Norm Quotient property (NNQ)} of \cite{Candes.2011}, which has been used to analyze the performance of the matrix Lasso $\Delta^{\text{Lasso}}_{q,\mu}$ and another related reconstruction map, the Matrix Dantzig selector \cite[(I.3)]{Candes.2011}, in the presence of the Gaussian noise with known variance. In that paper, the $S_1$-quotient property was shown for random measurement operators \cref{eq:lowrankrecovery:prob} that are defined by independent measurement matrices $A_j$ with i.i.d. Gaussian or subgaussian entries (which each are of full rank with high probability) if $m \le c n^2 / \log (m/n)$, where $c$ is a constant, $m$ is the number of measurements and $n$ the dimension of the domain of the measurement operator $\mathcal{A}$. For an analysis of the same algorithms, \cite{Liu.2011} showed the $S_1$-quotient property for Pauli measurements, which are of importance in quantum state tomography. To the best of our knowledge, the results on the $S_1$-quotient property of \cite{Candes.2011,Liu.2011} are the only ones available in the literature.

%\CK{Mention sparse recovery QP here?}
%For the related sparse recovery problem, the key to noise-blind results has been the $\ell_q$-quotient property, as it was observed in \cite{Wojtaszczyk.2010}, see also \cite[Chapter 11]{Foucart.2013} for an extensive discussion. When combined with a null space property and a reasonable reconstruction map based on convex optimization, the $\ell_q$-quotient property implies desired recovery guarantees for noise-blind sparse recovery. 

%In Subsection \ref{sec:noiseblind:guarantees}, we present results that on the one hand, extend the scope of \Cref{thm: Kabanava3.1} and \Cref{col: Lasso noiseless} to a \emph{noise-blind} setting. 

%Furthermore, our results show that equality-constrained nuclear norm minimization $\Delta$, which does not require any estimate of the magnitude of a noise vector $w$ and therefore is \emph{noise-blind}, exhibits similar noise-robustness properties as $\Delta_{q,\eta}$ and $\Delta_{q,\mu}^{\text{Lasso}}$ with properly tuned parameters $\eta$ and $\mu$ under a certain condition.

%Finally, in Subsection \ref{sec:rankone:results}, we show that this condition, the so-called \emph{$S_1$-quotient property}, is actually fulfilled with high probability for measurement operators defined by rank-one measurements derived from Gaussian vectors, which constitutes our main theoretical result. 

\section{Main results}\label{sec: results}
This section presents the main results of this paper. First, we present \Cref{thm: QP}, which shows that the $S_1$-quotient property holds for measurement operators defined from Gaussian rank-one measurements. \Cref{col: final} establishes then that noise-blind recovery is indeed possible from Gaussian rank-one measurements.

For the latter result, it is necessary to combine the $S_1$-quotient property with a $S_2$-robust rank null space property of \Cref{def:rNSP}. In \Cref{thm: quotient+NSP to optimality}, we provide an extensive statement characterizing the interplay of these two properties, noise-aware and noise-blind guarantees for convex reconstruction maps for low-rank matrix recovery.

We recall the measurement model of \cref{eq:lowrankrecovery:prob}: Let $\q{A}: \q H_n \to \bb R^{m}$ be an operator such that
\begin{equation} \label{eq:measop:rankone:mainresults} 
\q{A} (X)_j  = \langle a_j a_j^*, X \rangle_F
\end{equation}
for all $j= 1, \ldots, m$. %, where the $a_j$ are independent, complex standard Gaussian vectors. 
 With this, we state our main theorem about the $S_1$-quotient property.

%\subsection{Results for rank-one measurements} \label{sec:rankone:results}

%With the following theorem, we present the central result of this paper and establish the $S_1$-quotient property relative to the rank $r_*$ for Gaussian rank-ones measurements. Together with \Cref{thm: quotient+NSP to optimality}, this will lead to noise-blind guarantees for several algorithms, as we will see below.

\begin{theorem}[$S_1$-quotient property for Gaussian rank-one measurements] \label{thm: QP}
Let $\q A: \q H_n \to \bb R^{m}$ be a measurement operator defined by $m$ rank-one measurements of the form \cref{eq:measop:rankone:mainresults}. If the entries of the vectors $a_j$ are i.i.d. standard complex Gaussian random variables and if
\[
m \le ( n / c_2 \log(545 m)  )^{8/7}
\ \text{ and } m \text{ is sufficiently large,}
\]
for a constant $c_2>0$, then there exists $\gamma_2 > 0$, such that the scaled measurement operator $\frac{1}{\sqrt m}\q A$ possesses the $S_1$-quotient property with constant $d = \frac{128 \sqrt 2}{\kappa} \sqrt{\kappa m/n}$ and rank $\kappa m/n$ relative to the norm $\norm{\cdot}_{2}$ on $\bb{R}^m$ for all $\kappa >0$ with probability at least $ 1 - 14 e^{ -\gamma_2 m^{7/8} }$.
\end{theorem}

Compared to the related result involving dense, full-rank Gaussian measurement operators \cite{Candes.2011}, \Cref{thm: QP} contains a dimension dependence of the factor $d$ as it scales with $\sqrt{m/n}$.  However, when used in the analysis of noise-blind reconstruction maps, this will only lead to a slightly worse order in front of the noise term $\norm{w}_2$ in the theoretical guarantees, as can be seen below in \Cref{thm: quotient+NSP to optimality}.

%%%%%%%%%%%%%%%%%%%%%%%%%%%%%%%%%%%%%%%%%%%%%%%%%%%%%%%%%%%%%%%%%%%%%%%%%%%%%%%%%%%%%%%%%%%%%%%%%%%%%%%%%%%%%%%%%%%
% We wanted to remove this, didn't we? (C.K. Oct 13)
%\begin{remark}
%The upper bound on the number of measurements $m$ might seem as a restrictive requirement to the reader, since in \cite{Candes.2011}, which studies far less structured measurement operators constituted from unstructured measurement with i.i.d. entries, the requirement on $m$ is such that $m \le c n^2/\log(m/n)$, which is less restrictive (and near-optimal). However, the requirement of \Cref{thm: QP} is essential in the case of rank-one measurements as for growing $m$, it becomes harder to satisfy the existence condition of the quotient property. We consider this as an artifact due to the structural complexity of the analysis, which is illustrated by the observation that measurement operators $\mathcal{A}: \mathcal{H}_n \to \R^m$ defined by matrices with i.i.d. Gaussian entries (such as those used in \cite{Candes.2011}) have $m n^2$ independent random variables, while those defined by Gaussian rank-one measurements have only $m n$ independent variables. 
%
%We note that the actual lower bound on $m$ that results from our proof of \Cref{thm: QP} is extreme and impractical, but \Cref{thm: QP} still sheds light on the relationship between the problem dimensions $n,m$ and $r_*$ in terms of orders.
%\end{remark}
%
%%%%%%%%%%%%%%%%%%%%%%%%%%%%%%%%%%%%%%%%%%%%%%%%%%%%%%%%%%%%%%%%%%%%%%%%%%%%%%%%%%%%%%%%%%%%%%%%%%%%%%%%%%%%%%%%%%%

With the $S_1$-quotient property being established, we can go ahead and extend the scope of \Cref{thm: Kabanava3.1} to a \emph{noise-blind} setting. With the following theorem, we show that equality-constrained nuclear norm minimization $\Delta$, which does not require any estimate of the magnitude of a noise vector $w$ and therefore is noise-blind, exhibits similar noise-robustness properties as $\Delta_{q,\eta}$ with properly tuned parameters $\eta$ for the recovery of low-rank matrices.

\begin{theorem}[Noise-blind recovery guarantees for Gaussian rank-one measurements]\label{col: final}
Let $1 \le p \le 2$ and $0 < \rho < 1$. Let $\q A: \q H_n \to \bb R^{m}$ be a measurement operator defined by $m$ rank-one measurements of the form \eqref{eq:measop:rankone:mainresults} such that the entries of vectors $a_j$ are i.i.d. standard complex Gaussian random variables.
If 
\[
m \le \left( \frac{n}{c_2 \log(545 m)}  \right)^{8/7},
\ m \text{ sufficiently large}
\text{ and }
r \le \frac{ \rho^2 m}{c_1 n} = r_*
\]
for absolute constants $c_1,c_2>0$, then there exists constant $\gamma$ such that with probability at least $1 - 15 \exp \left\{ -\gamma m^{7/8} \right\}$, for all $X \in \q{H}_{n}$ and for all $w \in \bb{R}^{m}$
the error bound
\[
\norm{X - \Delta \left( \q{A} X + w \right) }_{p} \le \frac{ D_1 } { r^{ 1 - \frac{1}{p} } } \norm{X_{r}^{c} }_{1} + \frac{D_4 r_*^{ \frac{1}{p}} \norm{w}_2}{\sqrt m}
\]
holds, where  $D_1$ and $D_4$ are positive constants depending only on $\rho$, and $\Delta$ refers to the solution of the equality nuclear norm under constrained minimizer \cref{eq: rec noiseless}.

%a reconstruction map that can be chosen as either equality-constrained nuclear norm minimization $\Delta$
%\[
%\Delta(b) := \argmin \left\{ \norm{Z}_{1},\ Z \in \q H_n \text{ such that } \q{A} Z =  b  \right\}
%\] 
%or as the Lasso recovery program $\Delta_{2,\mu}^{\text{Lasso}}$
%\[
%\Delta_{2,\mu}^{\text{Lasso}}(b) := \argmin \left\{ \norm{\q{A} Z - b }_{2} + \mu \norm{Z}_{1}, Z \in \q H_n \right\}
%\]
%with parameter $\mu \leq \frac{C_\rho}{\sqrt{r_* m}}$, where $C_{\rho}$ is a constant that depends on $\rho$.\\
Furthermore, a minimizer $\Delta_{q,\eta}$ of inequality-constrained nuclear norm minimization \cref{eq: rec noisy} with noise estimate $\eta$ then fulfills
\[
\norm{X - \Delta_{q,\eta} \left( \q{A} X + w \right) }_{p} \le \frac{ D_1 } { r^{ 1 - \frac{1}{p} } } \norm{X_{r}^{c} }_{1} + \frac{D_5  r_*^{ \frac{1}{p}} \max\{ \norm{w}_2, \eta\}}{\sqrt m}
\]
for all $X \in \q H_n$ and all $w \in \R^m$, where  $D_5$ is a positive constant depending on $\rho$. 
\end{theorem}

Besides providing a thorough theoretical understanding of the noise-blind reconstruction map $\Delta$, \Cref{col: final} also provides improved guarantees for inequality-constrained nuclear norm minimization $\Delta_{q,\eta}$ that go beyond the ones of \Cref{thm: Kabanava3.1}, very much in the spirit of \cite{BA18}. In particular, unlike inequality \cref{eq: rec guarantees eta} of \Cref{thm: Kabanava3.1}, which only applies if $\|w\|_q \leq \eta$, we have obtained a guarantee that depends on $ \max(\norm{w}_2, \eta)$, suggesting a good performance of $\Delta_{q,\eta}$ in the case of \emph{underestimated} noise level, too.

As discussed in \Cref{sec:intro:rankone}, a noise-blind guarantee similar to \Cref{col: final} has been shown for the recovery of positive semidefinite low-rank matrices from rank-one measurements via \cref{eq:least:ellq} \cite{Candes.2014, Kabanava.2016}. 
However, the scaling of the guarantees of \cite{Candes.2014, Kabanava.2016} is slightly better with respect to the rank $r_*$ in the second summand of the error bound, as their results have a factor $r_{*}^{1/p-1/2}$ instead of $r_{*}^{1/p}$. On the other hand, their method and theoretical analysis is explicitly designed for positive semidefinite matrices and cannot generalize beyond this case.

The proof of \Cref{col: final} relies on a relation between the quotient property and noise-blind recovery, which is independent of the specific measurement scenario. 
The next statement summarizes this relation, which is very much analogous to the results in \cite{Wojtaszczyk10} (see also \cite[Chapter 11]{Foucart.2013}) and, for underestimated noise levels, in \cite{BA18} that have been obtained in the context of sparse recovery.
%Although the results of \Cref{col: final} are bound to the case Gaussian rank-one measurements,  the concept of the quotient property can be used for the noise-blind reconstruction of the abstract measurement operators, similarly to the robust rank null space property. 
%The next statement establishes that the $S_1$-quotient property and the $S_2$-robust rank null space property combined with  reasonable reconstruction map are sufficient for a stable recovery in the noise-blind scenarios. In sparse recovery, its analogue in the case of the noise-blind recovery can be found in \cite{Wojtaszczyk10} (see also \cite[Chapter 11]{Foucart.2013}) and in case of underestimated noise level in inequality constrained convex optimization in \cite{BA18}.

%For the related sparse recovery problem, the key to noise-blind results is the $\ell_q$-quotient property, as it was observed in \cite{Wojtaszczyk.2010}, see also \cite[Chapter 11]{Foucart.2013} for an extensive discussion. When combined with a null space property and a reasonable reconstruction map based on convex optimization, the $\ell_q$-quotient property implies desired recovery guarantees for noise-blind sparse recovery. 

%Following this analogy, we introduce the $S_q$-quotient property.

\begin{theorem}[{Low-rank matrix recovery analogue of \cite[Theorem 11.12]{Foucart.2013} and \cite[Theorem 4]{BA18} }]
\label{thm: quotient+NSP to optimality}
Let $1 \le p \le 2$. Assume that a measurement operator $\q{A}: \q{H}_{n} \to \bb{R}^{m}$ satisfies:\\
- the $S_2$-robust rank null space property of order $r_*$ and constants $0 < \rho < 1$ and $\tau > 0$ relative to norm $\norm{\cdot}$, and \\
- the $S_1$-quotient property with constant $d$ and rank $r_*$ relative to the norm $\norm{\cdot}$.
%\begin{enumerate}
% \item[A)]
%\indent If the recovery map $\q R: \bb{R}^{m} \to \q{H}_{n}$ possesses recovery guarantees for the noiseless recovery of approximately low rank matrices such that
%\begin{equation}\label{eq: requested rec guar}
%\norm{X - \q R \left( \q{A} X \right) }_p \le \frac{ 2 (1 + \rho)^2 } { (1 - \rho) r_*^{1 - 1/p} } \norm{X_{r_*}^{c} }_1, \text{ for all } X \in \q{H}_{n},
%\end{equation}
%then, for $r \le r_*$ and for all $X \in \q{H}_{n}$ and all $w \in \bb{R}^{m}$, inequality
%\[
%\norm{X - \q R \left( \q{A} X + w \right) }_{p} \le \frac{ D_1} { r^{ 1 - \frac{1}{p} } } \norm{X_{r}^{c} }_{1} + (D_2 d + \tau )r_*^{ \frac{1}{p} - \frac{1}{2} } \norm{w},
%\]
%holds, where $D_1, D_2>0$ are constants depending on $\rho$. \\
%\item[B)] Let $\eta > 0$. If a $\eta$-dependent recovery map $\q R_{\eta}: \bb{R}^{m} \to \q{H}_{n}$  possesses \emph{robust} low-rank matrix recovery guarantees, that is for all $X \in \q{H}_{n}$ and $w \in \bb{R}^{m}$ such that $\norm{w} \le \eta$ inequality
%\begin{equation*}
%\norm{X - \q R_{\eta} \left( \q{A} X + w \right) }_p \le \frac{ 2 (1 + \rho)^2 } { (1 - \rho) r_*^{1 - 1/p} } \norm{X_{r_*}^{c} }_1 + \frac{2 \tau (3 + \rho)}{ 1 - \rho } r_*^{1/p - 1/2} \eta
%\end{equation*}
%holds,
%then for $r \le r_*$ and for all $w \in \bb{R}^{m}$ it holds that
%\[
%\norm{X - \q R_{\eta} \left( \q{A} X + w \right) }_{p} \le \frac{ D_1} { r^{ 1 - \frac{1}{p} } } \norm{X_{r}^{c} }_{1} + D_3 r_*^{ \frac{1}{p} - \frac{1}{2} } \max \{ \norm{w}, \eta \},
%\]
%where $D_3>0$ is a constant depending on $\rho$,$\tau$ and $d$.
%\end{enumerate}

Let $\eta > 0$. If a $\eta$-dependent recovery map $\q R_{\eta}: \bb{R}^{m} \to \q{H}_{n}$  possesses \emph{robust} low-rank matrix recovery guarantees, that is for all $X \in \q{H}_{n}$ and $w \in \bb{R}^{m}$ such that $\norm{w} \le \eta$ inequality
\begin{equation*}\label{eq: requested rec guar}
\norm{X - \q R_{\eta} \left( \q{A} X + w \right) }_p \le \frac{ 2 (1 + \rho)^2 } { (1 - \rho) r_*^{1 - 1/p} } \norm{X_{r_*}^{c} }_1 + \frac{2 \tau (3 + \rho)}{ 1 - \rho } r_*^{1/p - 1/2} \eta
\end{equation*}
holds,
then for $r \le r_*$ and for all $w \in \bb{R}^{m}$ it holds that
\[
\norm{X - \q R_{\eta} \left( \q{A} X + w \right) }_{p} \le \frac{ D_1} { r^{ 1 - \frac{1}{p} } } \norm{X_{r}^{c} }_{1} + (D_2 d + D_3) r_*^{ \frac{1}{p} - \frac{1}{2} } \max \{ \norm{w}, \eta \},
\]
where $D_1, D_2>0$ are constants depending on $\rho$ and $D_3>0$ is a constant depending on $\rho$ and $\tau$.
\end{theorem}

We note that \Cref{thm: quotient+NSP to optimality} does \emph{not} apply to positive semidefinite-based reconstruction maps. The reason for this is the fact that the psd cone is a cone, but not a linear space and the quotient property will not hold. For instance, if the noise vector $w$ has negative entries, the corresponding matrix $U$ in the quotient property is Hermitian, but not a positive semidefinite and, hence, infeasible for the solvers.

\section{Numerical Experiments}\label{sec: numerics}

In this section, we explore the validity of the 	noise-blind recovery guarantees presented in \Cref{col: final} for low-rank matrix recovery problems with noisy rank-one measurements. We also compare reconstruction errors achieved by noise-blind reconstruction programs to those achieved by \emph{noise-aware} methods such as inequality-constrained nuclear norm minimization \cref{eq: rec noisy}.

We conduct the experiments for simple, small-sized random problem instances where the dimension $n$ of the matrix to be recovered  and the number of measurements $m$ relate in a way that would allow for exact recovery in the noiseless case. In all experiments, we construct the random measurement operator $\q{A}$ as in \eqref{eq:measop:rankone:mainresults}  by drawing i.i.d.~complex Gaussian rank-one measurements.
%\[
%X \mapsto \left(\q{A} (X)_j\right)_{j=1}^m
%\]
%with $\q{A} (X)_j =\langle a_j a_j^*,X \rangle_F$ for all $j=1,\ldots,m$, where the entries of $a_j$ are i.i.d. standard complex Gaussian.

As reconstruction methods, we use many of the optimization problems that have been discussed, in particular \emph{equality-constrained nuclear norm minimization} $\Delta$ (\texttt{NucNorm}), which is the method defined by \cref{eq: rec noiseless}, \emph{inequality-constrained nuclear norm minimization} $\Delta_{2,\eta}$ with parameter $\eta$ based on an $\ell_2$-error \cref{eq: rec noisy} (also called \emph{nuclear norm denoising}, \texttt{NucNormDN}), and the \texttt{MatrixLasso} $\Delta_{2,\mu}^{\text{Lasso}}$ with parameter data fit parameter $\mu$, cf. \cref{eq: rec Lasso}.

When the matrix to be recovered is positive semidefinite, we also compare these methods with approaches that make explicit use of the this property, such as \emph{inequality-constrained nuclear norm minimization} $\Delta_{2,\eta}^{\succeq 0}$ with parameter $\eta$, analogue of $\Delta_{2,\eta}$ constrained to 
%based on an $\ell_2$-error \emph{on the psd cone} \eqref{eq: rec noisy:psd}
\emph{the psd cone}  (\emph{phase lift denoising}, \texttt{PhaseLiftDN}) 
%\cref{eq: rec noisy:psd} 
%\begin{equation}\label{eq: rec noisy:psd}
%\Delta_{2,\eta}^{\succeq 0}(b) := \argmin \left\{ \norm{Z}_{S_1}, Z \succeq 0 \text{ such that } \norm{\q{A} Z - b }_{2} \le \eta \right\},
%\end{equation}
 and $\ell_2$- and $\ell_1$-minimization of the residual $\mathcal{A}(X)-b$ on the cone of positive semidefinite matrices, that is, the solution $\Delta_{q}^{\succeq 0}$ of \cref{eq:least:ellq} for $q=2$ and $q=1$, respectively (\texttt{PosDef-$\ell_2$-min} and \texttt{PosDef-$\ell_1$-min}).

The experiments were conducted on a Linux node with Intel Xeon E5-2690 v3 CPU with 28 cores and 64 GB RAM, using MATLAB R2019a. All optimization problem were modeled using the CVX package \cite{MichaelGrantandStephenBoyd.2020} and solved by SDPT3 \cite{SDPT3:Tutuncu:2003}.

\subsection{Robust Recovery of Rank-One Matrices} \label{sec:num:rankone}
%\begin{equation*}
%\Delta_{q,\eta}(b) := \argmin \left\{ \norm{Z}_{S_1}, Z \in \q H_n \text{ such that } \norm{\q{A} Z - b }_{\ell_q} \le \eta \right\}
%\end{equation*}
%\begin{equation*}
%\Delta_{q,\eta}^{\succeq 0}(b) := \argmin \left\{ \norm{Z}_{S_1}, Z \succeq 0 \text{ such that } \norm{\q{A} Z - b }_{\ell_q} \le \eta \right\},
%\end{equation*}
%\[
%\Delta(b) := \argmin \left\{ \norm{Z}_{1},\ Z \in \q H_n \text{ such that } \q{A} Z =  b  \right\}
%\] 
%or the Lasso recovery program
%\[
%\Delta_{2,\mu}^{\text{Lasso}}(b) := \argmin \left\{ \norm{\q{A} Z - b }_{2} + \mu \norm{Z}_{1}, Z \in \q H_n \right\}
%\]
In our first set of experiments, we study the noise robustness of different methods for the task of recovering rank-ones matrices $X_0 = x_0 x_0^*$ from rank-one measurements perturbed by different types of noise. This setting corresponds to (noisy) phase retrieval \cite{Candes.2011,Candes.2013,Candes.2014,Kueng.2017,Kabanava.2016}, as explained above in \cref{eq:phaseretrieval:measurements}.

In particular, we sample vectors $x_0 \in \bb C^n$ randomly with respect to the Haar measure on the complex unit sphere $S^{n-1} = \{x \in \bb C^n: \|x\|_2 = 1\}$ to define $(n \times n)$-ground truth matrices $X_0 = x_0 x_0^*$. Independently from $x_0$ and the complex Gaussian vectors $a_j$ defining $\q A$, we sample a random real vector $w \in \R^m$ from the sphere $ \eta S^{m-1} = \{w \in \R^m: \|w\|_2 = \eta \}$ with radius $\eta = 0.01$ and 
\begin{equation} \label{eq:noisymeas:num}
 b := |a_j^*x_0|^2 + w = \q A(x_0 x_0^*) + w,
\end{equation}
a measurement vector $b$ that is perturbed by \emph{spherical noise} $w$ such that $\|w\|_2 = \eta = 0.01$.

For $n=50$ and a range of parameters $m$ between $m=50$ and $m=500$, we compare the reconstructions $\widehat{X}$ of the recovery algorithms \texttt{NucNorm}, \texttt{NucNormDN}, \texttt{PhaseLiftDN}, \texttt{PosDef-$\ell_2$-min} and \texttt{PosDef-$\ell_1$-min}, which are provided with $b$ as an input, with $X_0$ and measure the relative Frobenius error $\|\widehat{X}-X_0\|_F/\|X_0\|_F$. For \texttt{NucNormDN} and \texttt{PhaseLiftDN}, we provide the \emph{oracle} noise level estimate $\eta$ as an input parameter.

In \Cref{fig:QP:experiment:1}, the resulting recovery errors are reported, averaged across $100$ independent realizations of experimental setup. 

\begin{figure}[!b]
%\centering
%\begin{subfigure}[b]{0.5\textwidth}
    \setlength\figureheight{65mm} 
    \setlength\figurewidth{65mm}
\input{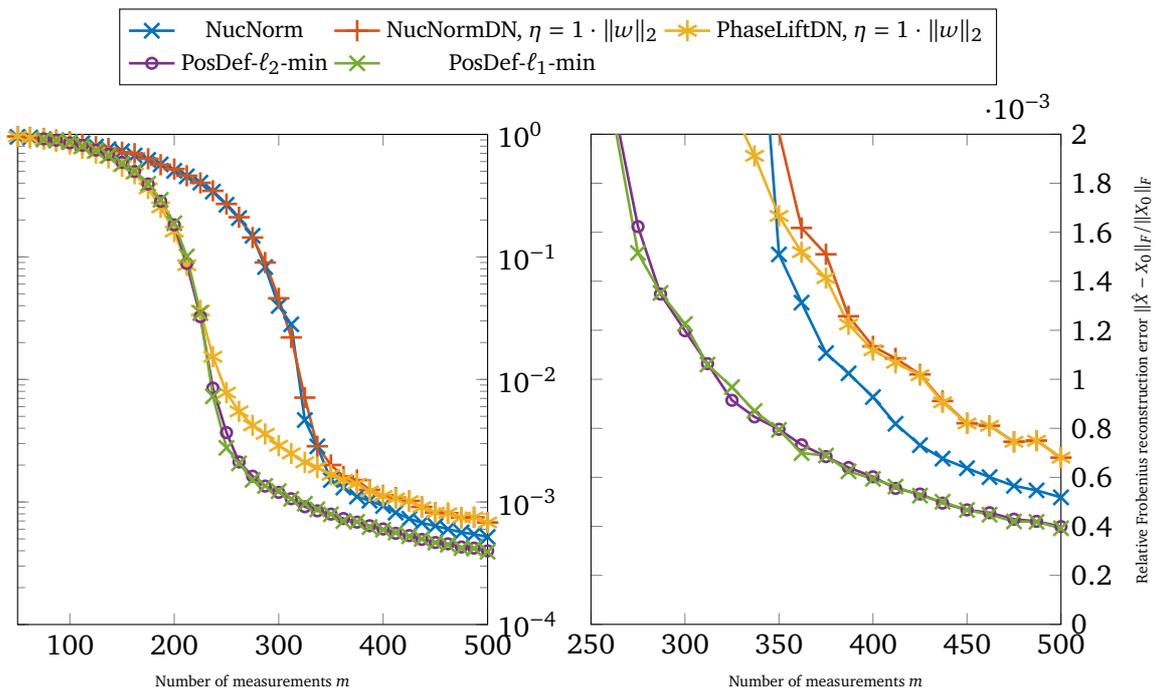}
%\end{subfigure}
%\begin{subfigure}[b]{0.5\textwidth}
%    \setlength\figureheight{50mm} 
%    \setlength\figurewidth{70mm}
%\input{img/experiment_rankQP_noisetype_gaussian_rank1_fig1b.tex}
%\end{subfigure}
%\vspace*{-8mm}
\caption{
\textbf{Comparison of reconstruction algorithms for phase retrieval with spherical noise}: Recovery of normalized rank-$1$ matrices $X_0=x_0 x_0^{*} \in \bb C^{n \times n}$, $n=50$, from noisy rank-one measurements $b = \mathcal{A}(X_0) + w \in \R^m$ perturbed by spherical noise $w$ such that $\|w\|_2= 0.01$. \emph{x-axis}: Number of measurements $m$, \emph{y-axis}: Relative Frobenius error $\|\widehat{X}-X_0\|_F/\|X_0\|_F$ of reconstruction $\widehat{X}$, averaged across $100$ experiments. \\
Left column: Logarithmic scaling of y-axis for $m \in \{50,\ldots,500\}$, right column: Linear scaling of y-axis for $m \in \{250,\ldots,500\}$.}
\label{fig:QP:experiment:1}
\end{figure}

We observe that for the algorithms \texttt{PhaseLiftDN}, \texttt{PosDef-$\ell_2$-min} and \texttt{PosDef-$\ell_1$-min}, which all optimize on the cone of positive semidefinite matrices, the relative error falls below $10^{-2}$ if the number of measurements surpasses a number between $m=200$ and $m=250$. On the other hand, \texttt{NucNorm}, \texttt{NucNormDN}, which do not use the positive definiteness, require at least $m=300$ measurements to pass the threshold of a relative Frobenius error of $10^{-2}$. This shows that for few measurements, including positive definiteness as a constraint helps to identify the desired matrix. As we can see in the right column of \Cref{fig:QP:experiment:1}, the behavior of \texttt{NucNormDN} starts to follow closely the one of \texttt{PhaseLiftDN} for $m \geq 350$. These two methods are fully noise-aware and use oracle knowledge of the $\ell_2$-norm of the noise $\eta$ as an input parameter. We observe that \texttt{NucNorm}, which is a noise-blind method, consistently exhibits a lower error than the noise-aware ones in the stable region between $m=350$ and $m = 500$ (for example, it is $24\%$ lower for $m=500$ with $5.18 \times 10^{-4}$ to $6.80 \times 10^{-4}$, respectively). This validates somewhat \Cref{col: final}, confirming that \texttt{NucNorm} returns estimates that are proportional to the noise magnitude $\|w\|_2$ for a reasonable set range of parameters $m$ and $n$. Furthermore, we observe that \texttt{PosDef-$\ell_2$-min} and \texttt{PosDef-$\ell_1$-min}, which are also noise-blind methods,  consistently return the reconstructions with the smallest error compared to the ones based on the  nuclear norm. The experiments suggest that for phase retrieval, at least for the considered noise model, incorporating the positive definiteness constraint indeed is beneficial if used in the formulations \texttt{PosDef-$\ell_2$-min} and \texttt{PosDef-$\ell_1$-min}.

\begin{figure}[t]
%\centering
\begin{subfigure}[b]{0.46\textwidth}
\vspace*{-5mm}
    \setlength\figureheight{45mm} 
    \setlength\figurewidth{60mm}
\definecolor{mycolor1}{rgb}{0.00000,0.44700,0.74100}%
\definecolor{mycolor2}{rgb}{0.85000,0.32500,0.09800}%
\definecolor{mycolor3}{rgb}{0.92900,0.69400,0.12500}%
\definecolor{mycolor4}{rgb}{0.49400,0.18400,0.55600}%
\definecolor{mycolor5}{rgb}{0.46600,0.67400,0.18800}%
\begin{tikzpicture}

\begin{axis}[%
width=0.951\figurewidth,
height=\figureheight,
at={(0\figurewidth,0\figureheight)},
scale only axis,
xmin=300,
xmax=2500,
xlabel style={font=\color{white!15!black}},
xlabel={Number of measurements $m$},
ymode=log,
ymin=5e-05,
ymax=0.01,
yminorticks=true,
ylabel style={font=\color{white!15!black}},
axis background/.style={fill=white},
yticklabel pos=right,
legend style={font=\fontsize{4}{30}\selectfont, anchor=south, legend columns = 2, at={(0.5,1.12)}},
xlabel style={font=\tiny},ylabel style={font=\tiny},
]
\addplot [color=mycolor1, line width=1.0pt, mark size=3.0pt, mark=x, mark options={solid, mycolor1}]
  table[row sep=crcr]{%
315.656616577619	0.0169864646463464\\
325	0.00466258098743298\\
337	0.00284271050033135\\
350	0.00150957174804032\\
362	0.00131293777801014\\
375	0.00110666258097337\\
387	0.00102404179154309\\
400	0.000927537964421178\\
412	0.000818529886631242\\
425	0.000731081924664235\\
437	0.000676305331058398\\
450	0.000637218854013397\\
462	0.000600861361371094\\
475	0.000565318756436521\\
487	0.00054672132060951\\
500	0.00051821970598209\\
550	0.000427741094682461\\
600	0.000366412727562629\\
650	0.000319308531837485\\
700	0.00028419393065258\\
750	0.000255925632868375\\
800	0.000237372929912933\\
850	0.000222280782194099\\
900	0.000199877120042443\\
950	0.000191391626716485\\
1000	0.000178744756345471\\
1050	0.000167151169104654\\
1100	0.000160466523574845\\
1150	0.000152383427001985\\
1200	0.000148822928679299\\
1250	0.000141538982563772\\
1300	0.000137133281068001\\
1350	0.000132920780492881\\
1400	0.000133961931861949\\
1450	0.000129186326352178\\
1500	0.000128300698470468\\
1550	0.000127739270498421\\
1600	0.000132884254562177\\
1650	0.000129115946498635\\
1700	0.000128876590207126\\
1750	0.000130552804082939\\
1800	0.000134342785249366\\
1850	0.000134507987832006\\
1900	0.00014116593038624\\
1950	0.000147082763953101\\
2000	0.000148373883620878\\
2050	0.000156008043353727\\
2100	0.000164386284818\\
2150	0.000160889146804289\\
2200	0.000175674343651467\\
2250	0.000192020492871909\\
2300	0.000206623167127772\\
2350	0.000249293604968907\\
2400	0.000304685626579701\\
2450	0.000439415499464469\\
2500	0.00814798048909337\\
};
\addlegendentry{NucNorm}

\addplot [color=mycolor2, line width=1.0pt, mark size=3.0pt, mark=+, mark options={solid, mycolor2}]
  table[row sep=crcr]{%
315.64830498281	0.0169864646463464\\
325	0.00465236797344677\\
337	0.00283747277499551\\
350	0.00150644417713225\\
362	0.0013063152868256\\
375	0.00110190361480999\\
387	0.0010167611385631\\
400	0.000921328699202183\\
412	0.000812010446232648\\
425	0.000725845436073308\\
437	0.000670087842896058\\
450	0.000630788619206331\\
462	0.000594731545767681\\
475	0.000559871244666711\\
487	0.000541229070908227\\
500	0.000511691536181678\\
550	0.000427071256297783\\
600	0.000363394437693406\\
650	0.000320959108945304\\
700	0.000278510493308043\\
750	0.00025367577119263\\
800	0.000231836192969633\\
850	0.000220578668004266\\
900	0.000189202909073345\\
950	0.000180116234521201\\
1000	0.000172680722550334\\
1050	0.000159850917406375\\
1100	0.000156736699279368\\
1150	0.000147445007347648\\
1200	0.000136385825252621\\
1250	0.000134842413397247\\
1300	0.000128327514467095\\
1350	0.000122939882693622\\
1400	0.000123937447052491\\
1450	0.00011920143732156\\
1500	0.000114169386380545\\
1550	0.000114661985384771\\
1600	0.000113539147760424\\
1650	0.000113339058715834\\
1700	0.000110674189212513\\
1750	0.00011081662977303\\
1800	0.000113795162560607\\
1850	0.000115785748852522\\
1900	0.000115659405364651\\
1950	0.00011538501893409\\
2000	0.000115675156056983\\
2050	0.000120906638103769\\
2100	0.00012530397146488\\
2150	0.000133616402522764\\
2200	0.000133152350042478\\
2250	0.000145966945643407\\
2300	0.000152754086237426\\
2350	0.000170065900449917\\
2400	0.000177830940863994\\
2450	0.000216473542075808\\
2500	0.0003244111312447\\
};
\addlegendentry{NNDN, $\eta=0.1\cdot\|w\|_2$}

\addplot [color=mycolor3, line width=1.0pt, mark size=3.0pt, mark=asterisk, mark options={solid, mycolor3}]
  table[row sep=crcr]{%
314.647870514332	0.0169864646463464\\
325	0.00666099755399848\\
337	0.00243019513578428\\
350	0.00167711646721522\\
362	0.00127948856078484\\
375	0.0012016941905784\\
387	0.000995454114777738\\
400	0.000899523553951144\\
412	0.00085327502536584\\
425	0.000756709625175632\\
437	0.000699379934502625\\
450	0.000622998842544948\\
462	0.000595447375609396\\
475	0.000556228092625624\\
487	0.000553883107122757\\
500	0.000505395121489744\\
550	0.000416593041862032\\
600	0.000350114533400402\\
650	0.000314771915330267\\
700	0.000273714744173765\\
750	0.000246089941865148\\
800	0.000222912599530358\\
850	0.000202319779802345\\
900	0.000184142368853323\\
950	0.000173549187880354\\
1000	0.000163513586706617\\
1050	0.000149887350878107\\
1100	0.000142561913449382\\
1150	0.000131380846066949\\
1200	0.000126707475148524\\
1250	0.000116801195865417\\
1300	0.000113020351323553\\
1350	0.000107848412230697\\
1400	0.000103040159184257\\
1450	9.73839278544863e-05\\
1500	9.3171120392144e-05\\
1550	9.00616100151365e-05\\
1600	8.5847400234269e-05\\
1650	8.33695656877767e-05\\
1700	7.98208362340396e-05\\
1750	7.79171674693405e-05\\
1800	7.32308454934656e-05\\
1850	7.42831034572251e-05\\
1900	7.02175773674068e-05\\
1950	6.78493077296241e-05\\
2000	6.71061243236234e-05\\
2050	6.57433220078738e-05\\
2100	6.42243268854808e-05\\
2150	6.2884494387453e-05\\
2200	6.09417950126326e-05\\
2250	6.15168329136191e-05\\
2300	6.01344742303931e-05\\
2350	5.90285516521282e-05\\
2400	5.9801114350929e-05\\
2450	5.94208476652057e-05\\
2500	5.76383974123906e-05\\
};
\addlegendentry{NNDN, $\eta=0.5\cdot\|w\|_2$}

\addplot [color=mycolor4, line width=1.0pt, mark size=3.0pt, mark=o, mark options={solid, mycolor4}]
  table[row sep=crcr]{%
314.989542382894	0.0169864646463464\\
325	0.00710642726863813\\
337	0.00285023567108106\\
350	0.0020042288540679\\
362	0.00161752447491539\\
375	0.00151002745668288\\
387	0.00125709911861021\\
400	0.00113447942649399\\
412	0.00108527071674993\\
425	0.00101999111659905\\
437	0.000911472550437422\\
450	0.000820842576313081\\
462	0.000810537591554103\\
475	0.000745144789798713\\
487	0.000749924550164586\\
500	0.000680294428711603\\
550	0.000560474373537836\\
600	0.000484258948507073\\
650	0.000429738241002964\\
700	0.00036654734467272\\
750	0.000328121569667185\\
800	0.000293195978295495\\
850	0.000273135033265735\\
900	0.000252212291975247\\
950	0.000227865450833098\\
1000	0.000211217018797265\\
1050	0.000198762992648363\\
1100	0.000187741530468471\\
1150	0.000174277870935386\\
1200	0.000169154059945913\\
1250	0.00015819978492535\\
1300	0.000149476322158501\\
1350	0.000140967037182966\\
1400	0.000136626667860077\\
1450	0.000132292032223592\\
1500	0.00012555980337737\\
1550	0.000118760797244995\\
1600	0.000115481937906505\\
1650	0.000114497463205662\\
1700	0.000106601475547407\\
1750	0.000101355642058235\\
1800	0.000101880174865011\\
1850	9.65849453220837e-05\\
1900	9.51138652385062e-05\\
1950	9.06513659109768e-05\\
2000	9.08649182494882e-05\\
2050	8.81263033792395e-05\\
2100	8.16643112061189e-05\\
2150	7.96436406964745e-05\\
2200	7.80097717975541e-05\\
2250	7.52155490049828e-05\\
2300	7.51496876581103e-05\\
2350	7.45158745629734e-05\\
2400	6.93450209752642e-05\\
2450	6.89606426656129e-05\\
2500	6.72242663731157e-05\\
};
\addlegendentry{NNDN, $\eta=1\cdot\|w\|_2$}

\addplot [color=mycolor5, line width=1.0pt, mark size=3.0pt, mark=x, mark options={solid, mycolor5}]
  table[row sep=crcr]{%
316.441789024305	0.0169864646463464\\
325	0.00890257470575533\\
337	0.00465169301204194\\
350	0.00365955837551787\\
362	0.00318019777072077\\
375	0.00299486068060646\\
387	0.00261516211602774\\
400	0.00240682164035953\\
412	0.00232164329458487\\
425	0.00214507068704389\\
437	0.0019707876052834\\
450	0.00177382582098467\\
462	0.00172333006370858\\
475	0.00163738498008311\\
487	0.00164556412153024\\
500	0.00149484349357893\\
550	0.00109651922763172\\
600	0.000949800849642947\\
650	0.000822512238657995\\
700	0.000744200082721983\\
750	0.000665299464760697\\
800	0.00061169025099843\\
850	0.000573291680077589\\
900	0.000529394138310705\\
950	0.000505526281809793\\
1000	0.000478274210489374\\
1050	0.000455083580397124\\
1100	0.000442273630761114\\
1150	0.000420771255199102\\
1200	0.000399379636382621\\
1250	0.000391147046572071\\
1300	0.000381346329999239\\
1350	0.000372683238698068\\
1400	0.000361633151626248\\
1450	0.000349945610806271\\
1500	0.000340583334036655\\
1550	0.000334764209766208\\
1600	0.000326779626696132\\
1650	0.000312036274030919\\
1700	0.00030744009154917\\
1750	0.000306081456051027\\
1800	0.000295083805950851\\
1850	0.00028830821268638\\
1900	0.000285430639203587\\
1950	0.000280976845768281\\
2000	0.000274201948445773\\
2050	0.000270809777302277\\
2100	0.000265045690859087\\
2150	0.000261256174392904\\
2200	0.000257340014096187\\
2250	0.000252733650337808\\
2300	0.000248747212894388\\
2350	0.000244276642080298\\
2400	0.000241988361320549\\
2450	0.000238208779743732\\
2500	0.000236052995302054\\
};
\addlegendentry{NNDN, $\eta=2\cdot\|w\|_2$}

\end{axis}
\end{tikzpicture}%
\end{subfigure}
%\vspace*{2mm}
\begin{subfigure}[b]{0.60\textwidth}
\hspace*{2mm}\vspace*{-3.0mm}
    \setlength\figureheight{45mm} 
    \setlength\figurewidth{60mm}
\definecolor{mycolor1}{rgb}{0.00000,0.44700,0.74100}%
\definecolor{mycolor2}{rgb}{0.85000,0.32500,0.09800}%
\definecolor{mycolor3}{rgb}{0.92900,0.69400,0.12500}%
\definecolor{mycolor4}{rgb}{0.49400,0.18400,0.55600}%
\definecolor{mycolor5}{rgb}{0.46600,0.67400,0.18800}%
\begin{tikzpicture}

\begin{axis}[%
width=0.951\figurewidth,
height=\figureheight,
at={(0\figurewidth,0\figureheight)},
scale only axis,
xmin=300,
xmax=2500,
xlabel style={font=\color{white!15!black}},
xlabel={Number of measurements $m$},
ymode=log,
ymin=5e-05,
ymax=0.01,
yminorticks=true,
ylabel style={font=\color{white!15!black}},
ylabel={Relative Frobenius reconstruction error $\|\hat{X}-X_0\|_{F}/\|X_0\|_{F}$},
axis background/.style={fill=white},
yticklabel pos=right,
legend style={font=\fontsize{4}{30}\selectfont, anchor=south, legend columns = 2, at={(0.5,1.12)}},
xlabel style={font=\tiny},ylabel style={font=\tiny},
]
\addplot [color=mycolor1, line width=1.0pt, mark size=3.0pt, mark=x, mark options={solid, mycolor1}]
  table[row sep=crcr]{%
315.656616577619	0.0169864646463464\\
325	0.00466258098743298\\
337	0.00284271050033135\\
350	0.00150957174804032\\
362	0.00131293777801014\\
375	0.00110666258097337\\
387	0.00102404179154309\\
400	0.000927537964421178\\
412	0.000818529886631242\\
425	0.000731081924664235\\
437	0.000676305331058398\\
450	0.000637218854013397\\
462	0.000600861361371094\\
475	0.000565318756436521\\
487	0.00054672132060951\\
500	0.00051821970598209\\
550	0.000427741094682461\\
600	0.000366412727562629\\
650	0.000319308531837485\\
700	0.00028419393065258\\
750	0.000255925632868375\\
800	0.000237372929912933\\
850	0.000222280782194099\\
900	0.000199877120042443\\
950	0.000191391626716485\\
1000	0.000178744756345471\\
1050	0.000167151169104654\\
1100	0.000160466523574845\\
1150	0.000152383427001985\\
1200	0.000148822928679299\\
1250	0.000141538982563772\\
1300	0.000137133281068001\\
1350	0.000132920780492881\\
1400	0.000133961931861949\\
1450	0.000129186326352178\\
1500	0.000128300698470468\\
1550	0.000127739270498421\\
1600	0.000132884254562177\\
1650	0.000129115946498635\\
1700	0.000128876590207126\\
1750	0.000130552804082939\\
1800	0.000134342785249366\\
1850	0.000134507987832006\\
1900	0.00014116593038624\\
1950	0.000147082763953101\\
2000	0.000148373883620878\\
2050	0.000156008043353727\\
2100	0.000164386284818\\
2150	0.000160889146804289\\
2200	0.000175674343651467\\
2250	0.000192020492871909\\
2300	0.000206623167127772\\
2350	0.000249293604968907\\
2400	0.000304685626579701\\
2450	0.000439415499464469\\
2500	0.00814798048909337\\
};
\addlegendentry{NucNorm}

\addplot [color=mycolor2, line width=1.0pt, mark size=3.0pt, mark=+, mark options={solid, mycolor2}]
  table[row sep=crcr]{%
312	0.0153618250682373\\
325	0.00711084642221174\\
337	0.00328022200324072\\
350	0.0033349774520125\\
362	0.00129606986943629\\
375	0.00110881269021092\\
387	0.000978660648859126\\
400	0.000875734407546196\\
412	0.000795529468238957\\
425	0.000753177571135449\\
437	0.000663426763389118\\
450	0.000628638932798695\\
462	0.000603029567444304\\
475	0.000565016945419073\\
487	0.000536846055247679\\
500	0.000512280542682399\\
550	0.000407298615856651\\
600	0.000352499074340924\\
650	0.000309046975291405\\
700	0.000277401907015603\\
750	0.000244741469362041\\
800	0.000220833028979761\\
850	0.000206460499631533\\
900	0.000185351347508018\\
950	0.000169251097319642\\
1000	0.000160020271647149\\
1050	0.000153286725743\\
1100	0.000142972984089511\\
1150	0.000133824964447069\\
1200	0.000124272585429127\\
1250	0.000119517073058568\\
1300	0.000113783670892474\\
1350	0.000107440524783703\\
1400	0.000103876717684939\\
1450	9.85929110853549e-05\\
1500	9.57819547868794e-05\\
1550	9.12744152905097e-05\\
1600	8.82463146731329e-05\\
1650	8.28817460899322e-05\\
1700	7.95625754120424e-05\\
1750	7.7740650962259e-05\\
1800	7.49457144129072e-05\\
1850	7.33284239579036e-05\\
1900	7.10893408915863e-05\\
1950	6.87391637252129e-05\\
2000	6.73733098091209e-05\\
2050	6.35901008009453e-05\\
2100	6.26713607855935e-05\\
2150	6.13548769331873e-05\\
2200	5.94013794398175e-05\\
2250	5.89383510415381e-05\\
2300	5.73009345219312e-05\\
2350	5.55515266564777e-05\\
2400	5.32428626800333e-05\\
2450	5.20303580342452e-05\\
2500	5.14825376960751e-05\\
};
\addlegendentry{MLasso, $\mu=0.025$}

\addplot [color=mycolor3, line width=1.0pt, mark size=3.0pt, mark=asterisk, mark options={solid, mycolor3}]
  table[row sep=crcr]{%
314.940075548473	0.0169864646463464\\
325	0.00702680661475925\\
337	0.00275905335408833\\
350	0.00191155266354831\\
362	0.00148680521854024\\
375	0.00139841811930184\\
387	0.00114465628850741\\
400	0.0010160521828757\\
412	0.000974840043163993\\
425	0.000875898935685297\\
437	0.000797855912126934\\
450	0.000702072787323034\\
462	0.000675235711313764\\
475	0.000629401971425779\\
487	0.000634165974029516\\
500	0.000570898454379686\\
550	0.000470185254848699\\
600	0.000403160292145051\\
650	0.000356544948123635\\
700	0.000312054927705331\\
750	0.000277087407937259\\
800	0.000246648863742112\\
850	0.000227171087977714\\
900	0.000208027852568223\\
950	0.000193233318881927\\
1000	0.000177682923408918\\
1050	0.000167016948258191\\
1100	0.000157840574574275\\
1150	0.000150324964855578\\
1200	0.000140776211156861\\
1250	0.000133755900105757\\
1300	0.000126681977078502\\
1350	0.00012270151512511\\
1400	0.000119773648892734\\
1450	0.000111367326415555\\
1500	0.000105520976424946\\
1550	0.000101726774897591\\
1600	9.83438229332704e-05\\
1650	9.47191724590961e-05\\
1700	9.24011889906013e-05\\
1750	8.83886978580448e-05\\
1800	8.55866666267683e-05\\
1850	8.47731263871389e-05\\
1900	7.93632459231693e-05\\
1950	7.87919503047564e-05\\
2000	7.51703905438208e-05\\
2050	7.43942950276602e-05\\
2100	7.21000605754204e-05\\
2150	6.96799227214113e-05\\
2200	6.80729205577018e-05\\
2250	6.62770045472414e-05\\
2300	6.48300011554885e-05\\
2350	6.2044181141334e-05\\
2400	6.03423450328025e-05\\
2450	6.00176087721257e-05\\
2500	5.8899850577977e-05\\
};
\addlegendentry{MLasso, $\mu=0.1$}

\addplot [color=mycolor4, line width=1.0pt, mark size=3.0pt, mark=o, mark options={solid, mycolor4}]
  table[row sep=crcr]{%
325	0.011271514236293\\
337	0.00414544415766048\\
350	0.00352335309291655\\
362	0.00294609440504968\\
375	0.00245711318282927\\
387	0.00219007493620935\\
400	0.00202409813282142\\
412	0.00173961557265812\\
425	0.00152579391571468\\
437	0.0014505609587962\\
450	0.00137255878822372\\
462	0.00120252153266335\\
475	0.00114450617544668\\
487	0.00105217657113917\\
500	0.00101819852928642\\
550	0.000759234555336828\\
600	0.000640471589784219\\
650	0.00055186345379722\\
700	0.000460000194825203\\
750	0.000420561166040956\\
800	0.0003676171651222\\
850	0.000336055972060586\\
900	0.000295996322943618\\
950	0.000283235485302548\\
1000	0.000257018144107318\\
1050	0.000238869592363354\\
1100	0.000224332315711282\\
1150	0.000207188256691994\\
1200	0.000200445339019693\\
1250	0.000185227509417758\\
1300	0.000176552170316938\\
1350	0.000165765709406313\\
1400	0.000158246312310622\\
1450	0.000148725378551299\\
1500	0.000144524545143698\\
1550	0.000139164052323255\\
1600	0.00013245347669538\\
1650	0.000129884075515999\\
1700	0.000122586772043585\\
1750	0.000119200791521652\\
1800	0.000112885321339084\\
1850	0.000110094932885172\\
1900	0.00010837181641556\\
1950	0.000100983883966028\\
2000	9.75609715463062e-05\\
2050	9.6554729358075e-05\\
2100	9.42475423969657e-05\\
2150	9.23935550896621e-05\\
2200	8.7755504564915e-05\\
2250	8.61471698488983e-05\\
2300	8.40381541781753e-05\\
2350	8.17923441808876e-05\\
2400	7.98672268475842e-05\\
2450	7.87886846852613e-05\\
2500	7.79459854429428e-05\\
};
\addlegendentry{MLasso, $\mu=0.25$}

\addplot [color=mycolor5, line width=1.0pt, mark size=3.0pt, mark=x, mark options={solid, mycolor5}]
  table[row sep=crcr]{%
325	0.0131699430825287\\
337	0.00843083862440618\\
350	0.00662342864675862\\
362	0.00546555204189234\\
375	0.00509993526419907\\
387	0.00415499070300463\\
400	0.00361807272274605\\
412	0.00345479855691498\\
425	0.00301397117372948\\
437	0.00270555686914759\\
450	0.00228596491959803\\
462	0.00216781858447504\\
475	0.00203798704003512\\
487	0.00206435966773473\\
500	0.00178498596634352\\
550	0.00138507153309241\\
600	0.00111607644859839\\
650	0.000929263073060644\\
700	0.000777085769224866\\
750	0.000705745129918051\\
800	0.000614627505519567\\
850	0.000553554336057251\\
900	0.000505133293831027\\
950	0.000467678392480752\\
1000	0.000423434926048339\\
1050	0.000398623746071207\\
1100	0.000369718230537398\\
1150	0.000342379825813632\\
1200	0.000325761591762357\\
1250	0.00030974757688167\\
1300	0.000293307301074079\\
1350	0.000278374000743963\\
1400	0.00027172374717804\\
1450	0.00025338354936554\\
1500	0.000241116705248307\\
1550	0.000232081269095971\\
1600	0.000227020613735685\\
1650	0.000214911809138109\\
1700	0.000207373992134722\\
1750	0.00020035856036941\\
1800	0.000196145964524206\\
1850	0.000184094262766738\\
1900	0.000180758758733167\\
1950	0.000174750789611592\\
2000	0.000172994634580631\\
2050	0.000164999884077273\\
2100	0.000157644414090058\\
2150	0.000155785277081957\\
2200	0.000151498880527016\\
2250	0.000145093538709188\\
2300	0.000142583976624641\\
2350	0.000141713687967292\\
2400	0.000135168844808156\\
2450	0.000132795489835819\\
2500	0.000131124149397064\\
};
\addlegendentry{MLasso, $\mu=0.5$}

\end{axis}
\end{tikzpicture}%
\end{subfigure}
%\vspace*{-8mm}
\caption{\textbf{Comparison of \texttt{NucNorm}, \texttt{NucNormDN} and \texttt{MatrixLasso} for phase retrieval from many noisy measurements}: Recovery of normalized rank-$1$ matrices $X_0=x_0 x_0^{*} \in \bb C^{50 \times 50}$, $n=50$, from noisy rank-one measurements $b = \mathcal{A}(X_0) + w \in \R^m$ with spherical noise $w$ such that $\|w\|_2= 0.01$. Left column: \texttt{NucNormDN (NNDN)} $\Delta_{2,\eta}$ for different parameters $\eta$, right column: \texttt{MatrixLasso (MLasso)} $\Delta_{2,\mu}^{\text{Lasso}}$  for different parameters $\mu$.}
\label{fig:QP:experiment:NNDNLasso:rank1}
\end{figure}
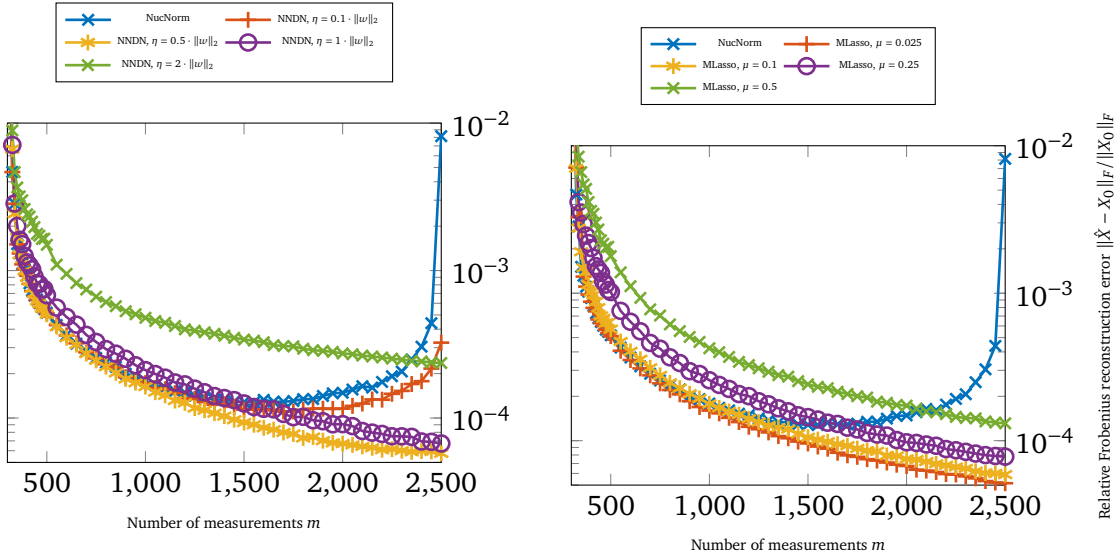

Next, we repeat the same experiment for the methods \texttt{NucNorm} as well as \texttt{NucNormDN (NNDN)} $\Delta_{2,\eta}$ and \texttt{MatrixLasso (MLasso)} $\Delta_{2,\mu}^{\text{Lasso}}$ with different choices of the noise balancing parameters $\eta$ and $\mu$, respectively, concentrating on a higher range of measurements with $m$ between $m=300$ and $m=2500$. For $m=2500$, given that we consider square matrices of with size $n=50$, we would be able to solve the reconstruction problem exactly in the absence of noise by simple linear algebra. We choose the parameter $\eta$ of \texttt{NucNormDN} in the range of $\eta \in  \{0.1,0.5,1,2\}$ and the parameter $\mu$ of \texttt{MatrixLasso} in the range of $\mu \in \{0.025,0.1,0.25,0.5\}$. We do not report results of this experiment for a number of measurements $m < 300$, as the results are all very similar in that range. We illustrate the results of this experiment in \Cref{fig:QP:experiment:NNDNLasso:rank1}. 

Our first observation is that the the choice of $\eta$ and $\mu$ is especially sensitive to \emph{overestimates}: We see in the first column of \Cref{fig:QP:experiment:NNDNLasso:rank1} that for a moderate number of measurements $m < 1000$, the smallest reconstruction errors are obtained by the noise-blind method \texttt{NucNorm} (which corresponds to \texttt{NucNormDN} with $\eta = 0$) and by the choices of $\eta=0.1 \|w\|_2$ and $\eta = 0.5 \|w\|_2$ which \emph{underestimate} the noise level. On the other hand, the ``oracle'' choice of $\eta = 1\|w\|_2$ is slightly worse in this range of $m$, and \emph{overestimate} of a factor $2$ by choosing $\eta= 2\|w\|_2$ deteriorates the relative Frobenius error of the reconstruction considerably. In the right column of \Cref{fig:QP:experiment:NNDNLasso:rank1}, we see that \texttt{MatrixLasso} the balancing parameter $\mu$ of the 
\texttt{MatrixLasso} $\Delta_{2,\mu}^{\text{Lasso}}$ exhibits similar behavior as $\eta$ for \texttt{NucNormDN}, up to a different scaling. Increasing the number of measurements $m$ reduces the relative Frobenius reconstruction error consistently for all methods as long as $m < 1000$.

For very large numbers of measurements $1500 \leq m \leq 2500$, however, we observe an additional phenomenon: The performance of \texttt{NucNorm} and \texttt{NucNormDN} with $\eta = 0.1 \|w\|_2$ actually \emph{deteriorates} as $m$ grows in that range. For \texttt{MatrixLasso}, this is not observed, at least not for the considered parameter choices of $\mu \geq 0.025$.

One explanation of this observation is that if $m = 2500$ or close to that, the linear system $\mathcal A X = b$ has only a unique solution (or a very low-dimensional set of solutions), so that the noise-blind method \texttt{NucNorm} has only a unique feasible point (or a very low-dimensional feasible set). If $b$ is quite noisy, it is reasonable to expect that the solution of \texttt{NucNorm} is farther away from the ground truth rank-one matrix $X_0$ than methods that do \emph{not} use a strict equality constraint, but allow for a trade-off between data-fit and complexity measure, such as \texttt{NucNormDN} or \texttt{MatrixLasso} for larger noise parameters $\eta$ and $\mu$.

These observations can be also interpreted in view of the $S_1$-quotient property. Evidently, for large $m$ and severely underestimated noise level, a reconstruction error that is proportional to the $\ell_2$-norm of the noise $\|w\|_2$, as in the statement of \Cref{col: final}, is not possible in general. We recall that the proof of \Cref{col: final} is based on the $S_1$-quotient property. Thus, it is clear that some sort of upper bound on the the number of measurements $m$ with respect to the dimension $n$ in the assumption of \Cref{col: final} is to be expected and not just an artifact of our proof.

\subsection{Robust Recovery of Hermitian Rank-Two Matrices} \label{sec:num:ranktwo:indefinite}
\begin{figure}[!b]
%\centering
%\begin{subfigure}[b]{0.5\textwidth}
    \setlength\figureheight{60mm} 
\setlength\figurewidth{136mm}
\definecolor{mycolor1}{rgb}{0.00000,0.44700,0.74100}%
\definecolor{mycolor2}{rgb}{0.92900,0.69400,0.12500}%
\definecolor{mycolor3}{rgb}{0.49400,0.18400,0.55600}%
\definecolor{mycolor4}{rgb}{0.46600,0.67400,0.18800}%
\definecolor{mycolor5}{rgb}{0.63500,0.07800,0.18400}%
\begin{tikzpicture}

\begin{axis}[%
width=0.951\figurewidth,
height=\figureheight,
at={(0\figurewidth,0\figureheight)},
scale only axis,
xmin=50,
xmax=625,
xlabel style={font=\color{white!15!black}},
xlabel={Number of measurements $m$},
ymode=log,
ymin=0.0001,
ymax=1,
yminorticks=true,
ylabel style={font=\color{white!15!black}},
ylabel={Relative Frobenius reconstruction error $\|\hat{X}-X_0\|_{F}/\|X_0\|_{F}$},
axis background/.style={fill=white},
yticklabel pos=right,
legend style={font=\fontsize{8}{30}\selectfont, legend cell align=left, align=left, draw=white!15!black},
xlabel style={font=\tiny},ylabel style={font=\tiny},
]
\addplot [color=mycolor1, line width=1.0pt, mark size=4.5pt, mark=x, mark options={solid, mycolor1}]
  table[row sep=crcr]{%
50	0.744558933086373\\
62	0.696558489655408\\
75	0.628152823986283\\
87	0.58246782682833\\
100	0.529222800984889\\
112	0.491806558192419\\
125	0.393760649835142\\
137	0.331730309650047\\
150	0.266891486317981\\
162	0.214990519799712\\
175	0.159816624534038\\
187	0.103063270480435\\
200	0.0599397175154906\\
212	0.0207706837049355\\
225	0.00592221248628261\\
237	0.00139180342382917\\
250	0.00119684037858719\\
262	0.000876341110566917\\
275	0.000791830770647284\\
287	0.000685391629690381\\
300	0.000595231613281585\\
312	0.000568200981736187\\
325	0.00049043922032269\\
337	0.000516907750828383\\
350	0.000443636188295329\\
362	0.000409813651759032\\
375	0.000394016589451916\\
387	0.00038939748058865\\
400	0.000379541482995713\\
412	0.000379672948782205\\
425	0.000380187597616557\\
437	0.000385907376414512\\
450	0.000336450408998027\\
462	0.000390580702599454\\
475	0.000392957216833118\\
487	0.000398580279454263\\
500	0.000382043040482965\\
525	0.000429328982162703\\
550	0.000498296370188164\\
575	0.000709920776933557\\
600	0.000668328562086573\\
625	0.0149600655350711\\
650	0.000636743161615789\\
};
\addlegendentry{NucNorm}

\addplot [color=mycolor2, line width=1.0pt, mark size=4.5pt, mark=asterisk, mark options={solid, mycolor2}]
  table[row sep=crcr]{%
50	0.744601660950051\\
62	0.696601673752524\\
75	0.628200562451027\\
87	0.582515433375694\\
100	0.529277133764327\\
112	0.491866443287587\\
125	0.393825086412863\\
137	0.3318027073558\\
150	0.266956414351485\\
162	0.215058804557282\\
175	0.159889024571484\\
187	0.103151746747758\\
200	0.0600601606763561\\
212	0.0209530911868248\\
225	0.00596557633012854\\
237	0.00151844998792346\\
250	0.00115075326824061\\
262	0.000909582730720645\\
275	0.000726832466255833\\
287	0.00066029188410391\\
300	0.000552587987332828\\
312	0.000506606733660192\\
325	0.000461305802498297\\
337	0.000453089675757339\\
350	0.000379160458480435\\
362	0.000369967888362122\\
375	0.000351941449610823\\
387	0.000321218661967556\\
400	0.000300859061629382\\
412	0.000295534060462412\\
425	0.00027719530368052\\
437	0.000270714081734091\\
450	0.000254533892242186\\
462	0.000241861606345105\\
475	0.000234799739897126\\
487	0.000221559864510248\\
500	0.000209847559075567\\
525	0.000201556108052536\\
550	0.000185117960072708\\
575	0.00017422866991967\\
600	0.000175963219871556\\
625	0.000176704784784502\\
650	0.000158547680354591\\
};
\addlegendentry{NNDN, $\eta=0.5\cdot\|w\|_2$}

\addplot [color=mycolor3, line width=1.0pt, mark size=4.5pt, mark=o, mark options={solid, mycolor3}]
  table[row sep=crcr]{%
50	0.744645143596799\\
62	0.696645062951588\\
75	0.628248759124299\\
87	0.582563154595297\\
100	0.529331834148988\\
112	0.491926004404626\\
125	0.393888464277258\\
137	0.331875682159639\\
150	0.267021518952151\\
162	0.215127793604363\\
175	0.159975187993247\\
187	0.103312889938175\\
200	0.0603529994267624\\
212	0.0214727303302781\\
225	0.00637556926909584\\
237	0.00193258150931762\\
250	0.00138900806900407\\
262	0.00116810603320864\\
275	0.000903744123898623\\
287	0.000864557477468958\\
300	0.000715634689274184\\
312	0.000657807115726813\\
325	0.000604309943314948\\
337	0.000584118343673384\\
350	0.000512774048838156\\
362	0.00050226163548182\\
375	0.000481477671794637\\
387	0.000416182107031357\\
400	0.000424729411295225\\
412	0.000383273191029966\\
425	0.000369516518842348\\
437	0.000374415153901127\\
450	0.000351021968851185\\
462	0.000312029097473732\\
475	0.000307957876634196\\
487	0.000302129027203449\\
500	0.000294960495638407\\
525	0.000266407042050576\\
550	0.000252179975661861\\
575	0.000237721769070706\\
600	0.000209528349966862\\
625	0.000200872332525897\\
650	0.000198979001358204\\
};
\addlegendentry{NNDN, $\eta=1\cdot\|w\|_2$}

\addplot [color=mycolor4, line width=1.0pt, mark size=4.5pt, mark=x, mark options={solid, mycolor4}]
  table[row sep=crcr]{%
50	0.744732747070037\\
62	0.696732121037755\\
75	0.628344641786331\\
87	0.582659307889464\\
100	0.529440802238133\\
112	0.492045161108082\\
125	0.394015563277534\\
137	0.332024004608034\\
150	0.267155630357862\\
162	0.215272909131541\\
175	0.16017876592339\\
187	0.103747022487528\\
200	0.0611816398785272\\
212	0.0229528204242746\\
225	0.0076751371721852\\
237	0.00312647081694673\\
250	0.00234756672373621\\
262	0.00201176409875718\\
275	0.00163201676947552\\
287	0.00160229184149503\\
300	0.00137092147565007\\
312	0.00128271083751678\\
325	0.00121022633959165\\
337	0.00118868633699717\\
350	0.00107695469125846\\
362	0.00106112464905571\\
375	0.00102880349840814\\
387	0.00094305033274913\\
400	0.000951161319640028\\
412	0.000900169703926485\\
425	0.000872825316138539\\
437	0.000887530246656077\\
450	0.000843197170879108\\
462	0.000807377751604838\\
475	0.000799625768053872\\
487	0.000778094579598133\\
500	0.00075888824773045\\
525	0.000730616547095216\\
550	0.000703208197755737\\
575	0.000673718178721649\\
600	0.00063612468828081\\
625	0.000623739876197023\\
650	0.000601140478433263\\
};
\addlegendentry{NNDN, $\eta=2\cdot\|w\|_2$}

\addplot [color=mycolor5, dashdotted, line width=1.0pt, mark size=4.5pt, mark=diamond, mark options={solid, mycolor5}]
  table[row sep=crcr]{%
50	0.744590598500047\\
62	0.696592212820235\\
75	0.628194577787721\\
87	0.582512343022435\\
100	0.529273723101967\\
112	0.491864032297974\\
125	0.393825586211017\\
137	0.331809675619626\\
150	0.266960964162926\\
162	0.215066220846701\\
175	0.159898963240254\\
187	0.103161937292863\\
200	0.060077614312296\\
212	0.0209743981564223\\
225	0.00597117877310139\\
237	0.00151187531648397\\
250	0.00115372710588381\\
262	0.000909965497646503\\
275	0.000729786908259564\\
287	0.000665985658436226\\
300	0.000556083491425201\\
312	0.000510806452238101\\
325	0.000464993756030524\\
337	0.000458207395358411\\
350	0.000382545288584843\\
362	0.00037396142160013\\
375	0.0003570853373622\\
387	0.000323463982755321\\
400	0.000306930810836514\\
412	0.000298852930956791\\
425	0.000280380835970869\\
437	0.000276444816900737\\
450	0.000258601055694153\\
462	0.000241214951049361\\
475	0.000236037046846237\\
487	0.000224775265973561\\
500	0.000214512629489481\\
525	0.00020124025730296\\
550	0.000182739436542939\\
575	0.000171430030290003\\
600	0.000161760948489931\\
625	0.000154883282944267\\
650	0.00014586833710825\\
};
\addlegendentry{MLasso, $\mu=0.025$}

\addplot [color=mycolor1, dashdotted, line width=1.0pt, mark size=4.5pt, mark=triangle, mark options={solid, rotate=180, mycolor1}]
  table[row sep=crcr]{%
50	0.742144050142206\\
62	0.70016454853509\\
75	0.641229372114157\\
87	0.581003014984322\\
100	0.518223674284233\\
112	0.466310680083782\\
125	0.390634013590138\\
137	0.335241697484722\\
150	0.279125055803603\\
162	0.218700992262659\\
175	0.152202203272614\\
187	0.1229746140878\\
200	0.0564751676106389\\
212	0.0345020325952294\\
225	0.0107835317951245\\
237	0.00565566446359283\\
250	0.00411314345034341\\
262	0.00343061438665022\\
275	0.0026121743044578\\
287	0.0025986517919122\\
300	0.00212697548752134\\
312	0.00190147898449098\\
325	0.00171128329883811\\
337	0.00164754443202202\\
350	0.0014989985424259\\
362	0.00140700847591889\\
375	0.00136213575328751\\
387	0.00122552177200958\\
400	0.00120175742457904\\
412	0.00109844321071106\\
425	0.00104499498873447\\
437	0.00105656945622722\\
450	0.000959194015200244\\
462	0.000925132081545231\\
475	0.000891630848842636\\
487	0.000852471532625673\\
500	0.000814646726456697\\
525	0.000780297900847694\\
550	0.000692144075199876\\
575	0.000686678155328135\\
600	0.000619596225882665\\
625	0.000579220682371027\\
650	0.000565018422446672\\
};
\addlegendentry{MLasso, $\mu=0.25$}

\end{axis}
\end{tikzpicture}%
\caption{\textbf{Comparison of \texttt{NucNorm}, \texttt{NucNormDN} and \texttt{MatrixLasso} for the recovery of indefinite rank-2 matrices recovery from noisy measurements}: Recovery of \mbox{rank-$2$} matrices $X_0 = 1\cdot v_1 v_1^*+ (-0.5) v_2 v_2^* \in \bb C^{n \times n}$, $n=25$, from noisy rank-one measurements $b = \mathcal{A}(X_0) + w \in \R^m$ with spherical noise $w$ such that $\|w\|_2= 0.01$.} %\emph{x-axis}: Number of measurements $m$, \emph{y-axis}: Relative Frobenius error $\|\widehat{X}-X_0\|_F/\|X_0\|_F$ of reconstruction $\widehat{X}$, averaged across $100$ independent experiments.
\label{fig:QP:experiment:3}
\end{figure}
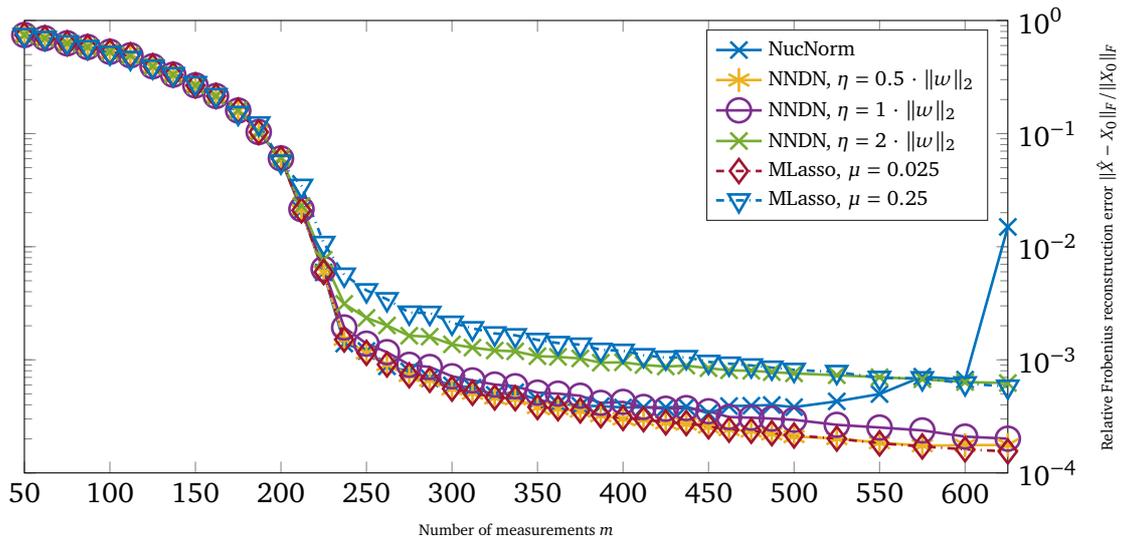

We recall that in the experiment of \Cref{fig:QP:experiment:1}, it was observed that methods incorporating the positive definiteness constraint into their optimization such as \texttt{PosDef-$\ell_2$-min} and \texttt{PosDef-$\ell_1$-min} resulted in the lowest reconstruction errors for the recovery of rank-one matrices.

In this subsection, we consider the recovery of Hermitian rank-$2$ matrices $X_0 = v_1 v_1^*+ (-0.5) v_2 v_2^*$ with orthonormal vectors $v_1,v_2 \in \R^n$ that are drawn uniformly from the Stiefel manifold. For this case, the only available reconstruction methods (among the ones considered in this paper) are \texttt{NucNorm}, \texttt{NucNormDN} and the \texttt{MatrixLasso} as $X_0$ is not positive semidefinite. Independently from $v_1$ and $v_2$ and the complex Gaussian vectors $a_j$ defining $\q A$, we sample a random real vector $w \in \R^m$ from the sphere $ \eta S^{m-1} = \{w \in \R^m: \|w\|_2 = \eta \}$ with radius $\eta = 0.01$ and define the measurement vector 
\begin{equation} \label{eq:noisymeas:num}
 b := \q A(X_0) + w,
\end{equation}
that is perturbed by \emph{spherical noise} $w$ such that $\|w\|_2 = \eta = 0.01$.

For $n=25$ and a range of parameters $m$ between $m=r n =50$ and $m= n^2= 625$, we compare the reconstructions $\widehat{X}$ of the recovery algorithms \texttt{NucNorm}, as well as \texttt{NucNormDN (NNDN)} $\Delta_{2,\eta}$ and \texttt{MatrixLasso (MLasso)} $\Delta_{2,\mu}^{\text{Lasso}}$, which are provided with $b$ as an input, with $X_0$ and measure the relative Frobenius error $\|\widehat{X}-X_0\|_F/\|X_0\|_F$. For \texttt{NucNormDN} and \texttt{MatrixLasso}, we provide different noise level parameters $\eta$ and $\mu$ as an input parameter.

\begin{figure}[t]
%\centering
%\begin{subfigure}[b]{0.5\textwidth}
    \setlength\figureheight{85mm} 
    \setlength\figurewidth{138mm}
\input{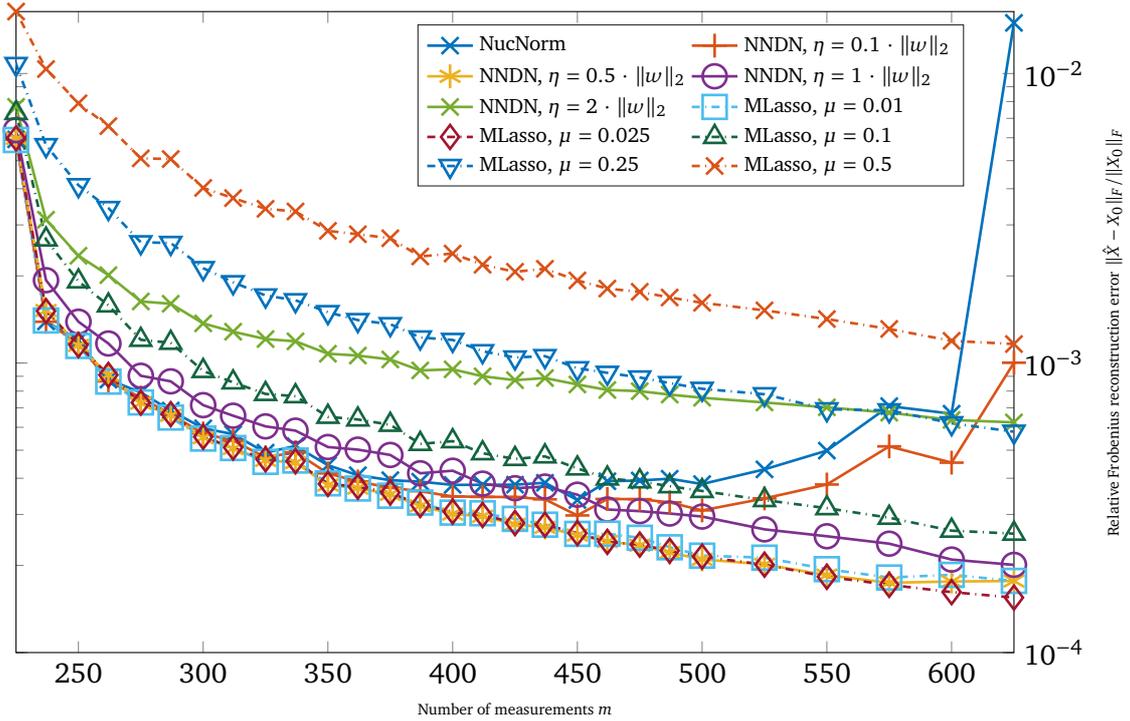}
\caption{\textbf{Comparison of \texttt{NucNorm}, \texttt{NucNormDN} and \texttt{MatrixLasso} for the recovery of indefinite rank-2 matrices recovery from many noisy measurements.} As \Cref{fig:QP:experiment:3}, but with a number of measurements in restricted range $m \in \{225,\ldots,625\}$.}
\label{fig:QP:experiment:4}
\end{figure}

In \Cref{fig:QP:experiment:3} we report the resulting recovery errors for $\eta \in \{0.5 \|w\|_2,1 \|w\|_2, 2 \|w\|_2\}$ and $\eta \in \{0.025, 0.25\}$, averaged across $100$ independent realizations of experimental setup. We observe that the relative Frobenius error of the reconstruction falls below $10^{-2}$ between $m=200$ and $m=250$ for all considered methods. This is interesting since compared to the experiments in \Cref{sec:num:rankone}, the dimension was halved from $n=50$ to $n=25$, but the rank was doubled from $r=1$ to $r=2$, and these methods needed more than $m=300$ measurements to obtain an error below $10^{-2}$. We note that among the considered methods, \texttt{NucNormDN} for $\eta = 0.5\|w\|_2$ and \texttt{MatrixLasso} for $\mu = 0.025$ result in the lowest errors for all considered $m$, with $\eta = 0.5\|w\|_2$ being a choice that \emph{underestimates} the noise level by $50\%$. \texttt{NucNorm} does almost equally well for a moderate number of measurements until $m \approx 350$, after which its performance deteriorates from around $3.8\cdot 10^{-4}$ to $1.5 \cdot 10^{-2}$ at $m=625$, when the system becomes square.

Finally, in \Cref{fig:QP:experiment:4}, we illustrate the results of the same experiment for more parameters $\eta$ and $\mu$ when restricted to $m \in \{225,\ldots, 625\}$ -- such a measurement complexity would result in exact recovery via \texttt{NucNormDN} in the case of noiseless measurements $w = 0$. We observe also here that an \emph{overestimate} of $\eta$ and $\mu$ (e.g., by $100\%$ with the choice of $\eta = 2 \|w\|_2$ or with $\mu = 0.5$ for \texttt{MatrixLasso}) has more negative consequences than underestimating their ``oracle'' choice. Choosing $\eta$ as small as $\eta = 0.1  \|w\|_2$ results in a qualitatively very similar behavior as \texttt{NucNorm} with a performance deterioration for large $m$.

As a summary, we note these experiments show that a noise-blind recovery of low-rank indefinite matrices is indeed possible via convex formulations such as \texttt{NucNorm} if the measurement matrices are random and rank-one. This is compatible with the theory of \Cref{thm: quotient+NSP to optimality,col: final,thm: QP} that are based on the $S_2$-robust null space property and the $S_1$-quotient property of the measurement operator $\q A$. We also note that noise-blind recovery works particularly well if the number of measurements $m$ is only \emph{moderate} such that $m$ is closer to $m \approx C r n$ than to $m \approx C n^2$.

\section{Proofs}\label{sec: proofs}
In this section, we detail the proofs of our main results \Cref{thm: QP} about the $S_1$-quotient property of random measurement operators with Gaussian rank-one matrices and \Cref{col: final} about noise-blind recovery guarantees for such measurements. The proof of a general result for arbitrary measurement operators, \Cref{thm: quotient+NSP to optimality} concludes this section.%,of \Cref{col: Lasso noiseless,col: final}.

\subsection{Proof of \Cref{thm: QP}} 
%\subsubsection{Auxiliary statements}
%
%Proving the $S_1$-quotient property directly by Definition \ref{def: qp} requires both proving the existence of matrix $U$ and an bound on the sum of its singular values, which is a cumbersome task. Likely, it is possible to provide an equivalent definition in terms of the adjoint operator and dual norms.

%The proof of the theorem again closely follows its $\ell_1$-quotient counterpart and can be found in the appendix. 

%\Cref{thm: equivalent qp} provides us a convenient  tool to prove the $S_1$-quotient property of the measurement operator $\q A$. However, it involves the adjoint operator $\q A^*$, which is yet unknown. Next lemma provides an explicit formula for $\q A^*$.

\subsubsection{Proof concept and structure}
The goal of this section is to establish that the scaled measurement operator $\frac{1}{\sqrt m}\q{A}$ possesses the \mbox{$S_1$-quotient} property with constant $\frac{128 \sqrt 2}{\kappa} \sqrt{\kappa m/n}$ and rank $\kappa m/n$ relative to the norm $\norm{\cdot}_2$ on $\bb{R}^m$ for all $\kappa > 0$ with high probability. We show this property by establishing an equivalent form, as given by the following proposition, which is exactly in line with the analogous result for sparse recovery (see, e.g., \cite[Lemma 11.17]{Foucart.2013}). For completeness, we provide a proof in the appendix.

\begin{proposition}\label{thm: equivalent qp}
	For $q \ge 1$, a measurement operator $\q{A}: \q{H}_{n} \to \bb{R}^{m}$ possess the $S_q$-quotient property with constant $d$ and rank $r_*$ relative to norm $\norm{\cdot}$ if and only if
	\begin{equation}\label{eq: qp in dual norms}
	\norm{w}_{*} \le d r_{*}^{\frac{1}{q}-\frac{1}{2}}\norm{\q{A}^* w}_{q^{*}}, \ \text{for all } w \in \bb{R}^{m},
	\end{equation}
	where $\q{A}^*$ is the adjoint of the measurement operator, $\norm{\cdot}_{*}$ is a dual norm associated with $\norm{\cdot}$ and $q^{*}$ is a H\"{o}lder dual of $q$.
\end{proposition}

%By \Cref{thm: equivalent qp}, it is equivalent to showing that Inequality \eqref{eq: qp in dual norms} holds with high probability. 
%To establish \eqref{eq: qp in dual norms},
 Recalling that the dual of $\norm{\cdot}_2$ is $\norm{\cdot}_2$ itself and the H{\"o}lder dual of $1$ is $\infty$, we observe that after normalization of $w$, Inequality \eqref{eq: qp in dual norms} reads as
\begin{equation}\label{eq: desired final}
\norm{\frac{1}{\sqrt m} \q{A}^* w}_\infty \ge \frac{n}{128\sqrt{2}m}, \quad \text{for all } w \in \bb{R}^m \text{ satisfying } \norm{w}_2 = 1.
\end{equation}

%Further, by dropping the scaling factor, it suffices to show inequality
%\[
%\norm{\q{A}^* w}_\infty \ge \frac{n}{128\sqrt{2m}}.
%\]

We will establish this inequality via a covering argument (see Section \ref{sec: joining results}), for which we need
\begin{equation}\label{eq: desired}
\norm{\frac{1}{\sqrt m} \q{A}^* w}_\infty \ge \frac{n}{64\sqrt{2}m},
\end{equation}
to hold with high probability for any  fixed vector $w$ with $\norm{w}_2 = 1$. 

%In , these bounds are united to form a uniform bound, which holds for all $w\in \bb{R}^m$ such that  $\norm{w}_2 = 1$ simultaneously.

The left hand side can be made explicit via the following well-known formula for the dual of  $\q{A}$ (see, e.g., \cite[Lemma 3.1]{CLS15}). Again we include a proof in the appendix for completeness.

\begin{lemma}\label{l: dual operator}
	The adjoint operator $\q{A}^{*}: \bb{R}^{m} \to \q H_n$ is given by
	\begin{equation*} \label{eq: dual operator}
	\q{A}^{*} w := \sum_{k = 1}^{m} w_k a_k a_k^{*}.
	\end{equation*}
\end{lemma}

%We start by recalling that by Lemma \ref{l: dual operator}, the adjoint operator $\q{A}^*:\R^m \to \q{H}_n$  is given by
%\[
%(w_k)_{k=1}^m \mapsto \q{A}^* w = \sum_{k=1}^m w_k a_k a_k^*.
%\] 
With this lemma and using that  $\q{A}^* w$ is Hermitian, we obtain the following estimate for the left hand side of \eqref{eq: desired}. % matrix $\q{A}^* w$ is Hermitian and hence it has real eigenvalues, so that $v^* (\q{A}^* w) v \in \bb{R}$ for all $v \in \bb{C}^n$. Moreover, the spectral norm $\norm{\q{A}^* w}_\infty$ is the eigenvalue with the largest magnitude. Then, it holds that 
\begin{equation}\label{eq: top eigen}
\norm{\q{A}^* w}_\infty  
= \max_{v \in \bb{C}^n, \norm{v}_2=1}  \left| v^* (\q{A}^* w) v \right|
= \max_{v \in \bb{C}^n, \norm{v}_2=1}  \left|  \sum_{k =1}^m w_k | a_k^*  v |^2 \right|
\end{equation}

Depending on the properties of the vector $w$, the remainder of the proof is split into three separate cases, each requiring a different approach. 

We will first consider vectors whose entries sum to a number significantly different from zero. As this leads to a non-zero expectation of the sum in \eqref{eq: top eigen}, the lower bound can be established via a concentration argument. 

Secondly, we will consider vectors with some large entries, as formalized by a lower bound on the supremum norm. In this case, we can bound the right hand side of \eqref{eq: top eigen} by choosing $v$ to be the normalized measurement vector $a_j$ corresponding to the largest entry.

In the last case of a vector with entries averaging to a small number, but without large entries, we select a subset of entries with sufficiently large magnitudes and construct a suitable $v$ from the associated measurement vectors.

 The common ingredient in all three cases is Bernstein's inequality for subexponential random variables. Recall that a random variable $X$ is subexponential if $\|X\|_{\psi_1}:= \sup_{p\geq 1} p^{-1} (\E |X|^p)^{1/p}$ is finite.  For more details about subexponential random variables we refer reader to Section 2.7 of \cite{Vershynin.2018}.

\begin{proposition}[{Bernstein's inequality, version of \cite[Theorem 2.8.2]{Vershynin.2018}}]\label{p: bernstein}
Let $K>0$ and let $\xi_1, \ldots, \xi_N$ be independent subexponential random variables with $\|\xi_1\|_{\psi_1}\leq K$ for all $i\in [N]$. Then, for every $w \in \bb{R}^N$ and every $t>0$ it holds that
\[
\bb{P} \left( \left| \sum_{j=1}^N w_j (\xi_j  - \E \xi_j) \right| \ge t \right) 
\le 2 \exp \left\{ -c \min \left\{  \frac{t^2}{K^2 \norm{w}_2^2}, \frac{t}{K \norm{w}_\infty} \right\} \right\},
\]
as well as 
\[
\bb{P} \left(  \sum_{j=1}^N w_j \xi_j  \le  \sum_{j=1}^N w_j \E \xi_j -t \right)
\le  \exp \left\{ -c \min \left\{  \frac{t^2}{K^2 \norm{w}_2^2}, \frac{t}{K \norm{w}_\infty} \right\} \right\}.
\]
\end{proposition}
%\textcolor{red}{Add here more argumentation, why Bernstein's inequality applies and why we have subexponentials.}
It will be  applied to random variables drawn from a $|\q{CN}(0,1)|^2$ distribution, which are subexponential. Indeed, by definition, standard complex Gaussian random variable $\xi$ is defined as
\[
\xi =(\alpha + i \beta)/\sqrt 2, 
\]
where $\alpha$ and $\beta$ are independent standard Gaussian random variables. Then, 
$\alpha^2 + \beta^2$ follows a chi-squared distribution with 2 degrees of freedom, which coincides with the exponential distribution with parameter $1/2$ and hence $|\xi|^2 = (\alpha^2 + \beta^2)/2$ is subexponential with norm $K =1$. %Hence $|\xi|^2 = (\alpha^2 + \beta^2)/2$ follows a $\frac{1}{2}\mathrm{Exp}(1/2)$ distribution and belongs to the class of subexponential random variables. 
It has expectation 1 and the sum of the expectations in Bernstein's inequality becomes
\[
\sum_{j=1}^N w_j \E \xi_j = \sum_{j=1}^N w_j.
\]

We note than in further statements we also denote constants by $c,C, \tilde C$, but their values may differ; even within a proof or a chain of inequalities.

 %Moreover, in the proofs we will often avoid changing the constant notation, when its value is changed, e.g., $2 c x^2= c x^2$. 

\subsubsection{Non-centered vectors}\label{sec: non-centered}
The first case only considers vectors $w$ with mean value far from zero. In this case, the following holds.
\begin{theorem}\label{thm: non-centered}
Let $w \in \bb{R}^m, \norm{w}_2= 1$ and assume that $w$ satisfies $\left| \sum_{k =1}^m w_k \right| \ge \frac{n}{32\sqrt{2 m}}$. 
Suppose that the number of measurements $m$ satisfies $m \le n^2$.
Then, inequality \eqref{eq: desired} holds on a random event $E_{1,1}(w)$ (depending on $w$), which occurs with a probability of at least $1 - 2\exp\left( - c n^2/\sqrt m \right)$. 
\end{theorem}
\begin{proof}
We establish a lower bound on $\norm{\q{A}^* w}_\infty$ by using Equality \eqref{eq: top eigen} and restricting the maximum occurring in the bilinear representation of $\norm{\q{A}^* w}_\infty$ to standard basis vectors. More precisely, 
\[
\norm{\q{A}^* w}_\infty  
= \max_{v \in \bb{C}^n, \norm{v}_2=1}  \left| v^* (\q{A}^* w) v \right|
\ge \max_{\ell \in [n]}   \left| e_\ell^* (\q{A}^* w) e_\ell \right|
= \max_{\ell \in [n]}  \left| \sum_{k =1}^m w_k | (a_k)_\ell |^2 \right| =: \max_{\ell \in [n]}  |S^1_\ell(w)|.
\]
Thus, it suffices to establish a lower bound for a maximum of $|S^1_\ell(w)|$. The first step in this direction is to fix index $\ell$ and establish the lower bound for a single $|S^1_\ell(w)|$. In order to do so, we introduce a random event
\[
E_{S^1_\ell}(w) := \left\{ |S^1_\ell(w) - \Ex S^1_\ell(w)  | < \frac{n}{64\sqrt{2 m}}\right\}
%\ \text{ and } \
%E_{S^1_\ell,2}(w) := \left\{ |S^1_\ell(w)|  > \frac{n}{64\sqrt{2 m}} \right\}.
.
\]

Under event $E_{S^1_\ell}(w)$, by reverse triangle inequality, it holds that
\[
\left| |S^1_\ell(w)| - |\Ex S^1_\ell(w)|  \right| \le  |S^1_\ell(w) - \Ex S^1_\ell(w)  | < \frac{n}{64\sqrt{2 m}},
\]
and consequently
\[
|S^1_\ell(w)|  > |\Ex S^1_\ell(w)| - \frac{n}{64\sqrt{2 m}}.
\]
The expectation of $S^1_\ell(w)$ is given by
$
\Ex S^1_\ell(w)
%= \sum_{k =1}^m w_k \Ex | (a_k)_\ell |^2
= \sum_{k =1}^m w_k ,
$
and thus, using the assumption on $w$, we obtain
\[
|S^1_\ell(w)|  
> |\Ex S^1_\ell(w)| - \frac{n}{64\sqrt{2 m}}
= \left| \sum_{k =1}^m w_k \right| - \frac{n}{64\sqrt{2 m}}
\ge \frac{n}{32\sqrt{2 m}} - \frac{n}{64\sqrt{2 m}} 
= \frac{n}{64\sqrt{2 m}}.
\]
%so that $E_{S^1_\ell,1}(w) \subseteq E_{S^1_\ell,2}(w)$ and $\bb{P}(E_{S^1_\ell,1}(w)) \le \bb{P}(E_{S^1_\ell,2}(w))$. Thus, it suffices to bound the smaller probability from below, which
The tail probability of the event $E_{S^1_\ell}(w)$ can be bounded via Proposition \ref{p: bernstein}. We recall that $(a_k)_\ell$ are i.i.d standard complex Gaussian random variables, and hence the probability of the complement of $E_{S^1_\ell}(w)$ is bounded from above as 
\[
\bb{P} (E_{S^1_\ell}^C (w) ) \le 2\exp\left\{ - c \min \left\{ \frac{n^2}{64^2 \cdot 2 m  \norm{w}_2^2 }, \frac{n}{64 \sqrt{2m}  \norm{w}_\infty } \right\} \right\}.
\]
Since the number of measurements satisfies $m \le n^2$ and condition $\norm{w}_2 =1$ implies $\norm{w}_\infty \le 1$, this bound simplifies to 
\[
\bb{P} (E_{S^1_\ell}^C (w) ) \le 2\exp\left\{ - c n/\sqrt m \right\}.
% \text{ and hence }
%\bb{P} (E_{S^1_\ell,2}^C (w) ) \le 2\exp\left\{ - c n/\sqrt m \right\}.
\] 
This establishes the desired result for a single $\ell \in [n]$. Our next step is to establish a similar upper bound for the random event 
\[
E_{1,1}(w) := \bigcup_{\ell \in [n]} E_{S^1_\ell}(w)
\subseteq \left\{ \max_{\ell \in [n]} |S^1_\ell(w)|  > \frac{n}{64\sqrt{2 m}} \right\}.
\]
In order to extend it for $E_{1,1}(w)$ without losses in probability, we observe that $S^1_1(w), \ldots, S^1_n(w)$ are i.i.d random variables as a consequence of the fact that the entries of the measurement vectors $(a_k)_\ell, k \in [m], \ell \in [n]$ are independent and the vector are independent as well. This implies that random events $E_{S^1_\ell}^C (w)$ are independent and, thus, using the \mbox{De Morgan's} law, the probability of $E_{1}(w)$ is bounded from below as
\begin{align*}
& \bb{P} \left(E_{1,1}(w) \right) 
= 1 - \bb{P} \left(E_{1,1}^C(w) \right) 
= 1 - \bb{P} \left( \bigcap_{\ell \in [n]} E_{S^1_\ell}^C(w) \right) \\
& \quad = 1 - \prod_{\ell \in [n]} \bb{P} \left(  E_{S^1_\ell}^C (w)  \right)
\ge 1 - 2\exp\left\{ - c n^2/\sqrt m \right\},
\end{align*}
which concludes the proof of Inequality \eqref{eq: desired} in the first case.
\end{proof}

\subsubsection{Spiky vectors}\label{sec: spiky}
The second case considers those of the remaining vectors $w$ which have at least one entry with large magnitude.
\begin{theorem}\label{thm: spiky}
Let $w \in \bb{R}^m, \norm{w}_2= 1$ and assume that $w$ satisfies $\norm{w}_\infty \ge m^{-1/4}$. Suppose that the number of measurements $m$ satisfies 
\[
m \le \min\{ (n/16)^{4/3} , (n/8)^{8/7} \}
\ \text{ and } m \text{ is sufficiently large.}
%(m^{3/4}+1) \log m < c m^{7/8}
\] 
Then, there exists a random event $E_{2,1}(w)$ depending on $w$ and random event $E_{2,2}$ independent of $w$ with tail probabilities
\[
\bb{P}\left(  E_{2,1}^C (w)\right) \le 2 \exp \left\{ -\tilde C n m^{1/8} \right\} 
\ \text{and} \
\bb{P}\left(  E_{2,2}^C \right) 
\le 2m  \exp \left\{ - c n \right\} + 4 m^{m^{3/4}+1} \exp \left\{ -C m^{7/8} \right\}
\] 
such than on $E_{2,1}(w) \cap E_{2,2}$, Inequality \eqref{eq: desired} holds. 
\end{theorem}
\begin{proof}
Let $w_j$ be the entry of $w$ with the largest magnitude so that $|w_j| = \norm{w}_\infty \ge m^{-1/4}$.
Then, in Equality \eqref{eq: top eigen}, we select single $v = a_j / \norm{a_j}_2$ and apply the reverse triangle inequality to obtain the lower bound
\begin{align*}
\norm{\q{A}^* w}_\infty  
& \ge \left| \sum_{k =1}^m w_k \left| a_k^* \frac{ a_j}{\norm{a_j}_2} \right|^2 \right| 
= \left| w_j \norm{a_j}_2^2  +  \sum_{\substack {k =1 \\ k \neq j } }^m w_k  \left| a_k^* \frac{ a_j}{\norm{a_j}_2} \right|^2 \right| 
 \ge |w_j| \norm{a_j}_2^2  
- \left| \sum_{\substack {k =1 \\ k \neq j } }^m w_k  \left| a_k^* \frac{ a_j}{\norm{a_j}_2} \right|^2 \right|.
\end{align*}
Next step is to further split the sum in two by separating the high and low magnitude entries. Let $J$ be an index set containing high magnitude entries, so that
\[
J:=J(w)= \left\{ k \in [m] \ \big| \ |w_k| > m^{-3/8} \right\}.
\]
The cardinality of $J$ is bounded by $\floor{m^{3/4}}$ since
\[
1
= \norm{w}_2^2
= \sum_{k=1}^{m} |w_k|^{2}
\ge \sum_{k \in J} |w_k|^{2}
> \sum_{k \in J} m^{-3/4}
= m^{-3/4} |J|.
\]
Then, we split the sum such that
\[
\left| \sum_{\substack {k =1 \\ k \neq j } }^m w_k  \left| a_k^* \frac{ a_j}{\norm{a_j}_2} \right|^2 \right|
\le
\left| \sum_{\substack {k \in J \\ k \neq j } } w_k  \left| a_k^* \frac{ a_j}{\norm{a_j}_2} \right|^2 \right|
+
\left| \sum_{k \in [m] \backslash J } w_k  \left| a_k^* \frac{ a_j}{\norm{a_j}_2} \right|^2 \right|
=:
\left| \sum_{\substack {k \in J \\ k \neq j } } w_k  \left| a_k^* \frac{ a_j}{\norm{a_j}_2} \right|^2 \right|
+
|S^2(w)|.
\]
The first sum is further bounded from above as 
\[
\left| \sum_{\substack {k \in J  \\ k \neq j } } w_k  \left| a_k^* \frac{ a_j}{\norm{a_j}_2} \right|^2 \right|
\le \sum_{\substack {k \in J   \\ k \neq j } } |w_k| \left| a_k^* \frac{ a_j}{\norm{a_j}_2} \right|^2 
\le |w_j|  \sum_{\substack {k \in J   \\ k \neq j } } \left| a_k^* \frac{ a_j}{\norm{a_j}_2} \right|^2 
:= |w_j| S^3(J,j),
\]
with $S^3(J,j)$ depending only on index set $J$ and $j$.
Combining these bounds, we obtain
\begin{align}
\norm{\q{A}^* w}_\infty  
& \ge |w_j| \left( \norm{a_j}_2^2  - S^3(J,j) \right) -  \left| S^2(w) \right|. \label{eq: spiky pre result}
\end{align}
%The next step is to obtain bound for each term separately. First bound for $\norm{a_j}_2^2$ is a corollary to the following lemma.
%\begin{lemma} \label{l: min norm}
%Let $a_1, a_2, \ldots, a_m$ be independent vectors distributed as $\mathcal{CN}\left(0,I\right)$ in $\bb{C}^{n}$.\\
%Then, a random event
%\[
%E_{\norm{\cdot}} := \left\{ \max_{j \in [m]} \norm{a_j}_{2}^2 < \frac{3n}{2} 
%\text{ and } 
%\min_{j \in [m]} \norm{a_j}_{2}^2 > \frac{n}{2}
%\right\}
%\] 
%occurs  with probability at least $1 - 2m\exp \left\{ -c n \right\}$.
%\end{lemma}
%\begin{proof}
Further, we proceed with separate bounds for each of the obtained terms.
For each $j \in [m]$, we can expand $\norm{a_j}_{2}^2$ as a sum of independent $|\q{CN}(0,1)|^2$ random variables with expectation 1, that is
\[
\norm{a_j}_{2}^2 = \sum_{\ell =1}^n |(a_j)_\ell|^2 %= \sum_{\ell =1}^n (|(a_j)_\ell|^2 -1) + n.
\]
Consider now the random events 
\[
E_{a_j} := \{|\norm{a_j}_{2}^2 - n| < n/2\},
\]
and
\[
E_{\norm{\cdot}} := \left\{ \max_{j \in [m]} \norm{a_j}_{2}^2 < \frac{3n}{2} 
\text{ and } 
\min_{j \in [m]} \norm{a_j}_{2}^2 > \frac{n}{2}
\right\}.
\] 
On $E_{a_j}$, it holds that $n/2 \le \norm{a_j}_{2}^2 \le 3n/2$ and observe that 
$E_{\norm{\cdot}} =\cap_{j=1}^{m} E_{a_j}$. 
By Proposition \ref{p: bernstein}, the probability of $E_{a_j}^C$ is bounded from above by
\[
\bb{P}\left( E_{a_j}^C \right) 
\le 2 \exp \left\{ -c \min \left\{ \frac{n^2}{4  n}, \frac{n}{2 } \right\} \right\}
= 2 \exp \left\{ -c n \right\}.
\]
Then, using union bound and De Morgan's law, we obtain
\[
\bb{P}(E_{\norm{\cdot}}) 
= 1- \bb{P}\left( E_{\norm{\cdot}}^C \right) 
\ge 1 - \bb{P}\left( \bigcup_{j=1}^{m} E_{a_j}^C \right) 
\ge 1 - \sum_{j=1}^{m} \bb{P}(E_{a_j}^C) 
\ge 1 - 2m  \exp \left\{ - c n \right\}.
\] 
We finally note that on $E_{\norm{\cdot}}$, it holds that
\begin{equation}\label{eq: spiky r1}
\norm{a_j}_{2}^2  \ge \min_{j \in [m]} \norm{a_j}_{2}^2 > \frac{n}{2}.
\end{equation}
For upper bound on $S^3(J,j)$, we follow similar steps. When $a_j$ is fixed, the 
$a_k^* \frac{ a_j}{\norm{a_j}_2}$, $k \in J, k \neq j$ are independent random variables distributed as $\q{CN}(0,1)$ distribution, since they are projections of complex Gaussian random vectors \cite{cochran_1934}. 
Hence, $S^3(J,j)$ is a sum of independent $|\q{CN}(0,1)|^2$ random variables. Consider a random event
\[
E_{S^3,J,j} := \{|S^3(J,j) - \Ex S^3(J,j) | < m^{7/8} \}.
\]
The expectation $\Ex S^3(J,j)$ is a number of summands in $S^3$, that is $|J|-1$.
Therefore, on the event $E_{S^3,J,j}$, it holds that 
\begin{equation}\label{eq: spiky r2}
S^3(J,j) < |J|-1 + m^{7/8} \le 2 m^{7/8}.
\end{equation} 
The probability of its complement can be again bounded by Proposition \ref{p: bernstein}. More precisely,
\begin{align*}
\bb{P}\left( E_{S^3,J,j}^C \ \big| \ a_j \right) 
& \le 2 \exp \left\{ -C \min \left\{ \frac{m^{7/4}}{4 (|J|-1)}, \frac{m^{7/8}}{2} \right\} \right\} \\
& \le 2 \exp \left\{ -C \min \left\{ \frac{m^{7/4}}{4 m^{3/4}}, \frac{m^{7/8}}{2} \right\} \right\}
= 2 \exp \left\{ -C m^{7/8} \right\}.
\end{align*}
Integrating out $a_j$ leads to the bound for the unconditional probability
\[
\bb{P}\left( E_{S^3,J,j}^C \right) = \int_{a_j} \bb{P}\left( E_{S^3,J,j}^C \ \big| \ a_j \right)  d \bb{P}(a_j) \le  2 \exp \left\{ -C m^{7/8} \right\}.
\]
The exponent of the obtained tail probability has order $m^{7/8}$ which is less than $m$ (dimension of $w$). It makes it impossible to apply covering argument (for details see Section \ref{sec: joining results}). However, the random event $E_{S^3,J,j}$ depends only on the choice of the index set $J$ and the index $j$. Therefore, we can consider all possible selections of $J$ and $j$ and resulting events $E_{S^3,J,j}$. Define a random event
\[
E_{S^3} := \bigcap_{ J \subset [m], |J| \le \floor{m^{3/4}}  } \bigcap_{j \in J} E_{S^3,J,j}.
\]   
Again, by union bound and De Morgan's law, the probability of $E_{S^3}$ is bounded from below as
\[
\bb{P}\left( E_{S^3} \right)
= 1 - \bb{P}\left( E_{S^3}^C \right)
= 1 - \bb{P}\left( \bigcup_{ J \subset [m], |J| \le \floor{m^{3/4}}  } \bigcup_{j \in J} E_{S^3,J,j}^C\right)
\ge 1 - \sum_{J \subset [m], |J| \le \floor{m^{3/4} } } \sum_{j \in J} \bb{P} ( E_{S^3,J,j}^C ).
\]
The total number of all non-empty subsets of $[m]$ with cardinality up to $\floor{m^{3/4}}$ is given by
\[
\sum_{r = 1}^{\floor{m^{3/4} } } \binom{m}{r} 
\le \sum_{r = 1}^{\floor{m^{3/4} } } \frac{m^r}{r!}
\le \sum_{r = 1}^{\floor{m^{3/4} } } m^r
\le m \frac{m^{m^{3/4}} - 1}{m-1} 
\le 2 m^{3/4}.
\] 
Thus, returning to probability we obtain
\[
\bb{P}\left( E_{S^3} \right) 
\ge 1 - 2 m^{m^{3/4}} |J| \cdot 2 \exp \left\{ -C m^{7/8} \right\}
\ge 1- 4 m^{m^{3/4}+1} \exp \left\{ -C m^{7/8} \right\},
\]
so the order in the exponent 
\[
-C m^{7/8} + (m^{3/4}+1) \log m
\] 
is negative for sufficiently large $m$. %, as assumptions require.
The proof for the last part yet again follows the same logic. For fixed $a_j$, sum $S^2(w)$ is a weighted sum of independent $|\q{CN}(0,1)|^2$ random variables. Its expectation is the sum of weights, that is 
\[
\Ex S^2(w) = \sum_{k \in [m] \backslash J} w_k 
\ \text{ and } \
|\Ex S^2(w)| \le \sum_{k \in [m] \backslash J} |w_k| \le \norm{w}_1 \le \sqrt m.
\]
Consider a random event
\[
E_{2,1}(w) :=  \{|S^2(w) - \Ex S^2(w) | < n/ 16 m^{1/4} \}.
\]
Under $E_{2,1}(w)$, by triangle inequality and assumptions on $m$, it holds that
\begin{equation}\label{eq: spiky r3}
|S^2(w)| \le |S^2(w) - \Ex S^2(w) |  + |\Ex S^2(w)| 
< \frac{n}{16 m^{1/4}} + \sqrt{m} 
%\le \frac{n}{16 m^{1/4}} + \frac{n}{16 m^{1/4}} 
\le \frac{n}{8 m^{1/4}} 
\end{equation}
In order to apply Proposition \ref{p: bernstein} consider a vector $\tilde w$ defined as 
\[
\tilde w_k =  
\begin{cases}
w_k, & k \in [m] \backslash J, \\
0, & k \in J.
\end{cases}
\]
Then, by definition of $J$, $\tilde w$ satisfies
\[
\norm{\tilde w}_2 \le 1 \ \text{ and } \norm{\tilde w}_\infty \le m^{-3/8}.
\]
Hence, the conditional probability of the complement of $E_{2,1}(w)$ is bounded from above as
\begin{align*}
\bb{P}\left( E_{2,1}^C(w) \ \big| \ a_j \right) 
& \le 2 \exp \left\{ -\tilde C \min \left\{ \frac{n^2}{128 m^{1/2} \norm{\tilde w}_2^2}, \frac{n}{16 m^{1/4} \norm{\tilde w}_\infty} \right\} \right\} \\
& \le 2 \exp \left\{ - \tilde C \min \left\{ \frac{n^2}{128 m^{1/2} }, \frac{n}{16 m^{1/4} m^{-3/8} } \right\} \right\} \\
& \le 2 \exp \left\{ - \tilde C n m^{1/8} \min \left\{ \frac{n}{ m^{5/8} }, 1 \right\} \right\} 
\le 2 \exp \left\{ - \tilde C n m^{1/8} \right\},
\end{align*}
where in the last inequality we used that $n \ge 8 m^{7/8} \ge 8 m^{5/8}$. Integrating out $a_j$ grants us
\[
\bb{P}\left(  E_{2,1}^C(w)\right) = \int_{a_j} \bb{P}\left( E_{2,1}^C(w) \ \big| \ a_j \right)  d \bb{P}(a_j) \le 2 \exp \left\{ - \tilde C n m^{1/8} \right\}.
\]
Finally, we define a random event $E_{2,2}: = E_{\norm{\cdot}} \cap E_{S^3}$ with tail probability
\[
\bb{P}\left(  E_{2,2}^C \right) 
= \bb{P}\left(  E_{\norm{\cdot}}^C \cup E_{S^3}^C \right) 
\le 2m  \exp \left\{ - c n \right\} + 4 m^{m^{3/4}+1} \exp \left\{ - C m^{7/8} \right\}. 
\] 
Then, under $E_{2,1}(w) \cap E_{2,2}$ all established bounds \eqref{eq: spiky r1}, \eqref{eq: spiky r2}, \eqref{eq: spiky r3} hold and we can return to the Inequality \eqref{eq: spiky pre result}. Hence, it holds that
\begin{align*}
\norm{\q{A}^* w}_\infty  
& \ge |w_j| \left( \norm{a_j}_2^2  - S^3(J,j) \right) -  \left| S^2(w) \right| 
> |w_j| \left( \frac{n}{2} - 2 m^{7/8} \right) -  \frac{n}{8 m^{1/4}} \\
& \ge |w_j| \frac{n}{4} -  \frac{n}{8 m^{1/4}} 
\ge \frac{n}{4 m^{1/4}} -  \frac{n}{8 m^{1/4}} 
= \frac{n}{8 m^{1/4}},
\end{align*}
where we used condition on $m$ in the second inequality and the choice of $w_j$ in the last inequality. Finally, note that
\[
\norm{\q{A}^* w}_\infty >  \frac{n}{8 m^{1/4}} \ge \frac{n}{64 \sqrt{ 2m} },
\]
which concludes the proof. 
\end{proof}
\subsubsection{Flat vectors} \label{sec: flat}
The last case is when the mean of $w$ is not big enough, so we cannot proceed and in Section \ref{sec: non-centered} and at the same time there is no significantly big entries to compensate the rest as in Section \ref{sec: spiky}. Therefore, the idea is to do something in between these two proof approaches, separate several relatively big entries which will have bigger impact than the rest. Our main result in this section is the following.
\begin{theorem}\label{thm: flat}
Let $w \in \bb{R}^m, \norm{w}_2= 1$ and assume that $w$ satisfies 
\[
\left| \sum_{k =1}^m w_k \right| < \frac{n}{32\sqrt{2 m}}
\ \text{ and }\ 
\norm{w}_\infty < m^{-1/4}.
\] 
Suppose that the number of measurements $m$ satisfies 
\[
m \le  (n/16)^{4/3} 
\ \text{ and } m \text{ is sufficiently large.}
\] 
Then, there exists a random event $E_{3,1}(w)$ depending on $w$ and random event $E_{3,2}$ independent of $w$ with tail probabilities
\[
\bb{P}\left(  E_{3,1}^C (w)\right) \le \exp \left\{ - \tilde C n m^{1/4} \right\}
\ \text{and} \
\bb{P}\left(  E_{3,2}^C \right) 
\le 4 m^{\sqrt{2m}/8}  \exp \left\{ - C n \right\}
\] 
such than on $E_{3,1}(w) \cap E_{3,2}$, Inequality \eqref{eq: desired} holds. 
\end{theorem}
\begin{proof}
As in the previous two cases, we start with the representation \eqref{eq: top eigen}. Similarly to the spiky case, we want to select vector providing high magnitudes $w_j$. Since, $\norm{w}_\infty < m^{-1/4}$ its effect is not strong enough to compensate the rest of the entries with sufficiently high probability. However, it is possible to apply independence argument similar to the one used in the proof of the first case. 
Let us first introduce the set from which we will select $w_j$'s by defining the index set $I$ as
\[
I = I(w):= \left\{ k \in [m] \ \big| \ |w_k| > 1/\sqrt{2m} \right\}.
\]  
Using the bound on infinity norm, we obtain the following bound
\[
1 
= \sum_{k \in [m]} |w_k|^2 
= \sum_{k \in I} |w_k|^2 + \sum_{k \in [m] \backslash I} |w_k|^2
\le |I| \frac{1}{\sqrt{m}} + m \frac{1}{2m}
= \frac{1}{2} + \frac{|I|}{\sqrt{m}}.
\]
Hence, the cardinality of $I$ is bounded from below as $|I| \ge \sqrt m / 2$.
Since entries of $w$ are real, they either satisfy $w_j>0$ or $w_j\le0$. Consider $I_+$ and $I_-$, subsets of $I$ with all positive and negative entries, so that
\[
I_+ = I_+(w) := \{ j \in I, w_j >0 \} 
\ \text{ and } \
I_- = I_-(w) := \{ j \in I, w_j < 0 \} .
\]
Since $|I_+| + |I_-| = |I|$, one of the sets has at least half of the entries of $I$. Without loss of generality, let $ |I_+| \ge \ceil{ \ceil{\sqrt m/2} /2} \ge \ceil{\sqrt{m/4}}$. Otherwise, for the rest of the proof consider $v = -w$ satisfying $\norm{\q{A}^* v}_\infty = \norm{\q{A}^* w}_\infty$ with $|I_+(v)| \ge |I_-(v)|$. Finally, we select an subset $L$ of $I_+$ with cardinality $|L| = \ceil{\sqrt{2m}/16} \le \ceil{\sqrt{m/4}}$ with indices sorted in increasing order. \\
Now, we introduce a lower bound for representation \eqref{eq: top eigen} as 
\begin{align*}
\norm{\q{A}^* w}_\infty  
& = \max_{v \in \bb{C}^n, \norm{v}_2=1}  \left|  \sum_{k =1}^m w_k | a_k^*  v |^2 \right|
\ge \max_{j \in L}  \left|  \sum_{k =1}^m w_k | a_k^*  \tilde a_j |^2 \right| 
%& \ge \max_{j \in L}   \sum_{k =1}^m w_k | a_k^*  \tilde a_j |^2
\end{align*}
where $\q{GS} := \{ \tilde a_j \}_{j \in L}$ is the Gram-Schmidt orthogonalization of the vectors $\{ a_j \}_{j \in L}$ according to the order in $L$. With notation 
\[
S^4_j(w) := \sum_{k \in [m] \backslash L}^m w_k | a_k^*  \tilde a_j |^2,
\]
the bound can be further elaborated as
\begin{align}
\norm{\q{A}^* w}_\infty  
& \ge \max_{j \in L}   \sum_{k =1}^m w_k | a_k^*  \tilde a_j |^2
= \max_{j \in L} \left[  \sum_{k \in [m] \backslash L} w_k | a_k^*  \tilde a_j |^2 
+  \sum_{k \in L} w_k | a_k^*  \tilde a_j |^2 \right] \nonumber \\
& \ge  \max_{j \in L}  S^4_j(w) 
+ \min_{j \in L} \sum_{k \in L} w_k | a_k^*  \tilde a_j |^2
\ge  \max_{j \in L} S^4_j(w) 
+ \min_{j \in L}  w_j | a_j^*  \tilde a_j |^2  \nonumber \\
& \ge \max_{j \in L} S^4_j(w) + \frac{1}{\sqrt{2m}} \min_{j \in L} | a_j^*  \tilde a_j |^2, \label{eq: flat proof}
\end{align}
where we used properties of $L$ in the last two inequalities. \par 
We note the first term can be treated similarly to the $\max_{\ell \in [n]} S^1_\ell(w)$ in the proof of \Cref{thm: non-centered}. Let vectors in $\q{GS}$ to be fixed. Then, consider two sums $S^4_{j}$ and $S^4_{\ell}$ for $j,\ell \in L, j \neq \ell$.  The sums are independent of each other, since $ a_k^*  \tilde a_j$ is independent of $ a_r^*  \tilde a_\ell$ for $k,r \in [m] \backslash L, k \neq r$ and $ a_k^*  \tilde a_j$ is independent of $ a_k^*  \tilde a_\ell$ as a projections on orthogonal directions for all $k\in [m] \backslash L$ \cite{cochran_1934}. 
The summands within the single sum are independent due to the independence of the $a_k$'s and follow a $|\q{CN}(0,1)|^2$ distribution. Hence, we introduce a random event
\[
%E_{S^4_j,1}(w) := \left\{ |S^4_j(w) - \Ex S^4_j(w) | < \frac{n}{16\sqrt{2m}} \right\},
%\text{ and }
E_{S^4_j}(w) := \left\{ S^4_j(w) >  \E S^4_j(w) - \frac{n}{16\sqrt{2m}} \right\} 
= \left\{ S^4_j(w) >  \sum_{k \in [m] \backslash L} w_k - \frac{n}{16\sqrt{2m}} \right\}.
\]
The tail probability of $E_{S^4_j}(w)$ bounded from above by Proposition \ref{p: bernstein} and assumptions on $w$ and $m$ as
\begin{align*}
\bb{P} \left( E_{S^4_j}^C(w) \ \big| \ \q{GS} \right) 
& \le \exp \left\{ - \tilde C \min \left\{ \frac{n^2}{512 m \norm{w}_2^2}, \frac{n}{16  \sqrt{2m} \norm{w}_\infty} \right\} \right\} \\
& \le \exp \left\{ - \tilde C \min \left\{ \frac{n^2}{512 m }, \frac{n}{16 \sqrt{2} m^{1/4}} \right\} \right\} \\
& \le \exp \left\{ - \frac{\tilde C n}{m^{1/4}} \min \left\{ \frac{n}{16 m^{3/4} }, 1 \right\} \right\}
\le \exp \left\{ - \frac{\tilde C n}{m^{1/4}} \right\}.
\end{align*}
By integrating out random variables $\q{GS}$, we obtain bound for the unconditional probability
\[
\bb{P} \left( E_{S^4_j}^C(w) \right)
= \int_{\Omega_{\q{GS} } } \bb{P} \left( E_{S^4_j}^C(w) \ \big| \ \q{GS} \right)  d \mu_{\q{GS}} 
\le \exp \left\{ - \frac{\tilde C n}{m^{1/4}} \right\},  
\]
where $(\Omega_{\q{GS} },\mu_{\q{GS}})$ denotes probability space generated by $\q{GS}$. %Note that $\Ex S^4_j(w) = \sum_{k \in [m] \backslash L} w_k$, and hence under event $E_{S^4_j,1}(w)$ it holds that
%\[
%S^4_j(w) 
%>  \Ex S^4_j(w) - \frac{n}{16\sqrt{2m}}
%=  \sum_{k \in [m] \backslash L} w_k - \frac{n}{16\sqrt{2m}},
%\]
%so that $E_{S^4_j,1}(w) \subseteq E_{S^4_j,2}(w)$ and 
%\[
%\bb{P} \left( E_{S^4_j,2}^C(w) \right) \le \bb{P} \left( E_{S^4_j,1}^C(w) \right) \le 2 \exp \left\{ - \frac{\tilde C n}{m^{1/4}} \right\}.
%\]
Finally, we define a random event
\begin{equation}\label{eq: flat proof in1}
E_{3,1}(w) := \left\{ \max_{j \in L} S^4_j(w) >  \sum_{k \in [m] \backslash L} w_k - \frac{n}{16\sqrt{2m}} \right\} = \bigcup_{j \in k} E_{S^4_j}(w)
\end{equation}
with tail probability
\begin{align*}
\bb{P} \left( E_{3,1}^C(w) \right) 
& = \bb{P} \left( \bigcap_{j \in L} E_{S^4_j}^C(w) \right)
= \prod_{j \in L} \bb{P} \left( E_{S^4_j}^C(w) \right)
\le \left(   \exp \left\{ - \frac{\tilde C n}{m^{1/4}} \right\} \right)^{|L|} \\
& =  \left(  \exp \left\{ - \frac{\tilde C n}{m^{1/4}} \right\} \right)^{\ceil{\sqrt{2m} /16}}
\le  \exp \left\{ - \tilde C n m^{1/4} \right\}.
\end{align*}
Turning to the second term in Inequality \eqref{eq: flat proof}, $\min_{j \in L} |a_j^* \tilde a_j|$, we recall that $\{ \tilde a_k \}_{k \in L}$ is an orthonormal system obtained by Gram-Schmidt orthogonalization. Hence, $|a_j^* \tilde a_j|^2$ can be expressed by the Pythagorean theorem as
\[
|a_j^* \tilde a_j|^2 
= \norm{a_j}_2^2 - \sum_{k \in L, k \neq j} |a_j^* \tilde a_k|^2
= \norm{a_j}_2^2 - \sum_{k \in L, k < j} |a_j^* \tilde a_k|^2
=: \norm{a_j}_2^2 - S^5_j(L),
\]
where in the second equality we used that $a_j \in \spn\{ \tilde a_1, \ldots \tilde a_j\}$ by construction. In the proof of \Cref{thm: spiky}, we showed that for a random event
\[
E_{a_j} := \{|\norm{a_j}_{2}^2 - n| < n/2\} 
\ \text{ it holds that } \
\bb{P}\left( E_{a_j}^C \right) 
\le 2 \exp \left\{ - c n \right\}.
\]
The sum $S^5_j(L)$ is again a sum of independent $|\q{CN}(0,1)|^2$ distributed random variables when $\q{S}_j := \{ \tilde a_1, \ldots \tilde a_{j-1}\}$ are fixed. Hence, a random event
\[
E_{S^5_j}(L) := \left\{ |S^5_j(L) - \Ex S^5_j(L)| < n/4 \right\}
\] 
has tail probability 
\begin{align*}
\bb{P} \left( E_{S^5_j}^C(L) \ \big| \ \q{S}_j \right) 
& \le 2 \exp \left\{ - C \min \left\{ \frac{n^2}{16 |L| }, \frac{n}{4 } \right\} \right\} 
\le 2 \exp \left\{ - C \min \left\{ \frac{n^2}{2 \sqrt{2m} }, \frac{n}{4 } \right\} \right\} \\
& \le 2 \exp \left\{ - C n \min \left\{ \frac{n}{\sqrt{m} }, 1 \right\} \right\}
\le 2 \exp \left\{ - C n \right\},
\end{align*}
where in the last inequality we used that $n \ge 16 m^{3/4} \ge 16 \sqrt{m}$.
Again, by integrating out $\q{S}_j$, unconditional tail probability is
\[
\bb{P} \left( E_{S^5_j}^C(L) \right) \le 2 \exp \left\{ - C n \right\}
\]
Under $E_{a_j} \cap E_{S^5_j(L)}$ it holds that 
\begin{align*}
|a_j^* \tilde a_j|^2 
& =  \norm{a_j}_2^2 - S^5_j(L)
> \frac{n}{2} - \frac{n}{4} - \Ex S^5_j(L) 
= \frac{n}{4} - |L|  \\
&\ge \frac{n}{4} - \ceil{\sqrt{m}/4}
\ge \frac{n}{4} - \frac{\sqrt{m}}{2}
\ge \frac{n}{4} - \frac{n^{2/3}}{16^{2/3}2}
\ge \frac{n}{4} - \frac{n}{16^{1/2} 2}
= \frac{n}{8}.
\end{align*}
Then, a random event
\[
E_{\tilde a_j}(L) := \left\{ |a_j^* \tilde a_j|^2 > \frac{n}{8} \right\} \supseteq E_{a_j} \cap E_{S^5_j(L)}
\]
has tail probability bounded from above by a union bound as
\[
\bb{P}\left( E_{\tilde a_j}^C (L) \right) 
\le \bb{P}\left(  E_{a_j}^C \cup E_{S^5_j}^C (L) \right) 
\le \bb{P}\left(  E_{a_j}^C\right)  + \bb{P}\left( E_{S^5_j}^C (L) \right)  
\le 4 \exp \left\{ - C n \right\}.
\]
Consequently, a random event
\begin{equation}\label{eq: flat proof in2}
E_{\min}(L) := \left\{ \min_{j \in L} |a_j^* \tilde a_j|^2 > \frac{n}{8} \right\} = \bigcap_{j \in L} E_{\tilde a_j}(L) 
\end{equation}
has tail probability bounded from above by a union bound as
\[
\bb{P}\left( E_{\min}^C (L) \right) 
= \bb{P}\left(  \bigcup_{j \in L} E_{\tilde a_j}^C (L) \right) 
\le \sum_{j \in L} \bb{P}\left(  E_{\tilde a_j}^C \right) 
\le 4 |L| \exp \left\{ - C n \right\}.
\]
Since order in the exponent is $n \le m$, it is insufficient to apply the covering argument. However, $ \min_{j \in L} |a_j^* \tilde a_j|^2$ is independent of entries of $w$ and only depends on a choice of a index set $L$. The number of all possible choices of the subset $L$ is given by
\[
\binom{m}{|L|} 
\le \frac{m^{|L|}}{(|L|)!}
\le \frac{m^{\sqrt{2m}/8}}{(|L|)!}
\]
Thus, a random event
\[
E_{3,2} := \bigcap_{L \subset [m], |L|= \ceil{\sqrt{2m}/16}} E_{\tilde a_j}(L) 
\]
has tail probability
\begin{align*}
\bb{P}\left( E_{3,2}^C \right) 
& = \bb{P}\left(  \bigcup_{  L \subset [m], |L|= \ceil{\sqrt{2m}/16}  } E_{\min}^C (L) \right) 
\le \sum_{ \substack{ L \subset [m],\\ |L|= \ceil{\sqrt{2m}/16} } } \bb{P}\left(  E_{\min}^C(L) \right) \\
& \le 4 \frac{|L| m^{\sqrt{2m}/8} }{(|L|)!} \exp \left\{ - C n \right\}
\le 4 m^{\sqrt{2m}/8}  \exp \left\{ - C n \right\}
\end{align*}
The order of the exponent can be bounded from above as
\[
- C n + \frac{\sqrt{2m}}{8} \log m \le -16 C m^{3/4} + \frac{\sqrt{2m}}{8} \log m < 0
\]
for sufficiently large $m$. \par
Last step is to observe that due to inequalities $m^{3/4} \ge n/16$ and $\norm{w}_\infty \le m^{-1/4}$ it holds that
\begin{equation}\label{eq: flat proof in3}
\sum_{k \in L} w_k 
\le |L| m^{-1/4} 
= \left \lceil \frac{\sqrt{2m}}{16} \right \rceil m^{-1/4}
\le \frac{\sqrt{2m}}{8 m^{1/4}}
= \frac{\sqrt 2 m^{1/4}}{8}
=\frac{2 m^{3/4}}{8 \sqrt{2m}}
\le \frac{n}{ 64 \sqrt{2m}}.
\end{equation}
Returning to the Inequality \eqref{eq: flat proof}, under $E_{3,1}(w) \cap E_{3,2}$ in the view of Inequalities \eqref{eq: flat proof in1}, \eqref{eq: flat proof in2}, \eqref{eq: flat proof in3} and assumption on $w$, we obtain
\begin{align*}
\norm{\q{A}^* w}_\infty   
& \ge \max_{j \in L} S^4_j(w) + \frac{1}{\sqrt{2m}} \min_{j \in L} | a_j^*  \tilde a_j |^2 
> \sum_{k \in [m] \backslash L} w_k - \frac{n}{16\sqrt{2m}} + \frac{n}{8\sqrt{2m}} \\
& = \sum_{k =1}^m w_k - \sum_{k \in L} w_k + \frac{n}{16\sqrt{2m}}
> -\frac{n}{32 \sqrt{2m}}  - \frac{n}{ 64 \sqrt{2m}} + \frac{n}{16\sqrt{2m}} 
= \frac{n}{64 \sqrt{2m}},
\end{align*}
which concludes the proof.
\end{proof}

\subsubsection{Combining the results}\label{sec: joining results}
All three cases can be united as a single statement.
%\textcolor{red}{Here I use results for flat case when single side Bernstein is applied instead of two-sided and annoying $2^{\sqrt{2m}/8}$ is avoided. Proof to be fixed after discussion.}

\begin{corollary}\label{col: joined}
Let $w \in \bb{R}^m, \norm{w}_2= 1$. Suppose that the number of measurements $m$ satisfies 
\[
m \le \min\{ (n/16)^{4/3} , (n/8)^{8/7}, n^2 \}
\ \text{ and } m \text{ is sufficiently large.}
%(m^{3/4}+1) \log m < c m^{7/8}
\] 
Then, there exists a random event $E_{1}(w)$ depending on $w$ and random event $E_{2}$ independent of $w$ with tail probabilities
\[
\bb{P}\left(  E_{1}^C (w)\right) \le 2 \exp \left\{ -\tilde C n m^{1/8} \right\}, 
\] 
and
\[
\bb{P}\left(  E_{2}^C \right) 
\le 8 m^{\sqrt{2m}/8} \exp \left\{ - c n \right\} + 4 m^{m^{3/4}+1} \exp \left\{ -C m^{7/8} \right\}
\] 
such than on $E_{1}(w) \cap E_{2}$, Inequality \eqref{eq: desired} holds. 
\end{corollary}

\begin{proof}
The proof of the corollary is obtained by the comparison of the tail of the corresponding events in \Cref{thm: non-centered,thm: spiky,thm: flat}.
The random event $E_1(w)$ is either $E_{1,1}(w)$, $E_{2,1}(w)$ or $E_{3,1}(w)$ depending on $w$. Let $\tilde C$ be the smallest of the constants in the tail probabilities of these random events. Due to the condition $16 m^{3/4} \le n$, the following line of inequalities holds
\[
n m^{1/8} 
\le n m^{1/4} 
= n m^{3/4} / \sqrt{m}
\le n^2 / (16 \sqrt{m}) 
\le n^2 /\sqrt{m},
\]
and hence tail probability of the event $E_1(w)$ is bounded from above as
\[
\bb{P}\left(  E_{1}^C (w)\right) 
\le \max \left\{ \bb{P}\left(  E_{1,1}^C (w)\right), \bb{P}\left(  E_{2,1}^C (w)\right),\bb{P}\left(  E_{3,1}^C (w)\right)\right\}
\le 2 \exp \left\{ -\tilde C n m^{1/8} \right\}.
\] 
We set random event $E_2$ as an intersection of events $E_{2,2}$ and $E_{3,2}$ and its tail probability is bounded from above by De Morgan's law and union bound, that is
\begin{align*}
\bb{P}\left(  E_{2}^C \right) 
& = \bb{P}\left(  E_{2,2}^C \cup E_{3,2}^C \right) 
\le  \bb{P}\left(  E_{2,2}^C \right) +  \bb{P}\left(  E_{3,2}^C \right) \\
& \le 2\exp \left\{ - c n \right\} + 4 m^{m^{3/4}+1} \exp \left\{ -C m^{7/8} \right\}
+ 4 m^{\sqrt{2m}/8} \exp \left\{ - c n \right\} \\
& \le 8 m^{\sqrt{2m}/8} \exp \left\{ - c n \right\} + 4 m^{m^{3/4}+1} \exp \left\{ -C m^{7/8} \right\},
\end{align*}
where constant $c$ is chosen to be smaller of two constants in the corresponding exponentials. 
\end{proof}

Finally, we extend our results for a single vector $w$ to a uniform bound, which hold for all $w$ simultaneously and grant us the $S_1$-quotient property of the scaled measurement operator $\frac{1}{\sqrt m}\q A$. 

\begin{proof}[Proof of \Cref{thm: QP}]
In order to extend the obtained results for all $w, \norm{w}_2 =1$, we will consider an $\delta$-covering (net) $\q C$ of the unit ball $\q B := \{w \in \bb R^m, \norm{w}_2 =1 \}$, that is for each $w \in \q B$ exist vector $\tilde w$ in $\q C$ such that 
\[
\norm{w - \tilde w}_2 \le \delta.
\]
The radius $\delta$ is chosen to be $\frac{1}{192 \sqrt{2}m}$. % and this exact value is explained later in the proof. 
The covering $\q C$ is a finite subset of $\q B$ and its cardinality is bounded from above by Corollary 4.2.13 in \cite{Vershynin.2018} as 
\[
|\q C| \le (1 + 2/\delta)^m = (1 + 384\sqrt 2 m)^m \le (1 + 544.1 m)^m \le (545m)^m.  %= \exp\{ m \log (1 + 2 ... m)\}.
\] 
First, we note than under conditions $c_2  m^{7/8} \log 545 m \le n$ and $m$ being sufficiently large all inequalities $m \le (n/8)^{8/7}$, $m \le (n/16)^{4/3}$ and $m < n^2$ are satisfied and thus \Cref{col: joined} can be applied for $\tilde w \in \q C$. Further, Inequality \eqref{eq: desired} can be extended to hold for all $\tilde w \in \q C$ simultaneously. That is the tail probability of the random event $E_{\q C}$ defined as 
\[
E_{\q C} := \left\{ \norm{\frac{1}{\sqrt m} \q{A}^* \tilde w}_\infty \ge \frac{n}{64\sqrt{2}m} \text { for all } \tilde w \in \q C \right\} = \left( \bigcap_{\tilde w \in \q C} E_1(\tilde w) \right) \cap E_2,
\] 
is bounded from above by De Morgan's laws and union bound
\begin{align*}
\bb{P}\left(  E_{\q C}^C \right) 
& = \bb{P}\left(  \left( \bigcup_{\tilde w \in \q C} E_1(\tilde w)^C \right) \cup E_2^C \right) 
\le  \sum_{\tilde w \in \q C} \bb{P}\left( E_1(\tilde w)^C \right) +  \bb{P}\left(  E_{2}^C \right) \\
& \le 2 (545m)^m \exp \left\{ -\tilde C n m^{1/8} \right\} + \bb{P}\left(  E_{2}^C \right) \\
& \le 2 (545m)^m \exp \left\{ -\tilde C n m^{1/8} \right\} + 8 m^{\sqrt{2m}/8} \exp \left\{ - c n \right\} + 4 m^{m^{3/4}+1} \exp \left\{ -C m^{7/8} \right\},
\end{align*}
Note that if $c_2$ is selected to satisfy $c_2 > \tilde C^{-1}$, then there exists $\gamma_2 >0$ such that $c_2 = (1 + \gamma_2) \tilde C^{-1}$ and due to condition $c_2  m^{7/8} \log 545 m \le n$ the first term transforms as 
\[
%(545m)^m \exp \left\{ -\tilde C n m^{1/8} \right\} = 
\exp \left\{ m \log (545m) -\tilde C n m^{1/8} \right\}
\le \exp \left\{ - \gamma_2 m \log (545m) \right\}
\le \exp \left\{ -\gamma_2 m^{7/8} \right\}.
\]
Similarly, the second and the third terms are bounded from above as
\[
%m^{\sqrt{2m}/8} \exp \left\{ - c n \right\} =  
\exp \left\{ (\sqrt{2m}/8) \log m - c n \right\}
\le  \exp \left\{ (\sqrt{2m}/8) \log m - c_2 m^{7/8} \log 545m \right\}
\le \exp \left\{ -\gamma_2 m^{7/8} \right\},
\]
and
\[
%m^{m^{3/4}+1} \exp \left\{ -C m^{7/8} \right\} = 
\exp \left\{ m^{3/4} \log m +\log m -C m^{7/8} \right\}
\le \exp \left\{ -\gamma_2 m^{7/8} \right\},
\]
respectively, for a sufficiently large $m$. Thus, the tail probability of $E_{\q C}$ is bounded from above by
\[
\bb{P}\left(  E_{\q C}^C \right)  \le 14 \exp \left\{ -\gamma_2 m^{7/8} \right\}.
\]

Final step is to use the properties of the covering under event $E_{\q C}$ to derive the uniform bound. Let $w \in \q B$ be arbitrary vector. Then, there exists $\tilde w \in \q C$ such that $\norm{w - \tilde w}_2 \le \delta$ and we bound the spectral norm $\norm{\frac{1}{\sqrt m} \q{A}^* w}_\infty$ as
\[
\norm{\frac{1}{\sqrt m} \q{A}^* w}_\infty 
= \norm{\frac{1}{\sqrt m} \q{A}^* (w - \tilde w + \tilde w)}_\infty
\ge \norm{\frac{1}{\sqrt m} \q{A}^* \tilde w}_\infty - \norm{\frac{1}{\sqrt m} \q{A}^* (w - \tilde w)}_\infty
\]
Under $E_{\q C}$ it holds that $\norm{\frac{1}{\sqrt m} \q{A}^* \tilde w}_\infty \ge \frac{n}{64\sqrt{2}m}$ and hence what remains is an upper bound for the second term. Using Lemma \ref{l: dual operator}, we bound it as
\begin{align*}
\norm{\frac{1}{\sqrt m} \q{A}^* (w - \tilde w)}_\infty 
& = \norm{ \frac{1}{\sqrt m} \sum_{k=1} (w - \tilde w)_k a_k a_k^*}_\infty
\le \frac{1}{\sqrt m} \sum_{k=1} |(w - \tilde w)_k | \norm{a_k a_k^*}_\infty \\
& = \frac{1}{\sqrt m} \sum_{k=1} |(w - \tilde w)_k | \norm{a_k}_2^2
\le \frac{1}{\sqrt m} \norm{w - \tilde w}_1  \max_{k \in [m]} \norm{a_k}_2^2 \\
& \le \norm{w - \tilde w}_2 \frac{3n}{2} 
\le \frac{3n \delta}{2}
= \frac{n}{128\sqrt{2}m},
\end{align*}
where we used that by construction
\[
E_{\q C} \subseteq E_2 \subseteq E_{2,2} \subseteq E_{\norm{\cdot}} = \left\{ \max_{j \in [m]} \norm{a_j}_{2}^2 < \frac{3n}{2} 
\text{ and } 
\min_{j \in [m]} \norm{a_j}_{2}^2 > \frac{n}{2}
\right\}.
\]
Thus, we obtain that under $E_{\q C}$ for all $w$ such that $\norm{w}_2=1$  it holds that
\[
\norm{\frac{1}{\sqrt m} \q{A}^* w}_\infty  
\ge \norm{\frac{1}{\sqrt m} \q{A}^* \tilde w}_\infty - \norm{\frac{1}{\sqrt m} \q{A}^* (w - \tilde w)}_\infty
\ge \frac{n}{64\sqrt{2}m} - \frac{n}{128\sqrt{2}m} 
= \frac{n}{128\sqrt{2}m}.
\]
The final term can be split as
\[
\frac{n}{128\sqrt{2}m} = \left(\frac{\kappa}{128\sqrt{2}}\sqrt{\frac{n}{\kappa m}} \right) \sqrt{\frac{n}{\kappa m}}.
\]
Then, by \Cref{thm: equivalent qp}, the scaled measurement operator $\frac{1}{\sqrt m}\q A$ possess the $S_1$-quotient property with constant $\frac{128 \sqrt {2}}{\kappa} \sqrt{\kappa m/n}$ and rank $\kappa m / n$ relative to the norm $\norm{\cdot}_2$.
\end{proof}

\subsection{Proof of \Cref{col: final}}
\Cref{col: final} is a consequence of applying \Cref{thm: quotient+NSP to optimality}. 

\begin{proof}[Proof of \Cref{col: final}]

	This proof is essentially joins the results. By \Cref{thm: NSP} and the choice of rank $r_* = c_1^{-1}\rho^2 m/n$, the scaled measurement operator $\frac{1}{\sqrt m}\q A$ satisfies the $S_2$-rank robust null space property of order $r_*$ with constants $0 < \rho < 1$ and $\tau>0$ relative to the $\norm{\cdot}_2$ on $\bb R^{m}$ with probability at least $1 - e^{-\gamma_1 m}$. By \Cref{thm: QP} with $\kappa = c_1^{-1}\rho^2$, $\frac{1}{\sqrt m}\q A$ posses the $S_1$-quotient property with constant $\frac{128\sqrt {2} c_1 }{\rho^2} \sqrt{r_*}$ and rank $r_*$ relative to the norm  $\norm{\cdot}_2$ on $\bb R^{m}$ with probability at least $1 - 14e^{-\gamma_2 m^{7/8}}$. Thus, by union bound, it satisfies both properties simultaneously with probability at least
	\[
	1 - e^{-\gamma_1 m} - 14e^{-\gamma_2 m^{7/8}}
	\ge 1 - e^{-\gamma m} - 14e^{-\gamma m^{7/8}}
	= 1  - 15e^{-\gamma m^{7/8}},
	\]
	with $\gamma := \min\{\gamma_1, \gamma_2\}$.
	By \Cref{thm: Kabanava3.1} the reconstruction program $\Delta$ admits the error bound \eqref{eq: requested rec guar} in the noiseless case ($\eta =0$).
	Hence, by \Cref{thm: quotient+NSP to optimality} we obtain
	\[
	\norm{X - \Delta \left( \frac{1}{\sqrt m}\q{A} X + \frac{1}{\sqrt m}w \right) }_{p} \le \frac{ D_1} { r^{ 1 - \frac{1}{p} } } \norm{X_{r}^{c} }_{1} + (D_2 \frac{128 \sqrt {2} c_1}{\rho^2} \sqrt{r_*} + D_3 )r_*^{ \frac{1}{p} - \frac{1}{2} } \frac{\norm{w}_2}{\sqrt m}.
	\]
	The last term can be bounded from above as
	\[
	(D_2 \frac{128 \sqrt {2} c_1}{\rho^2} \sqrt{r_*} + D_3 )r_*^{ \frac{1}{p} - \frac{1}{2} }
	\le (D_2 \frac{128 \sqrt {2}c_1}{\rho^2} \sqrt{r_*} + D_3 \sqrt{r_*})r_*^{ \frac{1}{p} - \frac{1}{2} }
	=: D_4 r_*^{ \frac{1}{p} },
	\]
	since $r_* \ge 1$. Note that constraint of reconstruction program $\Delta$ rescales error in the following way 
	\[
	\norm{\frac{1}{\sqrt m} \q A Z -  \frac{1}{\sqrt m}\q{A} X + \frac{1}{\sqrt m}w} = 0
	\text{ is equivalent to}
	\norm{\q A Z - \q{A} X +  w } = 0,
	\] 
	and therefore 
	\[
	\norm{X - \Delta \left( \q{A} X + w \right) }_{p} \le \frac{ D_1} { r^{ 1 - \frac{1}{p} } } \norm{X_{r}^{c} }_{1} + \frac{D_4 r_*^{ \frac{1}{p} } \norm{w}_2}{\sqrt m}.
	\]
	This concludes the proof for $\Delta$. The proof for $\Delta_{2,\eta}$ is analogous.
%The proof for $\Delta_{2,\nu}$ is obtained by applying part B) of \Cref{thm: quotient+NSP to optimality}.
%	The same statement applies to reconstruction map $\Delta_{2,\mu}^{\text{Lasso}}$ if usage of \Cref{thm: Kabanava3.1} is replaced by \Cref{col: Lasso noiseless}. The scaling in this case affects the choice of parameter $\mu$, that is
%	\[
%	\argmin_{Z \in \q H_n} \norm{\frac{1}{\sqrt m}\q{A} Z - \frac{1}{\sqrt m}\q A X - \frac{1}{\sqrt m} w }_{q} + \mu \norm{Z}_{1}
%	=
%	\argmin_{Z \in \q H_n} \norm{\q{A} Z - \q A X - w}_{q} + \sqrt m \mu \norm{Z}_{1}.
%	\]
%	Since $\mu \le \frac{(1 +\rho)^2}{2 \tau (3 + \rho)\sqrt{r_* m}}$, then $\sqrt m \mu \le \frac{(1 +\rho)^2}{2 \tau (3 + \rho)\sqrt{r_*}}$ and, thus, \Cref{col: Lasso noiseless} applies.
\end{proof}

\subsection{Proof of \Cref{thm: quotient+NSP to optimality}}
The proof of the \Cref{thm: quotient+NSP to optimality} closely follows the steps of its sparse recovery analogue Theorem 11.12 in \cite{Foucart.2013}.
\begin{proof}[Proof of \Cref{thm: quotient+NSP to optimality}]
Our goal is to bound $\norm{X - \q R_\eta (\q A X + w)}_q$ from above in terms of $\norm{X_r^c}$ and $\norm{w}$ for all $X \in \q{H}_{n}$ and $w \in \bb{R}^{m}$.
In the first case, when $\norm{w} \le \eta$, by assumption we have
\[
\norm{X - \q R_{\eta} \left( \q{A} X + w \right) }_p \le \frac{ D_1 } { r_*^{1 - 1/p} } \norm{X_{r_*}^{c} }_1 + \frac{2 \tau (3 + \rho)}{ 1 - \rho } r_*^{1/p - 1/2} \max\{ \norm{w},\eta \},
\]
with $D_1:=\frac{ 2 (1 + \rho)^2 } { (1 - \rho) }$. 

In the second case, let $\norm{w} > \eta$ and consider decomposition
\[
w = \frac{\eta}{\norm{w}} w + \left(1 - \frac{\eta}{\norm{w}}\right) w =: w_\eta + v. 
\]
Note that $\norm{w_\eta} = \eta$ and $\norm{v} = \norm{w} - \eta$. Furthermore, by the $S_1$-quotient property there exists matrix $U$ such that $\q A U = v$. Thus, by triangle inequality, it holds that
\begin{align*}
\norm{X - \q R_{\eta} \left( \q{A} X + w \right) }_p
& = \norm{X - \q R_{\eta} \left( \q{A} X + w_\eta + \q A U \right) + U - U}_p \\
& \le \norm{X + U - \q R_{\eta} \left( \q{A} (X + U) + w_\eta \right) }_p + \norm{U}_p.
\end{align*}
For matrix $X+U$ and noise $w_\eta$ the first case applies and one obtains the bound
\[
\norm{X - U \q  + \q R_{\eta} \left( \q{A}(X + U) + w_\eta  \right) }_p \le \frac{D_1}{r_*^{1 - 1/p}} \norm{ (X +U)_{r_*}^{c}}_1 + \frac{2 \tau (3 + \rho)}{ 1 - \rho } r_*^{1/p - 1/2} \eta.
\]
Further the first norm bounded using the best rank approximation properties \eqref{eq: rank proj} of \mbox{$(X + U)_{ r_*}^c$} and triangle inequality, that is
\[
\norm{(X + U)_{ r_*}^c}_1 
=\min_{\substack{ Z \in \bb{C}^{n \times n}, \\ \rank(Z) \le r_*} } \norm{X + U- Z}_1
\le \norm{U}_1 + \min_{\substack{ Z \in \bb{C}^{n \times n}, \\ \rank(Z) \le r_*}}  \norm{X- Z}_1
= \norm{U}_1 + \norm{X_{ r_*}^c}_1.
\]
Consequently, we obtain
\begin{equation}\label{eq: pre optimality}
\norm{X - \q R_\eta(\q A X + w)}_p
\le \frac{D_1}{r_*^{1 - 1/p}} \norm{X_{ r_*}^c}_1 
+ \frac{D_1}{r_*^{1 - 1/p}} \norm{U}_1  
+ \norm{U}_p
+ \frac{2 \tau (3 + \rho)}{ 1 - \rho } r_*^{1/p - 1/2} \eta.
\end{equation}
The first term is already a finalized. The second term is bounded by the $S_1$-quotient property as
\[
\frac{D_1}{  r_*^{1 - 1/p}} \norm{U}_1  
\le D_1 d  r_*^{1/p - 1/2}\norm{v},
\] 
and hence, only the third term remains.  
By Lemma 3.1 in \cite{Kabanava.2016}, it holds that 
\[
\norm{U_{ r_*}^c}_p \le \frac{\norm{U}_1}{ r_*^{1 -1/p} }.
\]
On the other hand, by H{\"o}lder's inequality and $S_2$-robust rank null space property, it follows that
\begin{align*}
\norm{U_{ r_*}}_p 
& \le  r_*^{1/p - 1/2} \norm{U_{ r_*}}_2
\le \frac{\rho  r_*^{1/p - 1/2} }{  r_*^{1 - 1/2}} \norm{U_{ r_*}^c}_1 + r_*^{1/p - 1/2} \tau \norm{\q A U} \\
& \le \frac{\rho}{  r_*^{1 - 1/p}} \norm{U}_1 + r_*^{1/p - 1/2} \tau \norm{v}.
\end{align*}
Thus, by triangle inequality and the $S_1$-quotient property we obtain
\begin{align*}
\norm{U}_p 
& \le \norm{U_{r_*}}_p  + \norm{U_{r_*}^c}_p
\le  \frac{\norm{U}_1}{ r_*^{1 -1/p} } + \frac{\rho}{ r_*^{1 - 1/p}} \norm{U}_1 + r_*^{1/p - 1/2} \tau \norm{v} \\
& \le \frac{(1 +\rho) d \sqrt{r_*} }{ r_*^{1 - 1/p}} \norm{v} + r_*^{1/p - 1/2} \tau \norm{v}
= ((1 + \rho)d  + \tau)  r_*^{1/p-1/2} \norm{v}.
\end{align*}
Next, we combine established bounds for the second and third terms with the inequality \eqref{eq: pre optimality}, so it holds that
\[
\norm{X - \q R_\eta(\q A X + w)}_p
\le \frac{D_1}{r_*^{1 - 1/p}} \norm{X_{ r_*}^c}_1  + (D_2 d + \tau) r_*^{1/p-1/2} \norm{v} + \frac{2 \tau (3 + \rho)}{ 1 - \rho } r_*^{1/p - 1/2} \eta,
\]
where $D_2 = D_1 + (1 + \rho)$.
Recall that $\norm{v} = \norm{w} - \eta$. Then, we can unite coefficients by $\eta$ as
$d_1 := \frac{2 \tau (3 + \rho)}{ 1 - \rho } - (D_2 d + \tau)$. If $d_1 \le 0 $ then the $\eta$-term can be neglected. Otherwise we bound $\eta < \norm{w} = \max\{\norm{w}, \eta\}$ and obtain 
\[
\norm{X - \q R_{\eta} \left( \q{A} X + w \right) }_p \le \frac{D_1}{r_*^{1 - 1/p}}\norm{X_{r_*}^{c} }_1 + d_2 r_*^{1/p - 1/2} \max\{\norm{w}, \eta\},
\]
where $d_2:=  (D_2 d + \tau) + \max\left\{d_1, 0 \right\}$. 
Finally, we join two cases and select coefficient by $\max\{\norm{w}, \eta\}$ as maximum of $d_2$ and  $\frac{2 \tau (3 + \rho)}{ 1 - \rho }$. Depending on the sign of $d_1$, $d_2$ is either $D_2 d + \tau$ or $\frac{2 \tau (3 + \rho)}{ 1 - \rho }$. Hence, it holds that
\[
\max\left\{d_2, \frac{2 \tau (3 + \rho)}{ 1 - \rho } \right\}
\le \max\left\{D_2 d + \tau, \frac{2 \tau (3 + \rho)}{ 1 - \rho } \right\}
\le D_2 d + \tau + \frac{2 \tau (3 + \rho)}{ 1 - \rho }
=: D_2 d + D_3,
\]
and desired inequality
\[
\norm{X - \q R_\eta(\q A X + w)}_p
\le \frac{D_1}{r_*^{1 - 1/p}} \norm{X_{ r_*}^c}_1  + (D_2 d + D_3) r_*^{1/p-1/2} \max\{\norm{w}, \eta\}
\]
holds.
At last we note that for all  $r \le  r_*$ it holds that $1/ r_* \le 1/r$ and $\norm{X_{ r_*}^c}_1 \le \norm{X_{r}^c}_1$, and thus
\[
\norm{X - \q R_\eta (\q A X + w)}_p
\le \frac{D_1}{r^{1 - 1/p}} \norm{X_r^c}_1  + (D_2 d + D_3) r_*^{1/p-1/2} \norm{w},
\]
which concludes the proof.
\end{proof}

\section{Conclusion and Future work}
In this paper, we established the $S_1$-quotient property for Gaussian rank-one measurements, allowing us to derive noise-blind guarantees for the recovery low-rank matrices from such measurements. We expect that the proof technique we used can help with the development of similar results for other applications with more structure such as randomized blind deconvolution. An addition difficulty of this scenario is that even for noise-aware nuclear norm minimization recovery guarantees analogous to those discussed in \Cref{sec:previousresults} are only possible with additional dimensional scaling factors \cite{KS19}.
%A significant part of our contribution is the new proof technique, 
%\CK{which we expect to help in the development of similar results for other applications such as random blind deconvolution, and beyond that, to even more structured measurements}.
  
For follow-up work, it will be of interest to study to which extent noise-blind recovery guarantees from rank-one measurements via a square-root Lasso as studied in \cite{GaiffasKlopp17} for random noise can also be established under a robust null space property in analogy to the results of \cite{PetersenJung20} that have been achieved for sparse recovery.

Another interesting line of research concerns the refinement of our proofs to get an improvement on the maximum number of measurements admissible in \Cref{thm: QP}. Our numerical experiments in \Cref{sec: numerics} and contributions regarding the quotient property for a different measurement scenario \cite{Candes.2011} suggest that the optimal bound of a form comparable to $m \le c n^2 / \log^\alpha (m/n)$, potentially leaving room for further improvement of \Cref{thm: QP}.

\section*{Acknowledgements}
OM was partially supported by the Helmholtz Association within the project \mbox{Ptychography 4.0}. FK acknowledges support by the German Science Foundation DFG in the context of an Emmy Noether junior research group (\mbox{project KR 4512/1-1}). The authors gratefully acknowledge the Leibniz Supercomputing Centre (LRZ) for providing computing time on their Linux cluster.

\appendix 
\section{Appendix}

\subsection{Proof of \Cref{thm: equivalent qp}} \label{app:proof:aux:dual}

Let $\q{A}$ satisfy the $S_q$-quotient property with constant $d$ and rank $r_*$ relative to $\norm{\cdot}$. For $w \in \bb{R}^{m}$ by definition of dual norm we have that 
\[
\norm{w}_{*} = \sup_{\norm{v} = 1 }\langle v ,w \rangle.
\]
Since $\bb{R}^m$ is finite-dimensional space and the unit ball is a compact set, there exists element $y \in \bb{R}^{m}, \norm{y} = 1$ such that supremum is attained $\norm{w}_{*} =  y^* w$.
By definition of the quotient property, $y = \q{A}U$ for some $U \in \q H_n$ with 
\[
\norm{U}_{q} \le d r_{*}^{1/q - 1/2}\norm{y} = d r_{*}^{1/q - 1/2}.
\] 
Next, we apply definition of the adjoint operator and H\"{o}lder's inequality for Schatten norms and obtain
\[
\norm{w}_{*}
= \langle y ,w \rangle
= \langle \q{A}U,w \rangle
= \langle U,\q{A}^{*}w \rangle_F
\le \norm{U}_{q} \norm{\q{A}^{*}w}_{q^{*}}
\le d r_{*}^{1/q - 1/2}\norm{\q{A}^{*}w}_{q^{*}},
\]
where $q^{*}$ is H\"{o}lder dual of $q$ and equals to $q/(q-1)$. It finishes the proof in one direction.

To show that reverse statement is true, let us assume that inequality \eqref{eq: qp in dual norms} holds. We start with case $q>1$. Recall that $q^{*} = q/(q-1)$.
If $w = 0$, then $U$ from the definition of the $S_q$-quotient property can be set as zero matrix. Let $w \in \bb R^{m} \backslash \{\bar{0}\}$ and set $U \in \q H_n$ as 
\[
U = \Delta_q(w) := \argmin_{Z \in \q H_n} \{~ \norm{Z}_{q} ~s.t. ~\q{A}Z = w ~\}
\]
The existence of the feasible point can be justified with proof by contradiction.

Suppose that there exists $v \in \bb R^{m}$ such that for all $Z \in \q H_n$ it holds that $\q{A}Z \neq v$.
This means that the measurement operator $\q{A}: \q H_n \to \bb R^{m}$ is not surjective, i.e., $\dim (\Ran \q{A}) < m$ and therefore $\dim \left((\Ran \q{A})^{\perp}\right)> 0$.
Consider now the adjoint operator $\q{A}^*: \bb R^{m} \to \q H_n$ of $\q{A}$.
Since $\q A$ is linear operator between finite-dimensional spaces, it holds that $(\Ran \q{A})^{\perp} = \ker (\q{A}^*)$ and this means that 
\[
\dim\left(\ker (\q{A}^*) \right) \ge 1
\]
and, therefore, there exists $\beta \in \bb R^{m}$, $\beta \neq 0$ such that $\q A^* \beta = 0$.
So, $\q{A}^{*} \beta$ is a zero matrix. Using Inequality \eqref{eq: qp in dual norms} we obtain a contradiction
\[
0 < \norm{\beta}_{*} \le d r_{*}^{1/q - 1/2}\norm{\q{A}^{*} \beta }_{q^{*}} = d r_{*}^{1/q - 1/2}\norm{0}_{q^{*}} = 0.
\]
Thus, for all $w \in \bb R^{m}$ optimization problem $\Delta_q(w)$ is feasible and $U = \Delta_q(w)$ is well-defined.

Fix $V \in \ker{\q{A}}$. Define $\tau := t e^{i \theta}$, where $\theta \in [0,2 \pi)$ and $t>0$ small enough to have $U + \tau V \neq 0$. We can present $U + \tau V$ using the singular value decomposition.
\[
U + \tau V = \sum_{j=1}^{n} \sigma_j(U + \tau V) u_j v_j^{*},
\]
Using the same vectors $u_j$ and $v_j$, $j \in [n]$, we define
\[
W^{\tau} := \frac{1}{\norm{U + \tau V}_{q}^{q-1}} \sum_{j=1}^{n} \sigma_j(U + \tau V)^{q-1} u_j v_j^{*}.
\]
Note that
\begin{align*}
\norm{W^{\tau}}_{q^{*}}
& = \left( \sum_{j=1}^{n} \left( \frac{\sigma_j(U + \tau V)}{\norm{U + \tau V}_{q}} \right)^{(q-1)q^{*}} \right)^{1/q^{*}}
= \left( \frac{1}{\norm{U + \tau V}_{q}^{q}} \sum_{j=1}^{n} \left( \sigma_j(U + \tau V) \right)^{q} \right)^{\frac{q-1}{q}} \\
& = \left( \frac{\norm{U + \tau V}_{q}^{q}}{\norm{U + \tau V}_{q}^{q}} \right)^{\frac{q-1}{q}}
= 1.
\end{align*}
Using that both $u_j$ and $v_j$, $j \in [n]$ form a basis in $\bb C^{n}$ we obtain 
\begin{align*}
\langle W^{\tau}, U + \tau V \rangle_{F}
&= \frac{1}{\norm{U + \tau V}_{q}^{q-1}} \sum_{j=1}^{n} \left( \sigma_j (U + \tau V) \right)^{q-1} \langle u_j v_j^{*}, U + \tau V \rangle_{F} \\
& = \frac{1}{\norm{U + \tau V}_{q}^{q-1}} \sum_{j=1}^{n} \left( \sigma_j (U + \tau V) \right)^{q-1} \sum_{k=1}^{n} \sigma_k (U + \tau V) \langle u_j v_j^{*}, u_k v_k^{*} \rangle_{F} \\
& = \frac{1}{\norm{U + \tau V}_{q}^{q-1}} \sum_{j=1}^{n} \sum_{k=1}^{n} \left( \sigma_j (U + \tau V) \right)^{q-1} \sigma_k (U + \tau V)  \langle u_j v_j^{*}, u_k v_k^{*} \rangle_{F} \\
& = \frac{1}{\norm{U + \tau V}_{q}^{q-1}} \sum_{j=1}^{n} \sum_{k=1}^{n} \left( \sigma_j (U + \tau V) \right)^{q} \langle u_j, u_k \rangle \langle v_k, v_j \rangle \\
&= \frac{1}{\norm{U + \tau V}_{q}^{q-1}} \sum_{j=1}^{n} \left( \sigma_j (U + \tau V) \right)^{q} 
= \frac{\norm{U + \tau V}_{q}^{q}}{\norm{U + \tau V}_{q}^{q-1}} 
= \norm{U + \tau V}_{q}.
\end{align*}
%where the Frobenius inner product disappears due to properties of the singular value decomposition, as can be seen from
%$$
%\langle u_j v_j^{*}, u_i v_i^{*} \rangle_{F}
%= \tr( (u_j v_j^{*})^{*} u_i v_i^{*} )
%= \tr( v_j u_j^{*} u_i v_i^{*} )
%= \langle u_j , u_i \rangle \langle v_j , v_i \rangle
%= \delta_{ij}.
%$$
Out next goal is to take a limit $\tau \to 0$, but first we need to justify its existence. By Weil's Theorem \cite[Theorem III.2.1]{Bhatia.1997}, we have that for all $j \in [m]$ 
\[
0 \le | \sigma_j(U + \tau V) - \sigma_j(U) |
\le \sigma_j(\tau V)
= \norm{\tau V}_{\infty}
= |\tau| \norm{V}_{\infty}
\to 0, \text{ as } \tau \to 0.
\]
Therefore, $\sigma_j(U + \tau V)$ converges to $\sigma_j(U)$, for $\tau \to 0$ and, consequently, 
\[
W^{\tau} \to \frac{1}{\norm{U}_{q}^{q-1}} \sum_{j=1}^{n} \sigma_j(U)^{q-1} u_j v_j^{*} := W, \quad \tau \to 0.
\]
Firstly, due to assumption $q>1$, $W$ is independent of $V$.
Secondly, it preserves properties of $W^{\tau}$, namely
\begin{equation}\label{eq: equivalent qp proof 1}
\norm{W}_{q^{*}} = 1 \text{ and } \langle W,U \rangle_F = \norm{U}_{q}.
\end{equation}
Recall that $U$ is solution of $\Delta_q(w)$ and $U + \tau V$ is a feasible point since $V \in \ker \q A$ and hence $\norm{U}_{q} \le \norm{U + \tau V}_{q}$. Together with H\"{o}lder's inequality it yields that
\[
\RE \langle W^{\tau}, U \rangle_F
\le |\langle W^{\tau}, U \rangle_F|
\le \norm{U}_{q} \norm{ W^{\tau} }_ {q^{*}}
= \norm{U}_{q}
\le \norm{U + \tau V}_{q}
= \RE \langle W^{\tau}, U + \tau V \rangle_F.
\]
Consequently, it holds that
\[
0 \le \frac{1}{t} \RE \langle W^{\tau}, \tau V \rangle_F 
= \frac{1}{t}\RE \langle W^{\tau}, t e^{i \theta} V \rangle_F 
= \RE \left( e^{-i \theta} \langle W^{\tau},  V \rangle_F \right). 
\]
Taking limit $t \to 0$ implies $\tau \to 0$ and we obtain that for all $\theta \in [0,2\pi]$ inequality 
\[
\RE \left( e^{-i \theta} \langle W,  V \rangle_F \right) \ge 0
\] 
is true. It is only possible if $\langle W,  V \rangle_F = 0$.
Again, recall that $V$ was an arbitrary matrix from $\ker \q{A}$ and we proved that $\langle W,  V \rangle_F = 0$. It implies that $W \in \left( \ker \q{A} \right)^{\bot} = \Ran \q{A}^{*}$. Thus, there exists $y \in \bb R^{m}$ so that $W = \q{A}^{*} y$. 
By Inequality \eqref{eq: qp in dual norms} and Equalities \eqref{eq: equivalent qp proof 1},  it holds that
\[
\norm{y}_{*}
\le d r_{*}^{1/q - 1/2}\norm{\q{A}^{*}y}_{q^{*}}
= d r_{*}^{1/q - 1/2}\norm{W}_{q^{*}}
= d r_{*}^{1/q - 1/2}.
\]
Therefore, by H\"{o}lder's inequality and Equalities \eqref{eq: equivalent qp proof 1} we obtain
\[
\norm{U}_{q}
= \langle W, U \rangle
= \langle \q{A}^{*}y, U \rangle
= \langle y, \q{A} U \rangle
= \langle y, w \rangle
\le \norm{y}_{*} \norm{w}
\le d r_{*}^{1/q - 1/2} \norm{w},
\]
what establishes the quotient property for $q>1$.

For the case $q=1$, consider a sequence $(q_j)_{j \in \mathbb{N}}$ such that $q_j>1$ for all $j \in \bb{N}$ and $ q_j \to 1$ for $j \to \infty$. 
We first note that 
%for two parameter $\tilde{p}$ and $\tilde{q}$ satisfying $1 \le \tilde{q} \le \tilde{p} \le \infty$ and for any $Z \in \q H_n$ it holds, that 
%Application of this inequality for $\tilde{q} = \infty$ and $\tilde{p} = q_j*$ grants us  
\begin{equation}\label{eq: equivalent qp proof 2}
%\norm{\q{A}^{*}w}_{\infty} \le \norm{\q{A}^{*}w}_{q_j^*} \text{ for all } j \in \bb{N}.
\text{for all } 1 \le \tilde{q} \le \tilde{p} \le \infty, \text{ for all } Z \in \q H_n \text{ it holds that } \norm{Z}_{\tilde{q}} \le \norm{Z}_{\tilde{p}}.
\end{equation}
Combining Inequality \eqref{eq: qp in dual norms} with \eqref{eq: equivalent qp proof 2} results in 
\[
\norm{w}_{*} 
\le d r_{*}^{1/2}\norm{\q{A}^{*}w}_{\infty} 
= (d r_{*}^{1- 1/q_j}) r_{*}^{1/q_j - 1/2}\norm{\q{A}^{*}w}_{\infty}
\le (d r_{*}^{1- 1/q_j}) r_{*}^{1/q_j - 1/2}\norm{\q{A}^{*}w}_{q_j^*}.
\]
Since for all $j \in \bb N$ it holds that $q_j>1$, above inequality is equivalent to the $S_{q_j}$-quotient property of $\q A$ and hence there exist matrix $U_j$ such that 
\[
\q A U_j = w  
\text{ and }  
\norm{U_j}_{q_j} \le (d r_{*}^{1- 1/q_j}) r_{*}^{1/q_j - 1/2} \norm{w}
= d r_{*}^{1/2} \norm{w}
\] 
By \eqref{eq: equivalent qp proof 2}, sequence $\{U_j\}_{j \in \bb N}$ is bounded in the $S_\infty$-Schatten norm, that is
\[
\norm{U_j}_{\infty} \le \norm{U_j}_{q_j} \le  d r_{*}^{1/2} \norm{w}, \quad \text{ for all } j \in \bb N.
\]
Space $\q H_n$ is a finite-dimensional space and thus, there exists a convergent in $S_\infty$-norm subsequence $(U_{j_k})_{k \in \mathbb{N} }$, such that $U_{j_k} \to U$ when $k \to \infty$.  Consequently, we obtain that $\q{A} U = w$.
%and the last step is to prove inequality $\norm{U}_{q} \le d r_{*}^{1/q - 1/2} \norm{w}$.
Moreover, since all norms in finite-dimensional spaces are equivalent $U_{j_k}$ converges to $U$ in $S_1$ norm as well. 

Last step is to prove that $\norm{U}_{1} = \lim_{k \to \infty} \norm{U_{j(k)}}_{q_{j(k)}}$.
Using continuity of the $S_p$ norms in parameter $p$, that is $\norm{U}_{1} = \lim_{k \to \infty} \norm{U}_{q_{j_k}}$ it follows that
\begin{align*}
\norm{U}_{1}
& = \lim_{k \to \infty} \norm{U}_{q_{j(k)}}
\le \lim_{k \to \infty} \norm{U - U_{j(k)} }_{q_{j(k)}}
+ \lim_{k \to \infty} \norm{U_{j(k)}}_{q_{j(k)}} \\
& \le \lim_{k \to \infty} \norm{U - U_{j(k)} }_{1}
+ \lim_{k \to \infty} \norm{U_{j(k)}}_{q_{j(k)}}
= \lim_{k \to \infty} \norm{U_{j(k)}}_{q_{j(k)}},
\end{align*}
where we used triangle inequality and \eqref{eq: equivalent qp proof 2}.
%By dropping everything in between, we have that
%$$
%\Norm{U}{1} \le \lim_{k \to \infty} \Norm{U_{j(k)}}{q_{j(k)}}.
%$$
On the other hand,
\[
\lim_{k \to \infty} \norm{U_{j(k)}}_{q_{j(k)}}
\le \lim_{k \to \infty} \norm{U_{j(k)}}_{1}
\le \lim_{k \to \infty} \norm{U_{j(k)} - U}_{1}
+ \lim_{k \to \infty} \norm{U}_{1}
= \norm{U}_{1}.
\]
Thus,
\[
\norm{U}_{1} = \lim_{k \to \infty} \norm{U_{j(k)}}_{q_{j(k)}} \le d r_{*}^{1/2} \norm{w},
\]
which concludes the proof.

\subsection{Proof of \Cref{l: dual operator}} \label{app:proof:aux:dualop}

Let $\{e^{(1)}, \dots, e^{(m)} \}$ be a standard basis vectors in $\bb{R}^{m}$. By the definition of adjoint linear operators in Hilbert spaces, for all $U \in \q H_n$ and for all $w \in \bb{R}^{m}$ it holds that
\[
\left \langle \q{A} U, w \right \rangle
= \left \langle U, \q{A}^* w \right \rangle_F.
\]
Applying this equality for basis vector $e^{(k)}, k \in [m]$, we obtain 
\[
\left \langle U, \q{A}^* e^{(j)} \right \rangle_F
= \left \langle \q{A}U, e^{(j)} \right \rangle
= \sum_{j=1}^{m} \tr(U (a_j a_j^{*})^{*} ) (e^{(k)})_j
= \tr(U (a_k a_k^{*})^{*} )
= \left \langle U, a_k a_k^{*} \right \rangle_F.
\]
Since it holds for all $U \in \q H_n$, we have that the action of $\q A^*$ on basis vectors is $\q{A}^* e^{(k)} = a_k a_k^{*}$. Therefore, by the linearity of $\q{A}^*$ we derive that
\[
\q{A}^* w
= \q{A}^* \left( \sum_{k=1}^{m} w_k e^{(k)} \right)
= \sum_{k=1}^{m} w_k \q{A}^* e^{(k)}
= \sum_{k=1}^{m} w_k a_k a_k^*,
\]
for all $w \in \bb{R}^m$, which concludes the proof.

%\bibliographystyle{IEEEtran}
%\bibliography{RankOneQuotientProperty}
\bibliography{ms.bbl}

\end{document}